\def\showtableofcontents{1}
\newif\ifdraft
\draftfalse

\newif\ifnotes
\notesfalse

\documentclass[11pt,letterpaper]{article}

\ifdraft
\ifnotes
  \usepackage[draft,notes=true,later=true]{dtrt}
\else
  \usepackage[draft,notes=false,later=false]{dtrt}
\fi

\else

\ifnotes
  \usepackage[notes=true,later=true]{dtrt}
\else
  \usepackage[notes=false,later=false]{dtrt}
\fi

\fi 

\usepackage[T1]{fontenc} %
\usepackage{lmodern} %
\usepackage[utf8]{inputenc} %
\usepackage{cmap} %
\usepackage{multirow}
\usepackage{float}
\usepackage{amsmath,amsthm,amssymb,algpseudocode,algorithm,cryptocode}
\usepackage{thmtools}
\usepackage{bm} %
\usepackage{graphicx}
\usepackage{fullpage}
\usepackage{caption}
\usepackage{tikz}
\usetikzlibrary{angles,quotes,quantikz}
\usepackage{physics}
\usepackage{mathtools}
\usepackage[capitalize,nameinlink]{cleveref}
  \crefname{step}{Step}{Steps}
  \crefname{claim}{Claim}{Claims}
  \crefname{appendix}{Appendix}{Appendices}

\makeatletter
\g@addto@macro\appendix{%
  \crefalias{section}{appendix}%
}
\makeatother

\hypersetup{colorlinks={true},linkcolor={blue},urlcolor={red},citecolor={blue}}
\usepackage{xspace} %
\usepackage{paralist}
\usepackage{enumitem}
\usepackage{mdframed}
\usepackage{bbm}
\usepackage{xcolor}
\allowdisplaybreaks

\newtheorem{theorem}{Theorem}[section]
\newtheorem{proposition}[theorem]{Proposition}

\newtheorem{lemma}[theorem]{Lemma}

\newtheorem{corollary}[theorem]{Corollary}

\newtheorem{remark}[theorem]{Remark}
\newtheorem{question}[theorem]{Question}

\theoremstyle{definition}
\newtheorem{definition}[theorem]{Definition}

\theoremstyle{remark}



\def\poly{{\rm poly}}

\newcommand{\Win}{\mathrm{Win}}

\let\oldket\ket
\renewcommand{\ket}[1]{\oldket*{#1}}
\let\oldketbra\ketbra
\renewcommand{\ketbra}[1]{\oldketbra*{#1}}

\newcommand{\cA}{\mathcal{A}}

\newcommand{\wt}{\mathrm{wt}}

\newcommand{\Id}{\mathbf{I}}

\renewcommand{\epsilon}{\varepsilon}

\newcommand{\Var}{\mathsf{Var}}

\newcommand{\stminus}{\frac{\ket{0}-\ket{1}}{\sqrt{2}}}
\newcommand{\stplusi}{\frac{\ket{0}+i\ket{1}}{\sqrt{2}}}
\newcommand{\stminusi}{\frac{\ket{0}-i\ket{1}}{\sqrt{2}}}

\newcommand{\EPR}{\mathrm{EPR}}

\newcommand{\Fid}{\mathrm{F}}

\crefname{equation}{}{}

\usepackage{titling}

\title{Explicit Separations for One-Query Unitary Synthesis}
\author{}
\date{July 28, 2026}
\author{Fangqi Dong \\ Princeton \and Alex Lombardi \\ Princeton \and Fermi Ma \\ NYU}

\begin{document}

\maketitle
\thispagestyle{empty} 
\begin{abstract}
    The unitary synthesis problem (Aaronson-Kuperberg, CCC 2007) asks whether every $n$-qubit unitary $U$ is computable by $\poly(n)$-size quantum circuits relative to some classical oracle $f = f_U$ depending on $U$. Recently, it was proved (Lombardi-Ma-Wright, STOC 2024) that Haar-random unitaries cannot be efficiently synthesized by algorithms that make \emph{one} query (or $\poly(n)$ parallel queries) to an arbitrary classical oracle. 

    In this work, we prove several results about the hardness (and easiness!) of different variants of unitary synthesis. Our results include the following:

    \begin{itemize}
        \item One-query vs. two-query unitary synthesis: we prove one-query lower bounds for synthesizing random permutation unitaries $P\ket{x} = \ket{\pi(x)}$, as well as random alternating-basis phase unitaries $F_2 \cdot H^{\otimes n} \cdot F_1$. This gives one-query lower bounds for ``explicit'' families of unitaries that have efficient (even two-query) unitary synthesis algorithms. 
        \item Upper bound for complex phase unitaries: we also consider complex phase unitaries $\ket{x}\mapsto \alpha_x \ket{x}$, which (similarly to permutations) have a clean two-query synthesis algorithm with no obvious one-query algorithm. However, in this case, we prove an \emph{upper bound}: there are one-query algorithms (relative to binary phase oracles) that constant-approximate these unitaries in diamond distance.
    \end{itemize}

    \noindent In order to prove our one-query lower bounds, we depart from prior work by introducing and analyzing two new cryptographic games -- the \emph{oracle state search game} and the \emph{oracle Choi state game} -- that serve as sources of hardness for unitary synthesis. As compared to the pseudorandomness-based approach of prior work, our framework is mathematically simple, more flexible in what it can prove, and more accurately captures the hardness of synthesizing unitaries that are not ``fully random.'' As a bonus, we obtain a simplification of the state-of-the-art lower bound for general-purpose unitary synthesis. 

Finally, we also use the oracle state search game to prove a new hardness-of-approximation result for quantum programs (synthesizing unitaries relative to quantum advice) for phase unitaries, giving a sharper separation between one-query unitary synthesis and quantum programs. 
     
\end{abstract}

\ifnum\showtableofcontents=1
{
\newpage
\setcounter{tocdepth}{2}
\thispagestyle{empty}
\tableofcontents
\thispagestyle{empty}
}
\newpage
\fi

\setcounter{page}{1}
\section{Introduction}

The unitary synthesis problem, introduced by Aaronson and Kuperberg in 2007 \cite{CCC:AarKup07,aaronson-notes}, asks whether every $n$-qubit unitary $U$ is computable by $\poly(n)$-size quantum circuits relative to some classical oracle $f = f_U$ depending on $U$. Informally, this question asks: does the task of implementing an arbitrary unitary efficiently reduce to that of implementing Boolean functions? A negative answer to this question, as conjectured by \cite{aaronson-notes,STOC:LomMaWri24}, would have far-reaching implications for the fields of unitary complexity theory \cite{ITCS:BEMPQY26} --- necessitating the study of unitary complexity classes with no correspondence to classical complexity or computability theory --- and quantum cryptography \cite{TQC:Kretschmer21,STOC:LomMaWri24}, raising the possibility of computational quantum cryptography that does not rely on separating any traditional complexity classes. 

Despite the question's importance, progress on resolving it has been relatively limited \cite{CCC:AarKup07,aaronson-notes,TQC:Rosenthal22,CCC:INNRY22,SODA:Rosenthal24,STOC:LomMaWri24}, with two main results to date:

\begin{itemize}
    \item \textbf{State synthesis}, or the task of constructing an arbitrary $n$-qubit quantum state $\ket{\psi}$, is \textbf{easy} relative to a ($\ket{\psi}$-dependent) classical oracle \cite{aaronson-notes,CCC:INNRY22,SODA:Rosenthal24}; in fact, there are algorithms making only a single query to a classical oracle \cite{SODA:Rosenthal24}.
    \item Full-fledged \textbf{unitary synthesis} is impossible for algorithms that make a \textbf{single query} to a classical oracle, provided that the input length of the query is $o(2^n)$ \cite{CCC:AarKup07,STOC:LomMaWri24}. 
\end{itemize}
While lower bounds against single-query algorithms may appear limited at first glance, the freedom to query an \emph{arbitrary} classical oracle makes this class of algorithms very powerful. Aside from solving state synthesis, they also easily simulate algorithms making a polynomial number of \emph{parallel} quantum queries to any classical oracle, including on entangled inputs \cite{Yue22,STOC:LomMaWri24}. 

Which families of unitaries are subject to the \cite{STOC:LomMaWri24} one-query lower bound? Roughly speaking, the answer to keep in mind is ``Haar-random unitaries,'' or at least ``unitaries with very large ($2^{\omega(n)}$) randomness complexity.''\footnote{Technically, \cite{STOC:LomMaWri24} study ``reflections about a highly random subspace $S \subset \mathbb C^{2^n}$,'' where $S = \mathrm{Span}\{\ket{\psi_1}, \hdots, \ket{\psi_K}\}$ for i.i.d. states $\ket{\psi_1}, \hdots \ket{\psi_K}$ that are either Haar-random or uniform binary phase states. The former setting also rules out synthesizing Haar-random unitaries.} These unitaries are so (apparently) hard to compute that synthesizing them could require as many as $2^{n/2}$ sequential classical oracle queries \cite{TQC:Rosenthal22}.

In contrast, in this work, we investigate the following question.

\begin{center}
    Can we prove lower bounds for synthesizing \emph{explicit} families of unitaries?
\end{center}
By ``explicit'', we mean families of unitaries that have $\poly(n)$-query synthesis algorithms, so they are explicit (in the usual sense) relative to some classical oracle.\footnote{Of course, fully explicit unitaries have $\poly(n)$-size quantum circuits, i.e., trivial $0$-query synthesis algorithms.} Proving such lower bounds would separate efficient unitary synthesis from $1$-query unitary synthesis, and therefore demonstrate the power of \emph{adaptivity} in unitary synthesis algorithms. 

More speculatively, new one-query lower bounds might be useful for resolving the full unitary synthesis question. The intuition is as follows. The \cite{STOC:LomMaWri24} bound says that a Haar-random unitary cannot be synthesized in one query. But can a Haar-random unitary be synthesized in, say, two queries? A natural idea is to try to invoke the \cite{STOC:LomMaWri24} one-query lower bound twice. The problem, however, is that after the algorithm has performed one query, the operation it must implement on the second query may be ``less random'' than a fully Haar-random unitary, since the algorithm has already made progress toward the target Haar-random unitary. Thus, it seems plausible that proving lower bounds against ``less-than-Haar-random'' unitaries could be a stepping stone toward an adaptive query lower bound.

\subsection{This work}
In this work, we answer this question by studying three natural explicit families of unitaries described below.

\begin{itemize}
    \item \textbf{Permutation unitaries}: for any permutation $\pi$ of $\{0,1\}^n$, we consider the unitary $P = P_\pi$ described by 
    \[ P\ket{x} = \ket{\pi(x)}
    \]
    that applies the permutation $\pi$ to the standard basis.
    \item \textbf{Alternating-basis binary phase unitaries}: for boolean functions $f_1, \hdots, f_t: \{0,1\}^n \rightarrow \{0,1\}$, we study unitaries of the form
    \[ F_t \cdot H^{\otimes n} \cdot \hdots \cdot F_2 \cdot H^{\otimes n} \cdot F_1,
    \]
    where $F_i \ket{x} = (-1)^{f_i(x)} \ket{x}$ and $H^{\otimes n}$ is the $n$-qubit Hadamard transform.
    \item \textbf{Complex phase unitaries}: for any function $\alpha: \{0,1\}^n \rightarrow S^1 \subset \mathbb C$ mapping strings to unit-norm complex numbers, we consider the unitary described by
    \[ F \ket{x} = \alpha(x) \ket{x}.
    \]
\end{itemize}
It is easy to see that all of these unitary families are explicit. The second family has a trivial $t$-query synthesis algorithm, as querying a binary phase oracle is algorithmically equivalent to querying a Boolean function. The first and third families both have \emph{two-query} algorithms. For permutations, compute
\[ \ket{x}\ket{0}\overset{\pi}{\mapsto} \ket{x}\ket{\pi(x)} \overset{\pi^{-1}}{\mapsto} \ket{0} \ket{\pi(x)} \overset{\mathrm{SWAP}}{\mapsto} \ket{\pi(x)} \ket{0},
\]
where the first step queries $\pi$ (XORing the answer onto an auxiliary register initialized to $\ket{0}$) and the second step queries $\pi^{-1}$ to erase the original input. For complex phase unitaries, compute
\[ \ket{x} \ket{0} \overset{\theta_{1:n}}{\mapsto} \ket{x} \ket{\theta(x)_{1\hdots n}} \overset{\text{ctrl-}R}{\mapsto} e^{i\cdot \theta(x)_{1:n}} \ket{x} \ket{\theta(x)_{1: n}} \overset{\theta_{1:n}}{\mapsto} e^{i\cdot \theta(x)_{1: n}} \ket{x} \ket{0} \approx \alpha(x) \ket{x} \ket{0},
\]
where $\theta(x) \in [0, 2\pi)$ is the angle satisfying $e^{i \cdot \theta(x)} = \alpha(x)$ and $\theta(x)_{1:n}$ denotes its $n$-bit truncation. These are both two-query algorithms by the previously mentioned observation about simulating parallel queries.

We ask whether there are one-query synthesis algorithms for all three of these unitary families. Indeed, all three variants capture natural questions about the nature of quantum query algorithms:

\begin{itemize}
    \item For $F_2 H F_1$, this is asking whether inserting a Hadamard change of basis between two phase queries makes them ``inherently sequential.''
    \item For $P$ as well as complex phase $F$, this is asking whether the sequential process of ``compute-then-uncompute'' can be shortcut through the use of a cleverly chosen classical oracle. 
    \item Finally, for complex phase $F$, this relates to another question about the \cite{STOC:LomMaWri24} technique for one-query lower bounds: their approach necessarily rules out one-query synthesis algorithms relative to arbitrary phase oracles, and thus intrinsically cannot separate complex phase unitaries from binary phase unitaries. Can they be separated in some other way?
\end{itemize}

\subsection{Our results}
We prove several results on the synthesis of these three unitary families. For our main results, we prove one-query lower bounds for synthesizing permutation unitaries as well as unitaries of the form $F_2 H F_1$.

\begin{theorem}[informal, see \cref{thm:perm-search-main}]\label{thm:main-permutation-intro}
    There is no efficient one-query unitary synthesis algorithm for random $n$-qubit permutation unitaries. 
\end{theorem}

\begin{theorem}[informal, see \cref{thm:f2hf1h-search-main}]\label{thm:main-f2hf1-intro}
    There is no efficient one-query unitary synthesis algorithm for unitaries of the form $F_2 H F_1$ for random $F_1, F_2$. 
\end{theorem}

These results both demonstrate separations between the power of one- and two-query unitary synthesis algorithms. Moreover, \cref{thm:main-f2hf1-intro} easily extends to one-query lower bounds for unitaries $F_t H \hdots F_2 H F_1$ with more alternations.

\begin{theorem}[informal, see \cref{cor:f2hf1-t-case-search-main}]
    There is no efficient one-query unitary synthesis algorithm for unitaries of the form $F_t H \hdots F_2 H F_1$ for random $F_1, F_2, \hdots, F_t$. 
\end{theorem}

We remark that in general, it is natural to ask whether more $FH$ alternations make the unitary harder to synthesize. Most aggressively, one could ask:

\begin{question}
    Does synthesizing $F_t H \hdots F_2 H F_1$ require $t$ sequential queries?
\end{question}

We prove this for $t=2$. A positive answer to this question for all $t = \poly(n)$ would prove the unitary synthesis conjecture.

\paragraph{Extension to other interleaving unitaries.} In \cref{sec:choi-hardness}, we extend \cref{thm:main-f2hf1-intro} to the case of alternations $F_2 U_0 F_1$ for a fixed unitary $U_0$ (\cref{thm:fuf-epr}). Of course, if $U_0$ is close to the identity then such unitaries can be approximately synthesized in one query. On the other hand, we prove that if all of the entries of $U_0$ are small (for example, if $U_0$ is a tensor power of any one-qubit unitary with all four entries bounded away from $0$), such unitaries cannot be synthesized in one query. We refer the reader to \cref{sec:choi-hardness} for more details. 

\paragraph{Upper bound for phase unitaries.} On the other hand, we give a constant-factor \emph{unitary synthesis approximation algorithm} in the case of (non-Boolean) phase unitaries!

\begin{theorem}[informal, see \cref{thm:diagonal-positive-main}]
\label{thm:diagonal-positive-main-intro}
There is a $\Omega(1)$-approximate (in diamond distance, see \cref{def:unitary-synthesis}) one-query unitary synthesis algorithm for arbitrary diagonal phase unitaries.
\end{theorem}
Since existing lower bound techniques also rule out approximation algorithms, this explains why they do not apply to complex phase unitaries!

In addition, through the use of a simple composition theorem, we conclude that to some level of approximation, unitary synthesis algorithms with binary and arbitrary complex phase oracles have the same computational power. 

\begin{corollary}
Any family of unitaries with a $(3/4+\epsilon)$ correct 1-query unitary synthesis algorithm relative to the class of complex phase unitaries also has an $\Omega(\epsilon^2)$-correct 1-query unitary synthesis algorithm relative to binary phase unitaries (or Boolean functions). 
\end{corollary}
We next describe important conceptual tools used to prove our results --- the oracle state search game and its cousin, the oracle Choi state game  --- which allows us to re-state the above theorems as well as discuss two additional results. 

\subsection{The oracle state search and Choi state games}
The existing one-query unitary synthesis lower bound of \cite{STOC:LomMaWri24} can be thought of as deriving unitary synthesis lower bounds for a unitary $U$ from \emph{upper bounds} on the maximum win probability of a distinguishing task. Letting $[K]\subset [N]$ denote a subset, the task is to distinguish the following two mixed states:

\begin{enumerate}
    \item $\frac1K\sum_{k\in [K]} \ketbra{\psi_k}$ for $\ket{\psi_k} = U^\dagger \ket{k}$ and $[K]\subset [N]$ some fixed subset,
    \item the maximally mixed $n$-qubit state. 
\end{enumerate}
Just as in the unitary synthesis problem, the algorithm (or ``adversary'') is allowed to make a single function query. Thus, \cite{STOC:LomMaWri24} derive unitary synthesis lower bounds from pseudorandomness results. 

Unfortunately, this approach seems to fail (or at least run into serious difficulties) for all of the questions addressed in this paper! We refer the reader to the technical overview (\cref{sec:tech-overview-distinguishing-fails}) for more details, but the upshot is that state pseudorandomness does not seem to naturally capture the hardness of these unitary synthesis tasks. Instead, we introduce a new source of ``cryptographic hardness'' to prove our results, called the ``oracle state search game.'' 

\begin{definition}[see \cref{def:state-search-game}]
    For a collection of states $\{\ket{\psi_k}\}_{k \in [K]}$, the \emph{oracle state search game} is a challenger-adversary game in which:
    \begin{itemize}
        \item The challenger samples a classical string $k \gets [K]$ and sends $\ket{\psi_k}$ to the adversary.
        \item The adversary returns a string $k'$ to the challenger and wins if $k' = k$. 
    \end{itemize}
\end{definition}
For all of the main results/settings of this paper, the states $\ket{\psi_k}$ are mutually orthogonal, so an all-powerful adversary can in fact win this game with probability $1$. We prove our one-query unitary synthesis lower bounds by proving that for appropriate families of $\ket{\psi_k}$, one-query adversaries can only win the state search game with negligible probability. This can be seen as proving the security of a ``single-copy'' variant of a one-way state generator \cite{C:MorYam22} against one-query adversaries; the variant we consider is powerful enough to imply quantum bit commitment \cite{ITCS:BraCanQia23}, thus having similar implications for quantum cryptography as the pseudorandomness notion from \cite{STOC:LomMaWri24}. 

While most of our results are derived using the oracle state search game, we also introduce an \emph{even harder-to-win} cryptographic game whose hardness still rules out unitary synthesis: the oracle Choi state game. 

\begin{definition}[see \cref{def:epr-state-game}]
For a unitary $U$, the \emph{oracle Choi state game} is a challenger-adversary game in which:

\begin{itemize}
    \item The challenger prepares the state $\ket{\psi}_{U^\dagger} = \frac 1 {\sqrt N} \sum_{x\in [N]} U^\dagger \ket{x} \otimes \ket{x}$ and sends the first register to the adversary.
    \item The adversary performs some quantum channel and returns the same register back to the challener.
    \item To decide if the adversary wins, the challenger measures whether the two-register state is the EPR state $\frac 1 {\sqrt N} \sum_{x\in [N]} \ket{x}\otimes \ket{x}$.
\end{itemize}
\end{definition}
We consider the Choi state game to be the weakest natural formulation of average-case unitary synthesis hardness and observe (see \cref{sec:Choi-game-def}) that (1) it is at least as hard as the search game using states $\ket{\psi_k} = U\ket{k}$ and (2) its hardness still suffices to construct quantum bit commitments. While almost all of our results are proved using the search game, we prove \cref{thm:fuf-epr} using the Choi state game, provide some alternative proofs of our main results using the Choi state game in \cref{sec:choi-hardness}, and more generally believe the game to be worthy of future study. 

\subsubsection{Search game formulations of our results}
Our main one-query lower bounds follow from bounds on the probability of winning the oracle state search game.

\begin{theorem}[see \cref{thm:perm-search-main}]
\label{thm:perm-search-main-intro}
Let $P$ be the permutation unitary associated with a uniformly random permutation $\pi$ on $\{0,1\}^n$, and consider the search game for the state family $\{PH\ket{k}\}_{k\in[K]\setminus\{0\}}$. Then every one-query adversary with workspace dimension $M$ satisfies
\[
\mathbb{E}_{\pi}\bigl[\Win(\cA\mid \pi)\bigr]
=O\!\left(\frac{\log^2 M\,\log^2 K}{K}\right).
\]
\end{theorem}

\begin{theorem}[see \cref{thm:f2hf1h-search-main}]
\label{thm:f2hf1h-search-main-intro}
Let $f_1,f_2:\{0,1\}^n\to\{0,1\}$ be uniformly random Boolean functions, and let $F_j=\sum_x (-1)^{f_j(x)}\ketbra{x}$ for $j\in\{1,2\}$. For the search game on the family $\{F_2HF_1H\ket{k}\}_{k\in[K]}$, every one-query adversary with workspace dimension $M$ satisfies
\[
\mathbb{E}_{f_1,f_2}\bigl[\Win(\cA\mid f_1,f_2)\bigr]
=O\!\left(\frac{\log M \cdot \log (MN)}{K}\right).
\]
\end{theorem}

Interestingly, in the case of permutations, we show that the \emph{same} state family \textbf{fails to be pseudorandom} in the sense of \cite{STOC:LomMaWri24}, demonstrating the utility of the search game:

\begin{theorem}[see \cref{thm:perm-distinguishing-main}]\label{thm:perm-distinguishing-main-intro}

There exists a one-query adversary $\cA$ such that for every permutation unitary $P$ there is a classical oracle $f_\pi$ for which $\cA^{f_\pi}$ distinguishes the ensemble $\{PH\ket{k}\}_{k\in[2^{n/2}]}$ from Haar-random input with constant advantage.
\end{theorem}

\paragraph{New proof of \cite{STOC:LomMaWri24}.} Another consequence of our approach, which we describe in the technical overview as well as \cref{app:iid-search}, is a simple proof of the hardness of general-purpose one-query unitary synthesis as in \cite{STOC:LomMaWri24}. This is accomplished by proving the hardness of the oracle state search game for i.i.d. random binary phase states $\ket{\psi_{R, k}} = \frac 1 {\sqrt N} \sum_x R(k,x) \ket{x}$; indeed, we can prove:

\begin{theorem}[see \cref{thm:iid-binary-phase-search}]\label{thm:iid-binary-phase-search-intro}
    For i.i.d. binary phase states $\ket{\psi_{R, k}}$, any one-query adversary $\mathcal A$ with workspace dimension $M$ wins the oracle state search game with probability at most $\frac{2\log M + O(1)}{K}$.
\end{theorem}

With a little more work, this simple analysis also extends to the oracle state search game with states $\ket{\psi_k} = U\ket{k}$ defined by a Haar-random unitary (see \cref{thm:haar-basis-search}).

\paragraph{Unitary Synthesis vs. Quantum Programs.} Finally, we consider the state search game for extremely simple phase unitaries $F=\sum_x(-1)^{f(x)}\ketbra{x}$ and $\ket{\psi_k} = F H \ket{k}$ (analogous to the $t=2$ setting above), and prove its hardness for zero-query algorithms with \emph{quantum advice} about $f$. In other words, this is a quantitative separation between $1$-query unitary synthesis and (approximation by) ``quantum programs,'' or (approximately) synthesizing unitaries relative to an advice state. 

\begin{theorem}[see \cref{thm:fh-advice-main}]\label{thm:fh-advice-main-intro}
Let $f:\{0,1\}^n\to\{0,1\}$ be uniformly random. Suppose a (zero-query) non-uniform algorithm uses $S$ qubits of advice depending on $f$ and outputs $k$ from one copy of $FH\ket{k}$ with success probability $\varepsilon$, for $k \gets [K]$. Then,
$\varepsilon\le O(S/K)$.
\end{theorem}
\cref{thm:fh-advice-main-intro} is tight up to constants when $K = \Omega(N)$, as a trivial algorithm without advice wins the search game with probability $1/N$ while memorizing the $N$-size truth table of $f$ would allow for winning the search game with probability $1$. In fact, we show in \cref{sec:fh-tightness} that the lower bound is tight over the entire range of $S$.

Notably, this search game is asymptotically \emph{harder} to win than for the classical states $\ket{\psi_{x,y}} = \ket{x} \ket{y\oplus f(x)}$. This is because the ``classical'' search game has a trivial algorithm whose win probability is equal to $2^{-\ell}$, where $\ell$ is the output length of $f$, which is always much larger than $2^{-(n+\ell)}$ (one over the relevant Hilbert space dimension).

\paragraph{An open question.} With these results in mind, a natural ``frontier question'' on the boundary of our current understanding is proving a one-query lower bound against algorithms that additionally receive quantum advice (before making their query); our proofs are currently limited to handling classical advice. This question is open for any family of unitaries, including permutations, $F_2 H^{\otimes n} F_1$, and Haar-random $U$.

\subsection{Acknowledgements} We thank William Kretschmer, Gregory Rosenthal, and John Wright for many helpful discussions, and in particular for posing the questions of whether there are one-query algorithms for permutation synthesis and complex phase unitary synthesis. %

F.D., A.L., and F.M. were all supported in part by a grant from the UC Noyce Initiative to the Simons Institute for the Theory of Computing. F.D. and A.L. were supported in part by NSF CAREER award CNS-2541300 and an E. Lawrence Keyes, Jr./Emerson Electric Co. Faculty Award.

\section{Technical Overview}

We begin with a recap of the approach of \cite{STOC:LomMaWri24} and why it does not appear capable of proving \cref{thm:main-permutation-intro,thm:main-f2hf1-intro}, including a discussion of \cref{thm:perm-distinguishing-main-intro}. 

Then, we introduce the oracle state search game and prove that it is hard when the input states are i.i.d. binary phase states. As a bonus, this gives an alternative proof of the original 1-query unitary synthesis lower bound, which we believe is simpler than the proof from \cite{STOC:LomMaWri24}.%

Finally, we discuss how to extend this new approach to prove \cref{thm:main-permutation-intro,thm:main-f2hf1-intro}. We leave discussion of \cref{thm:diagonal-positive-main-intro,thm:fh-advice-main-intro} to the body of the paper. {}

\subsection{Recap of LMW}
As discussed in the introduction, \cite{STOC:LomMaWri24} prove their unitary synthesis lower bound by studying a distinguishing task for families of states $\{\ket{\psi_{R, k}}\}_{k\in [K]}$ defined relative to an oracle $R$. The task, which corresponds to the security of single-copy pseudorandom states \cite{C:JiLiuSon18}, is to distinguish
\begin{itemize}
    \item $\sum_{k\in [K]} \ketbra{\psi_{R, k}}$ from
    \item the maximally mixed state
\end{itemize}
in an attack model where the adversary can make one Boolean function oracle query. If this game is hard for some $K<N/2$, then implementing the reflection about $\mathrm{Span}\{\ket{\psi_{R,k}}\}$ must be hard. 

In \cite{STOC:LomMaWri24}, they considered the case where
\[ \ket{\psi_{R, k}} = \frac 1 {\sqrt N} \sum_{x\in [N]} R(k, x) \ket{x}
\]
is a binary phase state, where $R(\cdot, \cdot)$ assigns an independent random sign to each input (or, alternatively, each $R(k,x)$ could be an independent complex Gaussian).

To analyze the maximum win probability in this game, they modeled an arbitrary adversary as having the form $\Pi \cdot \mathcal O_f \cdot V$, where $V: \mathbb C^N \rightarrow \mathbb C^M$ is a fixed isometry, $(\Pi, \Id - \Pi)$ is a fixed binary projective measurement, and $\mathcal O_f = \mathcal O_{f_R}$ is a binary phase oracle that can depend on the choice of $R$. Then, the adversary's win probability (for a fixed $R$) is given by
\[ \underset{f: [M]\rightarrow \{0,1\}}{\max}\Big|\underset{k \gets [K]}{\mathbb E} \bra{\psi_{R, k}} V^\dagger \mathcal O_f \Pi \mathcal O_f V \ket{\psi_{R, k}} - \underset{\ket{\psi}}{\mathbb E} \bra{\psi} V^\dagger \mathcal O_f \Pi \mathcal O_f V \ket{\psi} \Big|.
\]
To analyze this optimization problem, using the intuition that the states $\ket{\psi_{R,k}}$ (for a random choice of $R$) are Haar-random, \cite{STOC:LomMaWri24} introduce a \textbf{weight vector decomposition} 
\[ V \ket{\psi_{R,k}} = D_{R, k} \cdot \ket{\wt_V},
\]
where $\ket{\wt_V}$ is a fixed ``weight vector'' whose coordinates track the expected weight when $V$ is applied to a Haar-random state, and $D_{R,k}$ is an $(R, k)$-dependent rescaling matrix whose entries are \emph{linear functions} of the variables $\{R(k,x)\}_x$. The above expression can then be upper bounded by a \emph{spectral norm}  
\[ \Big|\Big|\underset{k \gets [K]}{\mathbb E} D_{R,k}^\dagger \Pi D_{R,k} -  \underset{\mathbf R}{\mathbb E} \hspace{.1cm} D_{\mathbf R}^\dagger \Pi D_{\mathbf R}  \Big| \Big|,
\]
where $D_{\mathbf R} = D_{\mathbf R,k}$ denotes the rescaling matrix distribution for a random $\mathbf R$ (which is independent of $k$).

This matrix norm can then be bounded --- either in expectation or with high probability over $R$ --- in one of two ways, but both methods \emph{crucially} rely on the fact that the variables $R(k, \cdot)$ are \textbf{independent across different choices of $k$.}

\begin{itemize}
    \item A matrix Bernstein inequality can bound the expression using only independence across different $k$ as well as a bound on $||D_{R, k}||$ with high probability over $R$. 
    \item Sharper bounds were proved by first ``decoupling'' the $D^\dagger_{R,k}$ from the $D_{R, k}$ (even before passing to the spectral relaxation), which is only possible for very specific distributions over $R$ (such as i.i.d. Gaussian or binary phase). 
\end{itemize}

\subsection{What goes wrong for other unitaries?}\label{sec:tech-overview-distinguishing-fails}
Suppose that we now want to prove lower bounds for synthesizing some family of unitaries $U$, such as $U = P$ or $U = F_2 H^{\otimes n} F_1$. Following \cite{STOC:LomMaWri24}, the natural idea would be to describe a family of states $\ket{\psi_{U, k}}$ exhibiting pseudorandomness properties.

Unfortunately, we immediately run into an issue: for a given $U$, what family of states $\ket{\psi_{U, k}}$ should we consider? A natural choice would be $\ket{\psi_{U, k}} = U \ket{k}$ (which works for Haar-random $U$), but for both $U = P$ and $U = F_2 H^{\otimes n} F_1$ such families fail to be pseudorandom for very simple reasons:

\begin{itemize}
    \item For permutations, it is easy to distinguish $P \ket{k} = \ket{\pi(k)}$ for $k \gets [K]$ from Haar-random with just a single query to $\pi^{-1}$: on input $\ket{x}$, compute $\pi^{-1}(x)$ and check whether it lies in the range $[K]$. 
    \item For $F_2 H F_1$, the state $F_2 H F_1 \ket{k} = (-1)^{f_1(k)} \cdot F_2 H^{\otimes n} \ket{k}$ can be synthesized (and therefore recognized) with a single query to $F_2$.
\end{itemize}

So in both cases, the distinguishing game with state family $\ket{\psi_{U,k}} = U\ket{k}$ is (possibly) much easier than the full-fledged synthesis task.

On the other hand, there is a natural alternative proposal for the state family: instead define $\ket{\psi_{U, k}} = U H^{\otimes n} \ket{k}$. At first glance, this choice appears to be promising, as the trivial attacks above (for our cases of interest) no longer apply. 

Unfortunately, it is completely unclear how to analyze the spectral norm of the matrix
\[ \underset{k \gets [K]}{\mathbb E} D_{R,k}^\dagger \Pi D_{R,k} -  \underset{\mathbf R}{\mathbb E} \hspace{.1cm} D_{\mathbf R}^\dagger \Pi D_{\mathbf R}
\]
arising from the \cite{STOC:LomMaWri24} argument. Superficially, the reason for this is that the matrices $\{D_{R,k}\}_k$, whose entries describe the amplitudes of the state $V \ket{\psi_{U,k}}$, are now \emph{highly dependent} across different choices of $k$. This rules out approaches based on Bernstein's inequality, and more generally, it seems very unclear how to argue about the concentration of such a random matrix. 

\paragraph{An attack.} In fact, this uncertainty is warranted, because we have a non-trivial attack on this pseudorandomness property! Specifically, we consider the permutation case, with states
\[ P H^{\otimes n} \ket{k} = \frac 1 {\sqrt N} \sum_{x\in [N]} (-1)^{\langle k, x\rangle } \ket{\pi(x)} = \frac 1 {\sqrt N} \sum_{x\in [N]} (-1)^{\langle k, \pi^{-1}(x)\rangle } \ket{x}.
\]
We claim that it is easy, with one function query, to distinguish $P H^{\otimes n} \ket{k}$ for $k \gets [\sqrt N]$ from a maximally mixed state. Specifically, we prove this when identifying $[N] = \{0,1\}^n, [K] = \{0^{n/2}\} \times \{0,1\}^{n/2}$ via binary representation. In this case, we re-name $k\in [K]$ as $(0^{n/2}, k) \in \{0,1\}^n$ and define $\pi^{-1}(x) = (g_1(x), g_2(x)) \in \{0,1\}^{n/2} \times \{0,1\}^{n/2}$, and write
\begin{align*} P H^{\otimes n} \ket{k} &= \frac 1 {\sqrt N} \sum_{x\in \{0,1\}^n} (-1)^{\langle k, g_2(x)\rangle } \ket{x} \\
&= \frac 1 {\sqrt N} \sum_{y\in \{0,1\}^{n/2}} (-1)^{\langle k, y\rangle } \cdot \sum_{x: g_2(x) = y} \ket{x} \\
&:= \frac 1 {N^{1/4}} \sum_y (-1)^{\langle k, y\rangle} \cdot \ket{\phi_y}
\end{align*}
The idea behind the attack is (just like before) to break pseudorandomness without fully inverting the unitary $P$. In this case, we make use of Rosenthal's one-query state synthesis algorithm \cite{SODA:Rosenthal24}: for any family of states $\ket{\psi_y}$ indexed by $y$, there is a classical oracle $f(\cdot, \cdot)$ relative to which $\ket{\phi_y}$ can be synthesized by querying $f(\cdot, y)$ (possibly along with some auxiliary junk state). This means that a single query to the function $f^*(\cdot, x) = f(\cdot, g_2(x))$ allows synthesizing $\ket{\phi_{g_2(x)}}$.

Applying this algorithm (in superposition) for $\ket{\psi_y} = \ket{y} \otimes \ket{\phi_y}$ gives us our attack, mapping $P H^{\otimes n} \ket{k}$ to 
\begin{align*} \frac 1 {N^{1/4}} \sum_{y\in \{0,1\}^{n/2}} (-1)^{\langle k, y\rangle} \cdot \ket{\phi_y} \otimes \ket{y} \otimes \ket{\phi_y} \otimes \ket{\mathrm{junk}_y},
\end{align*}
which we can recognize by applying a SWAP test to the first and third registers. 

In our opinion, this suggests that understanding one-query pseudorandomness properties of states of the form $\ket{\psi_{R, k}}$ (for structured randomness $R$) is extremely subtle!

\subsection{From decision to search}
With serious obstacles and negative results for generalizing \cite{STOC:LomMaWri24} outside of the setting of ``fully random'' states $\ket{\psi_k}$, we introduce the oracle state \emph{search} game as a new method for proving unitary synthesis lower bounds. As stated in the introduction, in the oracle state search game:

\begin{itemize}
    \item The challenger generates and sends $\ket{\psi_{R, k}}$ to the adversary for a random $k \gets [K]$.
    \item The adversary outputs a string $k' \in [K]$ and wins if $k = k'$. 
\end{itemize}
For example, if we have $\ket{\psi_k} = U\ket{k}$ for some unitary $U$, then synthesizing $U^\dagger$ is at least as hard as winning this game. 

To demonstrate our methodology, we now give a simple proof of \cref{thm:iid-binary-phase-search-intro}: that this game is hard for one-query adversaries when $\ket{\psi_{R,k}} = \frac 1 {\sqrt N} \sum_x R(k,x) \ket{x}$ for i.i.d. binary phases $R(k,x)$. 

\paragraph{Proof of \cref{thm:iid-binary-phase-search-intro}.} Without loss of generality, one-query adversaries for the oracle state search game have the following form:

\begin{itemize}
    \item Apply an isometry $V: \mathbb C^N \rightarrow \mathbb C^M$.
    \item Apply a phase unitary $F = F_R$ depending on $R$. 
    \item \textbf{Perform a projective $K$-outcome measurement} $\{\Pi_k\}_{k\in [K]}$. 
\end{itemize}
With this notation, the adversary's success probability is given by
\[ \Win(\mathcal A \mid R) = \underset{F}{\max}\underset{k \gets [K]}{\mathbb E}\Big[ \Big|\Big| \Pi_k F V\ket{\psi_{R, k}} \Big| \Big|^2 \Big]. 
\]
At first glance, this may appear more unwieldy than the adversary's advantage in the distinguishing game. However, a simple observation helps us a great deal: because $\{\Pi_k\}_k$ form a projective measurement, the states $\{\Pi_k F V \ket{\psi_{R,k}}\}$ are always orthogonal, so 
\[ \Win(\mathcal A \mid R) = \underset{F}{\max}\Big|\Big| \frac 1 {\sqrt K} \sum_{k\in [K]}\Pi_k F V\ket{\psi_{R, k}} \Big| \Big|^2 . 
\]
Now, using the \cite{STOC:LomMaWri24} diagonal decomposition for the state family
\[ \ket{\psi_{R, k}} = D_{R,k} \ket{\wt_V},
\]
we can upper bound this probability by the spectral relaxation
\begin{align*} \Win(\mathcal A \mid R) &=  \underset{F}{\max}\Big|\Big| \frac 1 {\sqrt K} \sum_{k\in [K]}\Pi_k F D_{R,k} \ket{\wt_V} \Big| \Big|^2 \\
&\leq \Big| \Big| \frac 1 {\sqrt K} \sum_{k \in [K]} \Pi_k D_{R, k} \Big| \Big|^2.
\end{align*}
Thus, we wish to upper bound the value
\[ \underset{R}{\mathbb E}\Big[ \big| \big| M_R\big| \big|^2\Big],
\]
where $M_R = \frac 1 {\sqrt K} \sum_k \Pi_k D_{R,k}$ is a mean zero random matrix. The big question is, \textbf{should we expect this quantity to be small}? To start with, we can calculate the \emph{matrix variance}, an important proxy for how large we expect this quantity to be:
\[ \underset{R}{\Var}(M_R) = \max\Big( \Big|\Big|\underset{R}{\mathbb E} \hspace{.1cm} M_R^\dagger M_R \Big| \Big|, \Big|\Big|\underset{R}{\mathbb E} \hspace{.1cm}  M_R M_R^\dagger \Big| \Big| \Big).
\]
Fortunately, the random matrix $M_R$ is quite well-behaved, and we can calculate
\begin{align*} \underset{R}{\mathbb E} \hspace{.1cm} M_R^\dagger M_R  &= \frac 1 K \sum_k \underset{R}{\mathbb E} [ D_{R,k}^\dagger \Pi_k D_{R,k}] \\
&= \frac 1 K \sum_k \underset{R}{\mathbb E} [ D_{R,0}^\dagger \Pi_k D_{R,0}] \\
&= \frac 1 K \underset{R}{\mathbb E}[ D^\dagger_{R,0} \cdot \Id \cdot D_{R,0}] = \frac 1 K \cdot \Id,
\end{align*}
as $\sum_k \Pi_k = \Id$ and the rescaling terms have been defined so that they square to $1$ on average. Similarly,
\begin{align*} \underset{R}{\mathbb E} \hspace{.1cm} M_R M_R^\dagger  &= \frac 1 K \sum_{k, k'} \underset{R}{\mathbb E} [  \Pi_{k'} D_{R,k'}D_{R,k}^\dagger \Pi_k] \\
&=\frac 1 K \sum_k \underset{R}{\mathbb E} [  \Pi_k D_{R,k}D_{R,k}^\dagger \Pi_k] \\
&= \frac 1 K \sum_k \underset{R}{\mathbb E} [  \Pi_k D_{R,0}D_{R,0}^\dagger \Pi_k] \\
&= \frac 1 K \sum_k \Pi_k = \frac 1 K \cdot \Id,
\end{align*}
where we additionally make use of the fact that $\underset{R}{\mathbb E}[ D_{R, k'} D_{R, k}^\dagger] = 0$ for $k\neq k'$. Thus, the variance parameter predicts the quantity $\big|\big| M_R\big|\big|^2$ to be roughly bounded by $\frac{1}{K}$ in expectation (up to $\poly(\log \dim M_R)$ factors). This is exactly what we are looking for!

To complete the proof in the case of i.i.d. binary phase states, we simply observe that since the entries of $D_{R,k}$ are linear combinations of the $R(k,x)$, the entire matrix $M_R$ is a ``matrix Rademacher series,'' or a Rademacher combination of fixed matrices, which is well-known to exhibit concentration governed by the matrix variance parameter \cite{tropp2015book}. This proves an $O(\frac{\log M}{K})$ win probability upper bound for adversaries acting on $M$ qubits. 

\subsection{Search game hardness beyond the random case}
While the analysis from the previous section was done with i.i.d. binary phase states in mind, it turns out that two very promising parts of the analysis hold under \emph{mild assumptions} on the distribution of coefficients $R$. That is: {}

\begin{itemize}
    \item The entries of the matrix $M_R$ always have a linear dependence on the coefficients $R(k,x)$. 
    \item The matrix variance $\underset{R}{\Var}(M_R) \leq \frac 1 K$, provided that (1) in expectation over $R$, for every $k$, $\underset{R}{\mathbb E} \ketbra{\psi_{R,k}}$ is maximally mixed over $\mathbb C^N$, and (2) for every $k\neq k'$ and every pair $(x, x')$, $\underset{R}{\mathbb E}[\overline{R(k, x)} R(k', x')] = 0$. 
\end{itemize}
Note that condition (1) is only about the \emph{marginal} mixed state $\underset{R}{\mathbb E} \ketbra{\psi_{R,k}}$, and does not require any level of independence between different $R(k, \cdot)$. In fact, we can relax this condition further, so that $\underset{R}{\mathbb E} \ketbra{\psi_{R,k}}$ is only required to be maximally mixed over some subspace independent of $k$.

Of course, the matrix variance statistic is a useful heuristic but does not guarantee that $ \underset{R}{\mathbb E} \hspace{.1cm} \big|\big|M_R \big|\big|^2$ is small (let alone on the order of $1/K$). Nevertheless, we are able to argue concentration for the two most prominent (much lower randomness complexity) distributions of unitaries one can ask about.

\paragraph{Permutations.} As before, we consider the family of states
\[\ket{\psi_{R,k}} = PH^{\otimes n}\ket{k} = \frac 1 {\sqrt N} \underset{x}{\sum} (-1)^{\langle k, \pi^{-1}(x)\rangle} \ket{x}
\]
corresponding to the function $R(k,x) = (-1)^{\langle k, \pi^{-1}(x)\rangle}$ for a random permutation $\pi$. Evidently, these values are highly correlated across $k$. Nevertheless, using the linearity of $M_R$, we can write
\[ M_R = \sum_{k,x} R(k,x) \cdot B_{k,x}
\]
for some fixed, reasonably explicit and well-behaved matrices $B_{k,x}$. And while our randomness has a lot of ``cross-$k$'' dependency, the functions $R(\cdot, x)$ are still mean zero\footnote{This requires excluding $k=0^n$ from the set of keys.} and extremely close to independent!  This conveniently means that the matrix variance parameter we calculated earlier is roughly bounded by $\max\Big(\big|\big|\sum_{k,x} B^\dagger_{k,x} B_{k,x}\big|\big|, \big|\big|\sum_{k,x} B_{k,x} B^\dagger_{k,x}\big|\big|\Big) \approx \frac 1 K$, so our heuristic is still good. 

\noindent We are ultimately able to analyze $M_R$ by writing it as a ``combinatorial matrix sum'' \cite{Annals:MJCFT14}
\[ M_R = \sum_{x} A_{x,\pi^{-1}(x)},
\]
for a $2^n \times 2^n$ family of matrices
\[ A_{x,y} = \sum_k (-1)^{\langle k, y\rangle} B_{k,x}
\]
satisfying two important properties:
\begin{itemize}
    \item For a random choice of $y$, the matrix $A_{x,y}$ is zero in expectation.
    \item The individual matrices $A_{x,y}$ are ``small'' in the expected sense: they have spectral norm at most $1/\sqrt{K}$, roughly speaking because the $B_{k,x}$ are mutually orthogonal.
\end{itemize}
It turns out that this information, plus a very similar calculation to the matrix variance bound from earlier, is enough to guarantee concentration \cite{Annals:MJCFT14}, so this proves \cref{thm:perm-search-main-intro}!

\paragraph{Alternating phases.} To rule out 1-query unitary synthesis of $F_2 H F_1$, we consider the states
\begin{align*}\ket{\psi_{R, k}} &= F_2 H F_1H \ket{k} \\
&= F_2 H \cdot \frac 1 {\sqrt N} \sum_y (-1)^{f_1(y) + \langle y, k\rangle } \ket{y} \\
&= \frac 1 N \sum_{x,y} (-1)^{f_2(x) + f_1(y) + \langle y, k\rangle  + \langle y, x\rangle} \ket{x},
\end{align*}
so our coefficients have the form
\[ R(k,x) = \frac 1 {\sqrt N} (-1)^{f_2(x)} \sum_y (-1)^{f_1(y) + \langle y, k+x\rangle} := (-1)^{f_2(x)} \cdot \alpha_{k+x}, 
\]
where the coefficients $\alpha_{k+x}$ depend on $f_1$ but not $f_2$. Thus, again writing
\[ M_R = \sum_{k,x} R(k,x) \cdot B_{k,x},
\]
we observe that although the $R(k,x)$ are not all independent, this \emph{is} a Rademacher matrix sum for every fixed $f_1$. This implies that
\[  \mathbb E_{f_1, f_2} \Big| \Big|M_R\Big|\Big|^2 \lesssim \mathbb E_{f_1} \Big|\Big| \underset{f_2}{\Var}(M_R) \Big|\Big| \lesssim \mathbb E_{f_1} \Big|\Big| \mathbb E_{f_2} M_R^\dagger M_R \Big|\Big| + \mathbb E_{f_1} \Big|\Big| \mathbb E_{f_2} M_R M_R^\dagger \Big|\Big|,
\]
so we have reduced the problem to another matrix concentration problem. Finally, while it may appear that arguing the concentration of $\widetilde M_{f_1} = \mathbb E_{f_2} M_R M_R^\dagger$ may be challenging because of a quadratic dependence on $f_1$, it turns out that there is a simple \emph{rectangular square root} $M_{f_1}$ that depends linearly on $f_1$ and satisfies $M_{f_1}^\dagger M_{f_1} = \widetilde M_{f_1}$. This allows us to bound $|| \widetilde M_{f_1}|| = || M_{f_1}||^2$ via a second matrix concentration inequality. We refer the reader to \cref{sec:f2hf1h-lb} for more details.

\section{Preliminaries}
\label{sec:preliminaries}

We use $N=2^n$ for the dimension of a main $n$-qubit register and $M=2^m$ for the dimension of a potentially larger workspace. For a positive integer $K$, we write $[K]:=\{0,1,\dots,K-1\}$. %

\subsection{The unitary synthesis problem}

We recall the oracle-circuit formulation from \cite{STOC:LomMaWri24}.

\begin{definition}[Approximating a unitary, \cite{STOC:LomMaWri24}]\label[definition]{def:approx_unitary}
Let $U$ be an $n$-qubit unitary, and let $\Phi_U$ be the associated quantum channel. Let $\Phi_{\mathrm{approx}}$ be a quantum channel with $n$-qubit input and output registers. We say that $\Phi_{\mathrm{approx}}$ $\varepsilon$-approximates $U$ if
\[
D_\diamond(\Phi_{\mathrm{approx}},\Phi_U)\le \varepsilon.
\]
\end{definition}

\begin{definition}[Channel implemented by an oracle circuit, \cite{STOC:LomMaWri24}]
\label{def:channel_implemented_by_oracle_circuit}
Given a $t$-query oracle circuit $\cA^{(\cdot)}$ with an $n$-qubit input register, an $m$-qubit workspace register, intermediate unitaries $U_1,\dots,U_{t+1}$ on $m$ qubits, and a Boolean function $f:\{0,1\}^{m}\to\{\pm 1\}$, the induced $n$-qubit channel $\Phi_{\cA^f}$ acts as follows.
\begin{enumerate}[leftmargin=1.5em]
    \item On input $\ket{\psi}$, prepare
    \[
    U_{t+1} \cdot O_f \cdot U_t \cdots O_f \cdot U_2 \cdot O_f \cdot U_1(\ket{\psi}\ket{0^{m-n}}).
    \]
    \item Output the first $n$ qubits and discard the remaining $m-n$ workspace qubits.
\end{enumerate}
More generally, if $W$ is an $m$-qubit unitary, we write $\Phi_{\cA^W}$ for the channel obtained by replacing each occurrence of $O_f$ above by $W$.
\end{definition}

\begin{definition}[(\(\mathcal C_1,\mathcal C_2\))-unitary synthesis]
Let $\mathcal C_1$ be a class of $n$-qubit unitaries and let $\mathcal C_2$ be a class of $m$-qubit unitaries. A $t$-query oracle circuit $\cA^{(\cdot)}$ is an $(\varepsilon,t)$-approximate $(\mathcal C_1,\mathcal C_2)$-synthesis algorithm if, for every $U\in \mathcal C_1$, there exists $W_U\in \mathcal C_2$ such that
\[
D_\diamond(\Phi_{\cA^{W_U}},\Phi_U)\le \varepsilon.
\]
\end{definition}
\noindent The standard variant of unitary synthesis concerns the case where $\mathcal C_2$ consists of binary phase unitaries (which implement Boolean functions). 

\begin{definition}[Unitary synthesis for a class $\mathcal C$]\label{def:unitary-synthesis}
Let $\mathcal F_{\{\pm 1\}}$ denote the class of binary phase oracles, i.e., unitaries of the form $O_f$.
We say that $\cA^{(\cdot)}$ is an $(\varepsilon,t)$-approximate synthesis algorithm for $\mathcal C$ if it is an $(\varepsilon,t)$-approximate $(\mathcal C,\mathcal F_{\{\pm 1\}})$-synthesis algorithm.
\end{definition}

\noindent A slight modification of the standard variant considers $\mathcal C_2$ to be the class of \emph{all} (not necessarily binary) phase unitaries. Since phase unitaries can be implemented to arbitrary precision given two queries to a binary phase oracle, it follows that $t$-query algorithms in this model can be simulated by $2t$-query algorithms in the standard model.

We observe that $(\mathcal C_1, \mathcal C_2)$-relative unitary synthesis obeys a simple composition theorem.

\begin{proposition}[Composition of relative unitary synthesis]
\label{prop:composition-relative-synthesis}
Let $\mathcal A^{(\cdot)}$ be an $(\varepsilon_1,t_1)$-approximate $(\mathcal C_1,\mathcal C_2)$-synthesis algorithm, and let $\mathcal B^{(\cdot)}$ be an $(\varepsilon_2,t_2)$-approximate $(\mathcal C_2,\mathcal C_3)$-synthesis algorithm. Then there exists an $(\varepsilon_1+t_1\varepsilon_2, t_1t_2)$-approximate $(\mathcal C_1,\mathcal C_3)$-synthesis algorithm.
\end{proposition}

\begin{proof}
Fix $U\in \mathcal C_1$. Let $W_U\in \mathcal C_2$ witness the approximation guarantee for $\mathcal A^{(\cdot)}$, so that
\[
D_\diamond(\Phi_{\mathcal A^{W_U}},\Phi_U)\le \varepsilon_1.
\]
Let $Z_U\in \mathcal C_3$ witness the approximation guarantee for $\mathcal B^{(\cdot)}$ applied to $W_U$, so that
\[
D_\diamond(\Phi_{\mathcal B^{Z_U}},\Phi_{W_U})\le \varepsilon_2.
\]
Construct a new oracle circuit $\mathcal C^{(\cdot)}$ by replacing each of the $t_1$ query gates to $W$ inside $\mathcal A^{(\cdot)}$ by a fresh copy of $\mathcal B^{(\cdot)}$. Then $\mathcal C^{(\cdot)}$ makes $t_1t_2$ queries to its oracle.

For $j\in\{0,1,\dots,t_1\}$, let $\Gamma_j$ denote the channel obtained from $\mathcal A^{W_U}$ by replacing the first $j$ query gates to $W_U$ by $\mathcal B^{Z_U}$, while leaving the remaining $t_1-j$ query gates ideal. Thus $\Gamma_0=\Phi_{\mathcal A^{W_U}}$ and $\Gamma_{t_1}=\Phi_{\mathcal C^{Z_U}}$.
For each $j\in[t_1]$, the channels $\Gamma_{j-1}$ and $\Gamma_j$ differ only in a single query slot, so by monotonicity of diamond distance under pre- and post-composition with channels,
\[
D_\diamond(\Gamma_{j-1},\Gamma_j)\le D_\diamond(\Phi_{\mathcal B^{Z_U}},\Phi_{W_U})\le \varepsilon_2.
\]
By the triangle inequality,
\[
D_\diamond(\Phi_{\mathcal C^{Z_U}},\Phi_{\mathcal A^{W_U}})
\le \sum_{j=1}^{t_1} D_\diamond(\Gamma_{j-1},\Gamma_j)
\le t_1\varepsilon_2.
\]
Combining this with the outer approximation error gives
\[
D_\diamond(\Phi_{\mathcal C^{Z_U}},\Phi_U)
\le D_\diamond(\Phi_{\mathcal C^{Z_U}},\Phi_{\mathcal A^{W_U}})
+ D_\diamond(\Phi_{\mathcal A^{W_U}},\Phi_U)
\le t_1\varepsilon_2+\varepsilon_1.
\]
Since $U\in\mathcal C_1$ was arbitrary, the claim follows.
\end{proof}

We will also use a different but related notion of closeness between the implemented channel and the target unitary channel, based on worst-case fidelity on worst-case inputs, possibly entangled with an auxiliary register. We call this notion \emph{auxiliary-input correctness}.

\begin{definition}[Auxiliary-input correctness of a synthesis algorithm]\label{def:aux-input-correctness}
Let $\cA^{(\cdot)}$ be a universal oracle circuit with induced channel $\Phi_{\cA^{(\cdot)}}$ on the $n$-qubit input register. We say that $\cA^{(\cdot)}$ has correctness $\eta$ for synthesizing a family $\{U_n\}_n$ of $n$-qubit unitary if, for every $U$ in the family, there exists an oracle $f_U$ such that for every pure state $|\Psi\rangle_{X,\mathrm{Aux}}$ on the input register $X$ together with an arbitrary auxiliary register $\mathrm{Aux}$,
\[
F\!\left(
(\Phi_U \otimes \mathrm{Id}_{\mathrm{Aux}})(\ketbra{\Psi}),
\;
(\Phi_{\cA^{f_U}} \otimes \mathrm{Id}_{\mathrm{Aux}})(\ketbra{\Psi})
\right)\ge \eta .
\]
\end{definition}

This fidelity-based notion of correctness is equivalent to the diamond-norm approximate formulation in \cref{def:approx_unitary}, up to constant factor parameter loss, due to the following standard fact from quantum information theory \cite[Theorem~3.33]{Watrous-book}.

\begin{proposition}[Diamond distance versus aux-input correctness]
\label{prop:diamond-vs-correctness}
Let \(U\) be an \(n\)-qubit unitary, and let \(\Phi\) be an \(n\)-qubit channel.
Then,

\begin{enumerate}
    \item If \(D_\diamond(\Phi,\Phi_U)\le \varepsilon\), then \(\Phi\) has aux-input correctness at least \(1-\varepsilon\) for \(U\).

    \item If \(\Phi\) has aux-input correctness at least \(\eta\) for \(U\), then
    \[
    D_\diamond(\Phi,\Phi_U)\le \sqrt{1-\eta}.
    \]
\end{enumerate}

\end{proposition}

This means that aux-input correctness obeys a composition theorem due to \cref{prop:composition-relative-synthesis}. However, we observe that at least for the case $t_1 = t_2 = 1$, there is a tighter composition theorem without passing through \cref{prop:diamond-vs-correctness}.

\begin{proposition}[Composition of aux-input correctness]
\label{prop:composition-aux-correctness}
Let $\mathcal A^{(\cdot)}$ and $\mathcal B^{(\cdot)}$ be one-query oracle circuits, and let $\mathcal C_1,\mathcal C_2,\mathcal C_3$ be classes of unitaries. Assume that the following hold:
\begin{enumerate}[leftmargin=1.5em]
    \item For every $U\in \mathcal C_1$, there exists $W_U\in \mathcal C_2$ such that for every pure state $\ket{\Psi}_{X,\mathrm{Aux}}$,
    \[
    F\!\left(
    (\Phi_U \otimes \mathrm{Id}_{\mathrm{Aux}})(\ketbra{\Psi}),
    \;
    (\Phi_{\mathcal A^{W_U}} \otimes \mathrm{Id}_{\mathrm{Aux}})(\ketbra{\Psi})
    \right)\ge \eta_1.
    \]
    \item For every $W\in \mathcal C_2$, there exists $Z_W\in \mathcal C_3$ such that for every pure state $\ket{\Phi}_{Y,\mathrm{Aux}}$,
    \[
    F\!\left(
    (\Phi_W \otimes \mathrm{Id}_{\mathrm{Aux}})(\ketbra{\Phi}),
    \;
    (\Phi_{\mathcal B^{Z_W}} \otimes \mathrm{Id}_{\mathrm{Aux}})(\ketbra{\Phi})
    \right)\ge \eta_2.
    \]
\end{enumerate}
Then there exists a one-query oracle circuit $\mathcal D^{(\cdot)}$ such that for every $U\in \mathcal C_1$, there exists $Z_U\in \mathcal C_3$ for which, for every pure state $\ket{\Psi}_{X,\mathrm{Aux}}$,
\[
F\!\left(
(\Phi_U \otimes \mathrm{Id}_{\mathrm{Aux}})(\ketbra{\Psi}),
\;
(\Phi_{\mathcal D^{Z_U}} \otimes \mathrm{Id}_{\mathrm{Aux}})(\ketbra{\Psi})
\right)\ge \eta_\star,
\]
where
\[
\eta_\star := \cos^2\!\left(\min\!\left\{\frac{\pi}{2},\arccos\sqrt{\eta_1}+\arccos\sqrt{\eta_2}\right\}\right).
\]
\end{proposition}

\begin{proof}
Let $\mathcal D^{(\cdot)}$ be obtained by replacing the unique oracle call inside $\mathcal A^{(\cdot)}$ by $\mathcal B^{(\cdot)}$. Fix $U\in \mathcal C_1$, and choose $W_U\in \mathcal C_2$ and $Z_U\in \mathcal C_3$ as in the hypotheses. Fix any pure state $\ket{\Psi}_{X,\mathrm{Aux}}$, and define
\[
\rho := (\Phi_{\mathcal D^{Z_U}}\otimes \mathrm{Id}_{\mathrm{Aux}})(\ketbra{\Psi}),
\qquad
\sigma := (\Phi_{\mathcal A^{W_U}}\otimes \mathrm{Id}_{\mathrm{Aux}})(\ketbra{\Psi}),
\qquad
\tau := (\Phi_U\otimes \mathrm{Id}_{\mathrm{Aux}})(\ketbra{\Psi}).
\]
By the first hypothesis, $F(\sigma,\tau)\ge \eta_1$.

It remains to lower bound $F(\rho,\sigma)$. Let $\ket{\Omega}_{Q,\mathrm{Aux}}$ be the pure state of the queried register of $\mathcal A$, together with all remaining workspace registers and the auxiliary register $\mathrm{Aux}$, immediately before the unique query gate of $\mathcal A^{W_U}$ on input $\ket{\Psi}$. Replacing that ideal query $W_U$ by the one-query implementation $\mathcal B^{Z_U}$ acts on the queried register while leaving $\mathrm{Aux}$ untouched, so the second hypothesis gives fidelity at least $\eta_2$ between the corresponding post-query states. Applying the common post-query channel of $\mathcal A$ to both branches and using monotonicity of fidelity under channels, we obtain $F(\rho,\sigma)\ge \eta_2$.

Now define the Bures angle $A(\alpha,\beta):=\arccos \sqrt{F(\alpha,\beta)}$. By the triangle inequality for the Bures angle (see, e.g., \cite[Section~9.2]{Watrous-book}),
\[
A(\rho,\tau)\le A(\rho,\sigma)+A(\sigma,\tau)
\le \arccos\sqrt{\eta_2}+\arccos\sqrt{\eta_1}.
\]
Therefore,
\[
F(\rho,\tau)\ge \cos^2\!\left(\min\!\left\{\frac{\pi}{2},\arccos\sqrt{\eta_1}+\arccos\sqrt{\eta_2}\right\}\right)=\eta_\star.
\]
Since $\ket{\Psi}_{X,\mathrm{Aux}}$ was arbitrary, the claim follows.
\end{proof}

\subsection{One-query normal form}

The lower bounds in this paper all concern one-query algorithms, so we describe a normal form that will be used throughout, following \cite{STOC:LomMaWri24}.

\begin{definition}[One-query unitary synthesis algorithm]
A one-query unitary synthesis algorithm on $n$-qubit inputs is specified by
\begin{itemize}[leftmargin=1.5em]
    \item an oracle $O_f$,
    \item an isometry $V:\mathbb{C}^{N}\to\mathbb{C}^{M}$, representing the computation before the oracle query, and
    \item a unitary $U$ on $\mathbb{C}^{M}$, representing the computation after the oracle query.
\end{itemize}
On input $\ket{\psi}\in\mathbb{C}^N$, the corresponding quantum channel prepares $U\,O_f\,V\ket{\psi}$ and then outputs the designated $n$-qubit subsystem.
\end{definition}

Fixing the computational basis on the $M$-dimensional workspace, every isometry $V:\mathbb{C}^{N}\to\mathbb{C}^{M}$ can be written as
\[
V=\sum_{i\in[M]} \ketbra{i}{v_i},
\]
where the vectors $\ket{v_i}\in\mathbb{C}^N$ satisfy
\[
\sum_{i\in[M]} \ketbra{v_i}=\Id_N.
\]

\subsection{The weight vector decomposition relative to an input distribution}\label{sec:weight-vector-decomposition}

We next define the \textbf{diagonal (weight vector) decomposition} of an isometry $V$ with respect to a distribution on input states. This generalizes the diagonal decomposition of \cite{STOC:LomMaWri24} to an arbitrary input distribution (\cite{STOC:LomMaWri24} considered only maximally mixed inputs). In our setting, the relevant weight vector is attached not just to the isometry $V$, but to $V$ together with the input distribution.

\begin{lemma}[Weight vector decomposition]
\label{lem:weight-vector-decomposition}
Let $V:\mathbb{C}^N\to\mathbb{C}^M$ be an isometry, and write
\[
V=\sum_{i\in[M]} \ketbra{i}{v_i}.
\]
Let $\mu$ be a distribution on pure states in $\mathbb{C}^N$, and define
\[
p_i:=\mathbb{E}_{\psi\sim\mu}\bigl[\abs{ \braket{v_i}{\psi}}^2\bigr] = \bra{v_i} \cdot \mathbb{E}_{\psi}  \Big[ \ketbra{\psi} \Big] \cdot \ket{v_i},
\qquad
\ket{\wt_{V,\mu}}:=\sum_{i\in[M]} \sqrt{p_i}\,\ket{i}.
\]
For each pure state $\ket{\psi}$, define the diagonal matrix
\[
D_{V,\psi}^{(\mu)}
:=
\sum_{i\in[M]:\,p_i>0}
\frac{\braket{v_i}{\psi}}{\sqrt{p_i}}\ketbra{i},
\]
with diagonal entry $0$ when $p_i=0$. Then $\ket{\wt_{V,\mu}}$ is a unit vector, and for $\mu$-almost every $\ket{\psi}$ we have
\[
V\ket{\psi}=D_{V,\psi}^{(\mu)}\ket{\wt_{V,\mu}}.
\]
\end{lemma}

\begin{proof}
Since $V$ is an isometry,
\[
\sum_{i\in[M]} \abs{\braket{v_i}{\psi}}^2
=
\bra{\psi}\left(\sum_{i\in[M]}\ketbra{v_i}\right)\ket{\psi}
=
\braket{\psi}
=1
\]
for every unit vector $\ket{\psi}$. Averaging over $\psi\sim\mu$ gives $\sum_i p_i=1$, so $\ket{\wt_{V,\mu}}$ has unit norm.

For the decomposition itself, if $p_i=0$ then the nonnegative random variable $\abs{\braket{v_i}{\psi}}^2$ has expectation $0$, and hence vanishes with probability $1$. Therefore, for $\mu$-almost every $\ket{\psi}$,
\[
D_{V,\psi}^{(\mu)}\ket{\wt_{V,\mu}}
=
\sum_{i\in[M]} \frac{\braket{v_i}{\psi}}{\sqrt{p_i}}\ketbra{i}
\left(\sum_{j\in[M]}\sqrt{p_j}\ket{j}\right)
=
\sum_{i\in[M]} \braket{v_i}{\psi}\ket{i}
=
V\ket{\psi}.
\]
\end{proof}

\subsection{Useful concentration inequalities}

In this section, we state two matrix concentration inequalities that are used in the proofs of \cref{thm:f2hf1h-search-main,thm:perm-search-main}, respectively. 

\begin{theorem}[Matrix Rademacher series, {\cite[Theorem~4.1.1 and Equation~(4.1.7)]{tropp2015book}}]
\label{thm:tropp-rademacher}
Let $\{\varepsilon_k\}_k$ be independent Rademacher random variables and let $\{B_k\}_k$ be fixed complex matrices of dimension $d_1\times d_2$. Define
\[
Z:=\sum_k \varepsilon_k B_k,
\qquad
v(Z):=\max\!\left\{
\left\|\sum_k B_k B_k^*\right\|,
\left\|\sum_k B_k^* B_k\right\|
\right\}.
\]
Then for all $t\ge 0$,
\[
\Pr[\norm{Z}\ge t]\le (d_1+d_2)\exp\!\left(-\frac{t^2}{2v(Z)}\right),
\]
and moreover
\[
v(Z)\le \mathbb{E}\norm{Z}^2\le 2\,v(Z)\bigl(1+\log(d_1+d_2)\bigr).
\]
\end{theorem}

\begin{theorem}[Bernstein inequality for a combinatorial matrix sum, {\cite[Corollary~10.3]{Annals:MJCFT14}}]
\label{thm:matrix-bernstein-permutation}
Let $(A_{jk})_{j,k=1}^n$ be Hermitian $d\times d$ matrices such that
\[
\sum_{j,k=1}^n A_{jk}=0
\qquad\text{and}\qquad
\norm{A_{jk}}\le r
\quad\text{for all }(j,k).
\]
Let $\pi$ be a uniformly random permutation of $\{1,\dots,n\}$ and define
\[
X:=\sum_{j=1}^n A_{j,\pi(j)}.
\]
Then for every $t\ge 0$,
\[
\Pr\bigl[\lambda_{\max}(X)\ge t\bigr]
\le d\cdot \exp\!\left(
-\frac{t^2}{12\sigma^2+4\sqrt{2}\,rt}
\right),
\]
where
\[
\sigma^2:=\frac{1}{n}\left\|\sum_{j,k=1}^n A_{jk}^2\right\| = \frac{1}{n}\left\|\sum_{j,k=1}^n A_{jk}^\dagger A_{jk}\right\|.
\]
\end{theorem}

\section{The Oracle State Search and Choi State Games}
\label{sec:search-game}
In this section, we define two new cryptographic games that will enable us to prove unitary synthesis lower bounds. First, we describe and study the \emph{oracle state search game}.

\begin{definition}[Oracle state search game]\label{def:state-search-game}
Fix a random variable $\mathbf R$ and, for every $R$ in its support, define a family of normalized states
\[
\bigl\{\ket{\psi_{R,k}}:k\in[K]\bigr\}\subseteq \mathbb{C}^{N}.
\]
Without loss of generality, we may describe $\ket{\psi_{R, k}}$ in the computational basis with the following normalization:

\[ \ket{\psi_{R, k}} = \frac 1 {\sqrt N}\sum_{x\in [N]} R(k,x) \ket{x}
\]
for some random variables $\{\mathbf R(k,x) \in \mathbb C\}_{k\in [K], x\in [N]}$. 

In the oracle state search game, the challenger samples $R$ together with a uniformly random key $k\in[K]$, gives the adversary one copy of $\ket{\psi_{R,k}}$, and the adversary must output $k$ after making one oracle query. The oracle can depend on the variable $R$ but not on $k$. 
\end{definition}

\begin{definition}[One-query search adversary]
A one-query adversary for the search game is specified by
\begin{itemize}[leftmargin=1.5em]
    \item a pre-query isometry $V:\mathbb{C}^{N}\to\mathbb{C}^{M}$, and
    \item a projective measurement $\{\Pi_k\}_{k\in[K]}$ on the $M$-dimensional post-query space.
\end{itemize}
For a fixed oracle $O_f$, the adversary applies $V$, makes one query to $O_f$, and then measures with $\{\Pi_k\}_{k\in[K]}$. It is shown in \cite{STOC:LomMaWri24} (Corollary 3.34) that this normal form is without loss of generality, where $\log(M) \leq n + \ell + \log K$ for $\ell$ equal to the \emph{length} of the adversary's oracle query. %
\end{definition}

\begin{definition}[Adversary's winning probability]
For a fixed state family defined by $R$, the adversary's winning probability is
\[
\Win(\cA\mid R)
:=
\max_f \frac{1}{K}\sum_{k\in[K]}
\bra{\psi_{R,k}}V^\dagger O_f^\dagger \Pi_k O_f V\ket{\psi_{R,k}}.
\]
\end{definition}

There are two related notions of hardness of the oracle state search game.

\begin{definition}[$(M, \epsilon)$-hardness in expectation]\label{def:M-eps-hardness}
    We say that the one-query oracle state search game is $(M, \epsilon)$-hard in expectation over a random variable $R$ if for all $t$-query adversaries acting on a Hilbert space of dimension $M$,
    \[ \mathbb E_R\Big[ \Win(\cA\mid R)\Big] \leq \epsilon.
    \]
\end{definition}

\begin{definition}[$(M, \epsilon,\delta)$-hardness]\label{def:M-eps-delta-hardness}
    We say that the one-query oracle state search game is $(M, \epsilon,\delta)$-hard over a random variable $R$ if for all $t$-query adversaries acting on a Hilbert space of dimension $M$,
    \[ \mathbb \Pr_R\Big[ \Win(\cA\mid R) > \epsilon\Big] \leq \delta.
    \]
\end{definition}

Of particular interest to us is the case where the states $\{\ket{\psi_{R, k}}\}$ are orthogonal, meaning that $\ket{\psi_{R,k}} = U_R \ket{k}$ for some unitary $U_R$ depending on $R$. In this case, we observe in \cref{app:haar-hardest} that the Haar-random distribution on $U_R$ is ``the hardest instance'' of the oracle search game: if \emph{any} distribution on $U_R$ is $(M, \epsilon)$-hard (respectively, $(M, \epsilon, \delta)$-hard), then so is the Haar distribution.

\subsection{Relationship to Unitary Synthesis 
}

In this subsection, we assume that for every $R$, the states $\bigl\{\ket{\psi_{R,k}}\bigr\}_{k\in[K]}$ are mutually orthogonal. Similar implications hold in relaxed settings where the states are only approximately orthogonal, but the orthogonal case is all that we will need in this paper.

We first formally state the fact that hardness of the oracle state search game implies the hardness of unitary synthesis. 

\begin{lemma}[Search hardness implies synthesis hardness]
\label{prop:search-implies-synthesis-hardness}
For any given $R$, let $U_R$ be any unitary satisfying
\[
U_R\ket{\psi_{R,k}}=\ket{k}
\qquad\text{for all }k\in[K].
\]
If there exists a $t$-query oracle circuit $\cA^{(\cdot)}$ that synthesizes the family $\{U_R\}_R$ with aux-input correctness $\eta$, then there exists a $t$-query search adversary whose winning probability in the oracle state search game is at least $\eta$ for every $R$.
In particular, if every $t$-query search adversary has expected winning probability at most $\delta$, then no $t$-query oracle circuit can synthesize the family $\{U_R\}_R$ with aux-input correctness greater than $\delta$.
\end{lemma}

\begin{proof}
Fix a choice of $R$. Let $\cA^{(\cdot)}$ be a $t$-query synthesis algorithm for $U_R$, and let $f_R$ be an oracle witnessing aux-input correctness $\eta$ for this unitary. Write the corresponding circuit using the normal form from \cref{def:channel_implemented_by_oracle_circuit}: it makes queries to oracle $O_{f_R}$, interleaved with fixed unitaries $(U_1,...,U_{t+1})$.

Partition the computational basis of the designated $n$-qubit output register into $K$ disjoint sets $S_k$ so that $k\in S_k$ for every $k\in[K]$; for instance, one may take $S_k=\{k\}$ for $k\neq 0$ and $S_0=\{0\}\cup([N]\setminus[K])$. Define projectors
\[
\Pi_k:=U_{t+1}^\dagger\left(\sum_{j\in S_k}\ketbra{j}\otimes \Id\right)U_{t+1}
\qquad\text{for }k\in[K].
\]
Then $\{\Pi_k\}_{k\in[K]}$ is a projective measurement on the full workspace, and hence defines a $t$-query search adversary: run $\cA^{f_R}$, except that just after the $t$-th query, apply measurement $\{\Pi_k\}_{k\in [K]}$ (instead of post-query unitary $U_{t+1}$).

Now fix any key $k\in[K]$ and feed the search adversary the challenge state $\ket{\psi_{R,k}}$. By correctness, the reduced output state has fidelity at least $\eta$ with the pure state $\ket{k}$. Since fidelity against a pure state equals the corresponding overlap, the probability that the output register lands in the set $S_k$ is at least $\eta$. Therefore measuring the full workspace with $\{\Pi_k\}_{k\in[K]}$ outputs $k$ with probability at least $\eta$. Averaging over the uniformly random key gives winning probability at least $\eta$ for this fixed $R$.
\end{proof}

Moreover, we observe in \cref{sec:relationship-to-quantum-crypto}, the hardness of the oracle state search game also implies non-trivial forms of quantum cryptography.

\subsection{The Oracle Choi State Game}\label{sec:Choi-game-def}

In this section, we introduce what we consider to be the weakest natural formulation of average-case hardness of unitary synthesis, which we call the oracle Choi state game. 

\begin{definition}[Oracle Choi state game]\label{def:epr-state-game}
Fix a random variable $\mathbf R$ and a family of unitary $\{U_R\}_{R\in\mathbf{R}}$ defined by $\mathbf{R}$.

In the oracle Choi state game, the challenger samples $R$ and prepares the following Choi state of unitary $U_R$ on register $H$ and $A$:
\[
\ket{\Psi_R}=\frac{1}{\sqrt{N}}\sum_{k\in\{0,1\}^n}\ket{k}_{H}\otimes U_R\ket{k}_{A}.
\]
Then it sends register $A$ to the adversary, keeping register $H$ hidden. The adversary will perform some computation on register $A$, by making queries to an oracle that might depend arbitrarily on $R$. After that, the challenger will apply a projective measurement $\{\ketbra{\Psi_\EPR}, I-\ketbra{\Psi_\EPR}\}$ on register $H$ and $A$, for $\ket{\Psi_{\EPR}}:=\frac{1}{\sqrt{N}}\sum_{k\in\{0,1\}^n}\ket{k}_H\otimes \ket{k}_A$. The adversary wins the game if and only if the measurement outcome is accepting. 
\end{definition}
Similarly to the case of the oracle state search game, we also describe a canonical form for one-query adversaries in the Choi state game.

\begin{definition}[One-query Choi adversary]\label{def:one-query-epr-adversary}
A one-query adversary for the oracle Choi state game is specified by 
\begin{itemize}
    \item a pre-query isometry $V:\mathbb{C}^{N}\to\mathbb{C}^{M}$, and
    \item a post-query unitary $U:\mathbb{C}^{M}\to\mathbb{C}^{M}$ on the $M$-dimensional post-query space.
\end{itemize}
For a fixed oracle $O_f$, the adversary applies $V$ (which maps register $A$ to a larger register $AA'$), makes one query to $O_f$, and then applies $U$.  
It is shown in \cite{STOC:LomMaWri24} (Corollary 3.34) that this normal form is without loss of generality, where $\log(M) \leq n + \ell + 1$ for $\ell$ as the \emph{length} of the adversary's oracle query.
\end{definition}

\begin{definition}[Adversary's winning probability]
For a fixed unitary $U_R$ defined by $R$, the adversary's winning probability in the oracle Choi state game is 
\[
    \Win(\cA\mid R)
    :=
    \max_f \left\|
    \Pi_{\mathrm{EPR}}\frac{1}{\sqrt{N}}
    \left(\sum_{k\in \{0,1\}^n}\ket{k}\otimes U\cdot O_f\cdot V\cdot U_R\ket{k}
    \right)
    \right\|^2
\]
where $\Pi_\EPR:=\ketbra{\Psi_{\EPR}}\otimes I_{A'}$.
\end{definition}

\paragraph{Relationship to Unitary Synthesis.} We observe that if the oracle Choi state game is hard for some class of adversaries, then worst-case unitary synthesis is hard for the same class of adversaries. 

\begin{lemma}[Choi state game hardness implies synthesis hardness]
If there exists an oracle circuit $\cA^{(\cdot)}$ that synthesizes the family $\{U_R^\dagger\}_R$ with aux-input correctness $\eta$ within $t$ queries, then there exists a $t$-query Choi adversary whose winning probability in the Choi state game is at least $\eta$ for every $R$.

In particular, if every $t$-query Choi adversary has expected winning probability at most $\delta$, then no $t$-query oracle circuit can synthesize the family $\{U_R^\dagger\}_R$ with aux-input correctness greater than $\delta$.
\end{lemma}
\begin{proof}
    Fix a choice of $R$. Let $\cA^{(\cdot)}$ be a $t$-query synthesis algorithm for $U_R^\dagger$, and let $f_R$ be an oracle witnessing aux-input correctness $\eta$ for this unitary. Write the corresponding circuit using the normal form from \Cref{def:channel_implemented_by_oracle_circuit}: it makes queries to oracle $O_{f_R}$, interleaved with fixed unitaries $(U_1,...,U_{t+1})$.

    This $(U_1,...,U_{t+1})$ and choice of $f_R$ actually define a $t$-query Choi adversary. By the aux-input correctness, the output state (together with the hidden state as the auxiliary state) has fidelity at least $\eta$ with the pure state $\ket{\Psi_\EPR}$. That is, the $t$-query Choi adversary will win with probability at least $\eta$.
\end{proof}

In fact, the oracle Choi state game has a natural interpretation as measuring
the Haar-average input correctness of a unitary-synthesis procedure.

For fixed $R$ and oracle $f$, let $\mathcal{E}_{R,f}$ denote the channel implemented
by the adversary $\mathcal A^f$ on register $A$, after tracing out any workspace
or ancilla. Since the adversary attempts to undo $U_R$ on register $A$ in the Choi state game, here we consider synthesizing 
$U_R^\dagger$. Its Haar-average correctness is naturally defined as 
\begin{align*}
    \eta_{R,f}
    :=
    \mathbb{E}_{\ket{\psi}\gets\mathrm{Haar}}
    \bra{\psi}
    \mathcal{E}_{R,f}
    \left(U_R\ketbra{\psi}U_R^\dagger\right)
    \ket{\psi}.
\end{align*}
Define $\Lambda_{R,f}(\rho)
    :=
    \mathcal{E}_{R,f}(U_R\rho U_R^\dagger)$, $\Phi_{\EPR}:=\ketbra{\Psi_{\EPR}}$,
we have
\begin{align*}
    \bra{\psi}\Lambda_{R,f}(\ketbra{\psi})\ket{\psi}
    =
    N\cdot \Tr\left[
        \left((\ketbra{\psi})^{\mathsf T}\otimes\ketbra{\psi}\right)
        (\Id\otimes\Lambda_{R,f})(\Phi_{\EPR})
    \right].
\end{align*}
Using the Haar second-moment identity
$\mathbb{E}_{\ket{\psi}\gets\mathrm{Haar}}
    [
        (\ketbra{\psi})^{\mathsf T}\otimes\ketbra{\psi}
    ]
    =
    \frac{I+N\cdot \Phi_{\EPR}}{N(N+1)}$,
we obtain
\begin{align*}
    \eta_{R,f}
    &=
    \frac{
        1+
        N\cdot \Tr\left[
            \Phi_{\EPR}
            (\Id\otimes\Lambda_{R,f})(\Phi_{\EPR})
        \right]
    }{N+1}
\end{align*}
where we used
$\Tr((\Id\otimes\Lambda_{R,f})(\Phi_{\EPR}))=1$.
The trace term above is exactly the winning probability in the oracle Choi
state game with oracle $f$. Therefore, maximizing over $f$ gives
\[
    \Win(\mathcal A\mid R)
    =
    \frac{(N+1)\max_f\eta_{R,f}-1}{N}.
\]
Thus, the Choi-game winning probability is an affine rescaling of the best
Haar-average correctness for adversary $\mathcal A$ to synthesize $U_R^\dagger$.

\paragraph{Relationship to the Oracle State Search Game.}
In this part, we will show that the search hardness implies the Choi state game hardness. Note that every unitary $U_R$ naturally defines state family as $\{U_R W\ket{k} \}_{k\in \{0,1\}^n}$ for some fixed unitary $W:\mathbb{C}^N\to \mathbb{C}^N$. 

In fact, the oracle Choi state game can be viewed as a coherent search game on state family $\{U_R\ket{k}\}_k$, trying to coherently map $U_R\ket{k}$ back to $\ket{k}$ on register $A$, as mapping $\sum_{k}\ket{k}_H\otimes U_R\ket{k}_A$ back to $\sum_{k}\ket{k}_H\otimes\ket{k}_A$ in the oracle Choi state game.

\begin{lemma}[Search hardness implies Choi state game hardness]\label{lem:search-hard->choi-hard}
If there exists a $t$-query Choi adversary $\cA^{(\cdot)}$ that wins the oracle Choi state game on unitary family $\{U_R\}_R$ with winning probability $\eta$, then for every fixed unitary $W:\mathbb{C}^N\to\mathbb{C}^N$, for state family $\{\ket{\psi_{R,k}}:=U_RW\ket{k}\}$, there exists a $t$-query search adversary whose expected winning probability in the oracle state search game is at least $\eta$.

That is, if for some fixed unitary $W_0:\mathbb{C}^N\to\mathbb{C}^N$ with correspondingly defined state family $\{\ket{\psi_{R,k}}:U_RW_0\ket{k}\}$, every $t$-query search adversary has expected winning probability at most $\delta$, then no $t$-query Choi adversary can win the oracle Choi state game on unitary family $\{U_R\}_R$ with expected winning probability greater than $\delta$.
\end{lemma}
\begin{proof}
    Fix a family of $\{U_R\}_R$. Let $\cA^{(\cdot)}$ be a $t$-query Choi adversary for $\{U_R\}_R$, and let $f_R$ be an oracle witnessing correctness $\eta_R$ for this unitary with $\eta=\mathbb{E}_{R}\eta_R$. Write the corresponding circuit using the normal form from \cref{def:channel_implemented_by_oracle_circuit}: it makes queries to oracle $O_{f_R}$, interleaved with fixed unitaries $(U_1,...,U_{t+1})$.
 
    For fixed unitary $W:\mathbb{C}^N\to\mathbb{C}^N$ and state family $\{\ket{\psi_{R,k}}:=U_RW\ket{k}\}$, define projectors
    \[
    \Pi_k:=U_{t+1}^\dagger \cdot \left(W\ketbra{k}W^\dagger\otimes\Id_{A'}\right)\cdot U_{t+1}
    \qquad\text{for }k\in\{0,1\}^n
    \]
    Then $\{\Pi_k\}_{k\in\{0,1\}^n}$ is a projective measurement on the full workspace, and hence defines a $t$-query search adversary: run $\cA^{f_R}$, except that just after the $t$-th query, apply measurement $\{\Pi_k\}_{k\in [K]}$ (instead of post-query unitary $U_{t+1}$).
    
    For a random $k\gets \{0,1\}^n$, the search adversary winning probability can be written as (define $U^{f_R}_{1\to t}:=O_{f_R}\cdot U_t\cdots O_{f_R}\cdot U_1$, as the part of the algorithm just after $t$ queries):
    \begin{align*}
        &\mathbb{E}_{k\in \{0,1\}^n}\left\|
        \Pi_k \cdot U^{f_R}_{1\to t}\cdot (\ket{\psi_{R,k}} \otimes \ket{0^{m-n}})
        \right\|^2
        \\
        &=
        \left\|
        \frac{1}{\sqrt{N}}\sum_{k\in \{0,1\}^n}\ket{k}\otimes \Pi_k\cdot U^{f_R}_{1\to t}\cdot (U_RW\ket{k}\otimes \ket{0^{m-n}})
        \right\|^2
        \\
        &=
        \left\|
        \left(
        \sum_k\ketbra{k}\otimes W\ketbra{k}W^\dagger\otimes \Id_{A'}
        \right)
        \frac{1}{\sqrt{N}}\sum_{k}\ket{k}\otimes U_{t+1}\cdot U^{f_R}_{1\to t}\cdot (U_RW\ket{k}\otimes \ket{0^{m-n}})
        \right\|^2
        \\
        &\ge 
        \left\|
        \left(
        (\Id\otimes W)\ketbra{\Psi_\EPR}(\Id\otimes W^\dagger)
        \right)\otimes \Id_{A'}\cdot 
        \frac{1}{\sqrt{N}}\sum_{k}\ket{k}\otimes U_{t+1}\cdot U^{f_R}_{1\to t}\cdot (U_RW\ket{k}\otimes \ket{0^{m-n}})
        \right\|^2
        \\
        &= 
        \left\|
        \bra{\Psi_\EPR}(W^*\otimes \Id)\cdot 
        \frac{1}{\sqrt{N}}\sum_{k}\ket{k}\otimes U_{t+1}\cdot U^{f_R}_{1\to t}\cdot (U_RW\ket{k}\otimes \ket{0^{m-n}})
        \right\|^2
        \\
        &= 
        \left\|
        \bra{\Psi_\EPR}\cdot 
        \frac{1}{\sqrt{N}}\sum_{k}\ket{k}\otimes U_{t+1}\cdot O_{f_R}\cdot U_t\cdots O_{f_R}\cdot U_1\cdot (U_R\ket{k}\otimes \ket{0^{m-n}})
        \right\|^2.
    \end{align*}
By the winning definition of the oracle Choi state game, for fixed $R$ this is at least $\eta_R$. This means that for fixed $R$, the $t$-query search adversary can win with probability at least $\eta_R$. Averaging over $R$, this will give expected winning probability in the oracle state search game at least $\eta$.
\end{proof}

\subsection{Relationship to Quantum Cryptography}\label{sec:relationship-to-quantum-crypto}

We conclude this section by explaining how hardness of the oracle state
search game and the Choi state game gives rise to quantum-cryptographic primitives. In particular, their hardness will imply the security of a quantum bit commitment scheme.

\paragraph{From the Choi game.} The oracle Choi state game can also be viewed as the task of breaking the binding security of the following commitment scheme (relative to $\mathbf R$):
\begin{align*}
    \ket{\Psi_0} &:=\ket{\Psi_{\EPR}} =\frac{1}{\sqrt{N}}\sum_k\ket{k}_H\otimes \ket{k}_A,
    \\
    \ket{\Psi_1}
    &:=\ket{\Psi_R} =\frac{1}{\sqrt{N}}\sum_k\ket{k}_H\otimes U_R\ket{k}_A.
\end{align*}
To commit to bit $b$, the sender prepares $\ket{\Psi_b}$ (note that one does not need to synthesize $U_R$ to prepare $\ket{\Psi_1}$; synthesizing a state can be easier \cite{SODA:Rosenthal24}). Then it sends register $H$ to the receiver. To open the commitment, the sender announces $b$ and sends register $A$. 

This commitment scheme is perfectly hiding. For binding security (see \cite{AC:Yan22,ITCS:BraCanQia23,STOC:GJMZ23} for discussion), an adversarial sender that starts from an honestly generated commitment $\ket{\Psi_1}$, acts only on register $A$ and successfully opens it as a commitment to $0$ as $\ket{\Psi_0}$, is exactly an adversary for the oracle Choi state game. Therefore, hardness of the oracle Choi state game implies the security of the above perfectly-hiding computationally-binding quantum bit commitment scheme.

\paragraph{From the search game.} By \Cref{lem:search-hard->choi-hard}, hardness of the oracle state search game implies hardness of the corresponding oracle Choi
state game. Therefore, hardness of the oracle state search game also implies quantum bit commitment through the construction similar as above:
\begin{align*}
    \ket{\Psi_0} &:=\ket{\Psi_{\EPR}} =\frac{1}{\sqrt{N}}\sum_k\ket{k}_H\otimes \ket{k}_A,
    \\
    \ket{\Psi_1}
    &:=\frac{1}{\sqrt{N}}\sum_k\ket{k}_H\otimes \ket{\psi_{R,k}}_A.
\end{align*}

\section{Main Theorems}
In this section, we formally state (or re-state) the results that were outlined in the introduction. 

\begin{theorem}[Permutation family search bound]
\label{thm:perm-search-main}
Let $P$ be the in-place permutation unitary associated with a uniformly random permutation $\pi$ on $\{0,1\}^n$, and consider the search game for the state family $\{PH\ket{k}\}_{k\in[K]\setminus\{0\}}$. Then every one-query adversary with workspace dimension $M$ satisfies
\[
\mathbb{E}_{\pi}\bigl[\Win(\cA\mid \pi)\bigr]
=O\!\left(\frac{\log^2 M\,\log^2 K}{K}\right).
\]
\end{theorem}

\begin{theorem}[$F_2HF_1H$ search bound]
\label{thm:f2hf1h-search-main}
Let $f_1,f_2:\{0,1\}^n\to\{0,1\}$ be uniformly random Boolean functions, and let $F_j=\sum_x (-1)^{f_j(x)}\ketbra{x}$ for $j\in\{1,2\}$. For the search game on the family $\{F_2HF_1H\ket{k}\}_{k\in[K]}$, every one-query adversary with workspace dimension $M$ satisfies
\[
\mathbb{E}_{f_1,f_2}\bigl[\Win(\cA\mid f_1,f_2)\bigr]
=O\!\left(\frac{\log M \cdot \log (MN)}{K}\right).
\]
\end{theorem}

\begin{corollary}[$F_tHF_{t-1}H\cdots  F_2HF_1$ search bound]
\label{cor:f2hf1-t-case-search-main}
Let $f_1,\dots,f_t:\{0,1\}^n\to\{0,1\}$ be uniformly random Boolean functions, and let $F_j=\sum_x (-1)^{f_j(x)}\ketbra{x}$ for $j\in\{1,\dots,t\}$. For the search game on the family $\{F_tHF_{t-1}H\cdots  F_2HF_1H\ket{k}\}_{k\in[K]}$, every one-query adversary with workspace dimension $M$ satisfies
\[
\mathbb{E}_{f_1,\ldots,f_t}\bigl[\Win(\cA\mid f_1,\ldots,f_t)\bigr]
=O\!\left(\frac{\log M \cdot \log (MN)}{K}\right).
\]
\end{corollary}

\begin{theorem}[One-query distinguishing attack for a structured subset of $PH\ket{k}$]
\label{thm:perm-distinguishing-main}
There exists a one-query adversary $\cA$ such that for every in-place permutation unitary $P$ there is a classical oracle $f_\pi$ for which $\cA^{f_\pi}$ distinguishes the ensemble $\{PH\ket{k}\}_{k\in[2^{n/2}]}$ from Haar-random input with constant advantage.
\end{theorem}

\begin{theorem}[Constant-correctness synthesis for phase unitaries]
\label{thm:diagonal-positive-main}
For phase unitaries of the form $D(F)=\sum_x \alpha(x)\cdot \ketbra{x}$, there is a one-query synthesis algorithm with constant correctness (\cref{def:aux-input-correctness}). In the special case of $\alpha(x) \in \{1, i, -1, -i\}$, the achieved correctness is at least $1/2$; for the general case, the achieved correctness is at least $1/4$.
\end{theorem}

\begin{corollary}\label{cor:diagonal-synthesis}
Any family of unitaries with a $(3/4+\epsilon)$ correct 1-query unitary synthesis algorithm relative to the class of complex phase unitaries also has an $\Omega(\epsilon^2)$-correct 1-query unitary synthesis algorithm relative to binary phase unitaries (or Boolean functions). 
\end{corollary}

\begin{theorem}[Quantum-advice lower bound for $FH\ket{k}$]
\label{thm:fh-advice-main}
Let $f:\{0,1\}^n\to\{0,1\}$ be uniformly random and let $F=\sum_x(-1)^{f(x)}\ketbra{x}$. Suppose a (zero-query) non-uniform algorithm uses $S$ qubits of advice depending on $f$ and outputs $k$ from one copy of $FH\ket{k}$ with success probability $\varepsilon$. Then
\[
\varepsilon^{t}\le 2^S\left(\frac{8t}{N}\right)^{t}
\qquad\text{for every integer }t\ge 1.
\]
In particular, setting $t=S$ gives $\varepsilon\le 16S/N$.
\end{theorem}

\section{A Generic Search Reduction}
In this section, we generically reduce the problem of upper bounding $\underset{R}{\mathbb E}[\Win(\cA\mid R)]$ to the calculation of the expected (squared) spectral norm of a random matrix. We require only mild assumptions on the distribution over $R$:
\begin{itemize}
    \item For the first step, we require only that $\mathbb E_R \ketbra{\psi_{R, k}}$ is independent of $k$.
    \item For the second step, we require that $\mathbb E_R \ketbra{\psi_{R, k}}$ is maximally mixed over a linear subspace of $\mathbb C^N$. 
\end{itemize}
\subsection{Generic spectral relaxation under identical marginals}

We now utilize the weight-vector decomposition from \cref{sec:weight-vector-decomposition} to analyze the search game from \cref{sec:search-game}. 

\begin{lemma}[Generic spectral relaxation]
\label{lem:generic-spectral-reduction}
Assume that for each $k\in[K]$, the mixed state $\mathbb E_R \ketbra{\psi_{R, k}}$ is independent of $k$. Let $\cA=(V,\{\Pi_k\}_{k\in[K]})$ be a one-query adversary for the search game, and let
\[
D_{V,R,k}:=D_{V,\psi_{R,k}}^{(\mathbf R,k)}
\]
be the diagonal matrix from \cref{lem:weight-vector-decomposition}, formed using the distribution on $\ket{\psi_{\mathbf R,k}}$ for a random choice of $\mathbf R$. Define
\[
M_R:=\frac{1}{\sqrt{K}}\sum_{k\in[K]} \Pi_k D_{V,R,k}.
\]
Then, for every $R$,
\[
\Win(\cA\mid R)\le \norm{M_R}^2.
\]
Consequently,
\[
\underset{R}{\mathbb E}\bigl[\Win(\cA\mid R)\bigr]\le \underset{R}{\mathbb E}\norm{M_R}^2.
\]
\end{lemma}

\begin{proof}
Fix $R$. By \cref{lem:weight-vector-decomposition},
\[
V\ket{\psi_{R,k}} = D_{V,R,k}\ket{\wt_{V,\mathbf R,k}}
\qquad\text{for each }k\in[K].
\]
By our assumption that $\mathbb E_{\mathbf R} \ketbra{\psi_{\mathbf R, k}}$ is independent of $k$, we see that $\ket{\wt_{V,\mathbf R,k}} = \ket{\wt_{V}}$ is independent of $k$. Hence
\begin{align*}
\Win(\cA\mid R)
&=\max_f \frac{1}{K}\sum_{k\in[K]}
\bra{\wt_V} D_{V,R,k}^\dagger O_f^\dagger \Pi_k O_f D_{V,R,k}\ket{\wt_V}\\
&=\max_f \bra{\wt_V} O_f^\dagger
\left(\frac{1}{K}\sum_{k\in[K]} D_{V,R,k}^\dagger \Pi_k D_{V,R,k}\right)
O_f\ket{\wt_V}\\
&\le \left\|\frac{1}{K}\sum_{k\in[K]} D_{V,R,k}^\dagger \Pi_k D_{V,R,k}\right\|\\
&=\left\|\left(\frac{1}{\sqrt{K}}\sum_{k\in[K]} D_{V,R,k}^\dagger \Pi_k\right)
\left(\frac{1}{\sqrt{K}}\sum_{k\in[K]} \Pi_k D_{V,R,k}\right)\right\|\\
&=\norm{M_R^\dagger M_R}
=\norm{M_R}^2.
\end{align*}
Averaging over $R$ gives the final inequality.
\end{proof}

\subsection{Description of $M_R$ for subspace-uniform state families}

Our applications will rely on \cref{lem:generic-spectral-reduction} in a more concrete setting: the challenge state distribution is, in expectation, maximally mixed on a fixed subspace. In this case, there is a simple description of the random matrix $M_R$.

\begin{lemma}
\label{lem:subspace-specialization}
Assume the hypotheses of \cref{lem:generic-spectral-reduction}. In addition, suppose there is a subspace $S\subseteq \mathbb{C}^N$ of dimension $L$ with projector $\Pi_S$ such that for every $k\in [K]$,

\[
    \underset{R}{\mathbb E}\bigl[\ketbra{\psi_{R,k}}\bigr]=\frac{1}{L}\Pi_S
\]

Let $\ket{\widetilde{v}_i}:=\Pi_S \ket{v_i}$ and
\[
p_i:=\frac{\norm{\ket{\widetilde{v}_i}}^2}{L}.
\]
Then
\[
M_R=\sum_{k\in[K],\,x\in[N]} R(k,x)\,B_{k,x},
\]
where
\[
B_{k,x}:=\frac{1}{\sqrt{NK}}\sum_{i\in[M]:\,p_i>0} \frac{\braket{\widetilde{v}_i}{x}}{\sqrt{p_i}}\,\Pi_k\ketbra{i},
\]
and these matrices satisfy
\begin{align}
\sum_{k,x} B_{k,x}B_{k,x}^\dagger \preceq \frac{L}{NK}\,\Id,
\label{eq:generic-BBd}\\
\sum_{k,x} B_{k,x}^\dagger B_{k,x} \preceq \frac{L}{NK}\,\Id.
\label{eq:generic-BdB} \\
B_{k_1,x}^\dagger B_{k_2, y} = 0 \text{ for }k_1\neq k_2 \label{eq:generic-orthogonal}
\end{align}
\end{lemma}

\begin{proof}
Fix $k\in[K]$. By assumption,
\[
\underset{R}{\mathbb E}\bigl[\ketbra{\psi_{R,k}}\bigr]=\frac{1}{L}\Pi_S.
\]
Hence
\[
\underset{R}{\mathbb E}\bigl[\abs{\braket{v_i}{\psi_{R,k}}}^2\bigr]
=
\bra{v_i}\left(\frac{1}{L}\Pi_S\right)\ket{v_i}
=
\frac{\norm{\Pi_S \ket{v_i}}^2}{L}
=
\frac{\norm{\ket{\widetilde{v}_i}}^2}{L}
= p_i.
\]
Moreover, because every challenge state lies in $S$, we have
\[
\braket{v_i}{\psi_{R,k}} = \braket{\widetilde{v}_i}{\psi_{R,k}} 
\qquad\text{for all }R,k.
\]
Therefore
\begin{align*}
D_{V,R,k}
&=
\sum_{i\in[M]:\,p_i>0}
\frac{\braket{\widetilde{v}_i}{\psi_{R,k}}}{\sqrt{p_i}}\ketbra{i}\\
&=
\frac{1}{\sqrt{N}}
\sum_{i\in[M]:\,p_i>0}\sum_{x\in[N]}
R(k,x)\,\frac{\braket{\widetilde{v}_i}{x}}{\sqrt{p_i}}\ketbra{i}.
\end{align*}
Substituting this into the definition of $M_R$ yields
\begin{align*}
M_R
&=\frac{1}{\sqrt{K}}\sum_{k\in[K]} \Pi_k D_{V,R,k}\\
&=
\sum_{k\in[K],\,x\in[N]} R(k,x)
\left(
\frac{1}{\sqrt{NK}}\sum_{i\in[M]:\,p_i>0} \frac{\braket{\widetilde{v}_i}{x}}{\sqrt{p_i}}\,\Pi_k\ketbra{i}
\right)\\
&=\sum_{k,x} R(k,x) B_{k,x}.
\end{align*}

It remains to prove \cref{eq:generic-BBd} and \cref{eq:generic-BdB}. For the first bound,
\begin{align*}
\sum_{k,x} B_{k,x}B_{k,x}^\dagger
&=\frac{1}{NK}\sum_{k,x}
\Pi_k\left(\sum_{i\in[M]:\,p_i>0} \frac{\abs{\braket{\widetilde{v}_i}{x}}^2}{p_i}\ketbra{i}\right)\Pi_k\\
&=\frac{1}{NK}\sum_{k\in[K]} \Pi_k\left(\sum_{i\in[M]:\,p_i>0} \frac{\norm{\ket{\widetilde{v}_i}}^2}{p_i}\ketbra{i}\right)\Pi_k\\
&=\frac{1}{NK}\sum_{k\in[K]} \Pi_k\left(\sum_{i\in[M]:\,p_i>0} L\ketbra{i}\right)\Pi_k
\preceq \frac{L}{NK}\sum_{k\in[K]}\Pi_k
=\frac{L}{NK} \cdot \Id.
\end{align*}
For the second bound,
\begin{align*}
\sum_{k,x} B_{k,x}^\dagger B_{k,x}
&=\frac{1}{NK}\sum_{k,x}
\left(\sum_{i_1\in[M]:\,p_{i_1}>0}\frac{\braket{x}{\widetilde{v}_{i_1}}}{\sqrt{p_{i_1}}}\ketbra{i_1}\Pi_k\right)
\left(\sum_{i_2\in[M]:\,p_{i_2}>0}\frac{\braket{\widetilde{v}_{i_2}}{x}}{\sqrt{p_{i_2}}}\Pi_k\ketbra{i_2}\right)\\
&\preceq \frac{1}{NK}\sum_{k,x}\sum_{i\in[M]:\,p_i>0} \frac{\abs{\braket{\widetilde{v}_i}{x}}^2}{p_i}\ketbra{i}\\
&=\frac{1}{NK}\sum_{i\in[M]:\,p_i>0} \frac{\norm{\ket{\widetilde{v}_i}}^2}{p_i}\ketbra{i}
\preceq\frac{L}{NK}\Id.
\end{align*}
For the third identity, note that $B_{k,x} = \Pi_k B_{k,x}$ for all $k, x$ and that $\Pi_{k_1} \Pi_{k_2} = 0$ for $k_1\neq k_2$. 
\end{proof}

\cref{thm:perm-search-main,thm:f2hf1h-search-main} (as well as \cref{thm:iid-binary-phase-search,thm:haar-basis-search}) prove upper bounds on the search game win probability by invoking \cref{lem:generic-spectral-reduction,lem:subspace-specialization}, and then upper bounding $\underset{R}{\mathbb E}[\norm{M_R}^2]$.

\section{One-query lower bound for permutation unitaries}
\label{sec:permutation-lb}
In this section, we analyze the permutation state family $\{PH\ket{k}\}_k$, for the permutation unitary $P$
associated with a uniformly random permutation $\pi$, $P:\ket{x}\mapsto \ket{\pi(x)}$. 

We will first give a one-query algorithm for a distinguishing game for this state family in \cref{sec:PH-distinguishing}. This motivates our focus on the oracle state search game, with the formulation in \cref{sec:PH-search-formulation}. We then prove the one-query lower bound by analyzing the oracle state search game. The analysis proceeds by first writing the relevant random matrix as a combinatorial matrix sum over permutations in \cref{sec:PH-combinatorial-sum}, computing variance parameters for this sum in \cref{sec:PH-parameters}, and applying the matrix Bernstein inequality for combinatorial matrix sums \cite{Annals:MJCFT14} in \cref{sec:PH-bound-prob}. By invoking the appropriate matrix tail inequalities, we also prove a classical advice lower bound for non-uniform one-query algorithms in \cref{sec:PH-classical-advice}.

\subsection{A one-query distinguishing attack}
\label{sec:PH-distinguishing}
We consider the task of distinguishing a single copy of a phase state generated by an in-place permutation (applied to a fixed subspace of phase states in $\mathbb C^N$) from a Haar random state. This distinguishing game is played as follows.
\begin{enumerate}
    \item The challenger samples a permutation $\pi$ over $\{0,1\}^n$ together with a random bit $b\in\{0,1\}$.
    \item The challenger generates and sends to the adversary one copy of a state $\ket{\psi}$:
    \begin{itemize}
        \item If $b=0$, the challenger samples a uniformly random key $k\in[K]\subset [N]$, and gives the adversary one copy of $\ket{\psi}=PH\ket{k}$ (unitary $P$ is defined by $\pi$, $P:\ket{x}\mapsto\ket{\pi(x)}$).
        \item If $b=1$, the challenger samples a uniformly random $x\in[N]$ and gives the adversary one copy of $\ket{\psi}=\ket{x}$. 
    \end{itemize}
    \item The adversary is asked to output $b$ after making one oracle query, where the oracle can depend only on $\pi$.
\end{enumerate}  
This is similar to the oracle state distinguishing game in \cite[Definition~3.8]{STOC:LomMaWri24}, where the adversary wishes to distinguish a single copy of a random binary phase state $\{\ket{\psi_{R_k}}\}_k$ from Haar random.

While our eventual goal is to show that synthesizing in-place permutations is infeasible, we first prove that for states generated by in-place permutation unitaries, winning the one-query distinguishing game can be easy! This indicates that analyzing oracle state distinguishing games may be insufficient for a one-query permutation synthesis lower bound. 

\begin{theorem}[\cref{thm:perm-distinguishing-main} restated]
\label{clm:perm-distinguishing}
There exists a one-query adversary $\cA$ such that, for every in-place permutation unitary $P:\ket{x}\mapsto \ket{\pi(x)}$ on $\{0,1\}^n$, there is a classical oracle $f_\pi:\{0,1\}^{\ell(n)}\to\{0,1\}$ for some polynomially bounded $\ell(n)$ for which $\cA^{f_\pi}$ distinguishes the ensemble $\{PH\ket{k}\}_{k\in[2^{n/2}]}$ from Haar-random input with constant advantage.
\end{theorem}

\begin{proof}
Let
\[
\mathcal{K}_{\mathrm{small}}:=\{0^{n/2}z:z\in\{0,1\}^{n/2}\}\subseteq \{0,1\}^n,
\]
so that $|\mathcal{K}_{\mathrm{small}}|=2^{n/2}$. Define
\[
g_\pi(x):=\pi^{-1}(x)_{[n/2+1,n]},
\]
namely the last $n/2$ bits of $\pi^{-1}(x)$. For every $y\in\{0,1\}^{n/2}$, define the state
\[
\ket{\phi_{\pi,y}}:=2^{-n/4}\sum_{x:\,g_\pi(x)=y}\ket{x}.
\]
By the one-query state-synthesis algorithm of Rosenthal \cite[Theorem~4.1]{SODA:Rosenthal24}, there is a polynomial-size quantum circuit $C_n$ and, for each pair $(\pi,y)$, a classical oracle $f_{\pi,y}$ such that the reduced state on the first $n$ qubits of $C_n^{f_{\pi,y}}\ket{0^{\poly(n)}}$ is within trace distance $2^{-n}$ of $\ket{\phi_{\pi,y}}$.

We now combine all of these oracles into a single oracle
\[
f_\pi(x,u,z):=(-1)^{g_\pi(x)\cdot u} f_{\pi,g_\pi(x)}(z).
\]
On input $\ket{\psi}=\sum_x \alpha_x\ket{x}$, the adversary proceeds as follows.
\begin{enumerate}[leftmargin=1.5em]
    \item Append ancilla $\ket{+}^{\otimes n/2}\ket{0^{\poly(n)}}$, producing three registers:
    \[
    \sum_x \alpha_x \ket{x}_1\otimes \ket{+}^{\otimes n/2}_2\otimes \ket{0^{\poly(n)}}_3.
    \]
    \item Run $C_n$ on the third register, answering its oracle query using $f_\pi$ on the joint state. On basis states $\ket{x,u,\cdot}$, this applies the phase $(-1)^{g_\pi(x)\cdot u}$ together with the oracle $f_{\pi,g_\pi(x)}$ needed by $C_n$.
    \item Apply $H^{\otimes n/2}$ to the second register and measure it in the computational basis, obtaining some $y\in\{0,1\}^{n/2}$.
    \item Perform a swap test between the first register and the first $n$ qubits of the third register. Output ``structured'' if and only if the swap test accepts.
\end{enumerate}

Suppose first that the input is $PH\ket{k}$ for some $k\in\mathcal{K}_{\mathrm{small}}$. Because the first $n/2$ bits of $k$ vanish,
\[
PH\ket{k}
=\frac{1}{\sqrt{N}}\sum_x (-1)^{k\cdot \pi^{-1}(x)}\ket{x}
=\frac{1}{\sqrt{N}}\sum_x (-1)^{k\cdot g_\pi(x)}\ket{x}.
\]
After step~2, the joint state is proportional to
\[
\sum_x (-1)^{k\cdot g_\pi(x)}\ket{x}
\otimes \sum_u (-1)^{g_\pi(x)\cdot u}\ket{u}
\otimes C_n^{f_{\pi,g_\pi(x)}}\ket{0^{\poly(n)}}.
\]
Applying $H^{\otimes n/2}$ to the second register maps this to a superposition proportional to
\[
\sum_{y\in\{0,1\}^{n/2}} (-1)^{k\cdot y}
\ket{\phi_{\pi,y}}\otimes \ket{y} \otimes C_n^{f_{\pi,y}}\ket{0^{\poly(n)}}.
\]
Conditioned on measuring $y$, the first register is exactly $\ket{\phi_{\pi,y}}$, while the first $n$ qubits of the third register $\rho_{\pi, y}$ are within trace distance $2^{-n}$ of $\ket{\phi_{\pi,y}}$. Therefore the swap test accepts with probability at least
\[
\frac12+\frac12\bra{\phi_{\pi,y}}\rho_{\pi,y}\ket{\phi_{\pi,y}}
\ge 1-\frac 1 2 \cdot 2^{-n}.
\]

Now suppose the input is a uniformly random computational basis state $\ket{x}$. After step~3, the first two registers are $\ket{x}\ket{g_\pi(x)}$, and the first $n$ qubits of the third register are still within trace distance $2^{-n}$ of $\ket{\phi_{\pi,g_\pi(x)}}$. The overlap between $\ket{x}$ and $\ket{\phi_{\pi,g_\pi(x)}}$ is exactly $2^{-n/4}$, so
~the swap test accepts with probability at most
\[
\frac12+\frac12\bra{x}\rho_{x, g_\pi(x)}\ket{x}
\le \frac12 + \frac 1 2 \cdot 2^{-n/2}.
\]
This gives constant distinguishing advantage.
\end{proof}

\subsection{Search formulation}\label{sec:PH-search-formulation}

We now turn to the search problem. Fix a key set $\mathcal{K}\subseteq \{0,1\}^n$ of size $K$ containing $0$, and identify the search key space with $\mathcal{K}\setminus\{0\}$. Given one copy of $PH\ket{k}$ for uniformly random $k\in \mathcal{K}\setminus\{0\}$, the goal is to output $k$ using one oracle query.

For a uniformly random permutation $\pi$, define
\[
R_\pi(k,x):=(-1)^{k\cdot \pi^{-1}(x)}
\qquad\text{so that}\qquad
PH\ket{k}=\frac{1}{\sqrt{N}}\sum_x R_\pi(k,x)\ket{x}.
\]

In order to invoke \cref{lem:generic-spectral-reduction,lem:subspace-specialization}, we take advantage of one additional property of this family. Let
\[
\Pi_\perp := \Id - \ketbra{+_N},
\]
where $\ket{+_N}=N^{-1/2}\sum_x \ket{x}$. For every nonzero key $k$, the state $PH\ket{k}$ is orthogonal to $\ket{+_N}$, and its marginal over random $\pi$ is the maximally mixed state on $\ket{+_N}^\perp$:
\[
\mathbb{E}_\pi\bigl[PH\ketbra{k} H P^{-1}\bigr]=\frac{1}{N-1}\Pi_\perp.
\]
Indeed, permutation symmetry forces this density matrix to commute with every permutation matrix, hence to have the form $a\ketbra{+_N}+b\Pi_\perp$; since every $PH\ket{k}$ with $k\neq 0$ lies in $\ket{+_N}^\perp$, we have $a=0$, and the trace-one condition gives $b=1/(N-1)$.

Write $V=\sum_i \ketbra{i}{v_i}$ as usual, and define
\[
\ket{\widetilde{v}_i}:=\Pi_\perp\ket{v_i},
\qquad
p_i:=\frac{\norm{\ket{\widetilde{v}_i}}^2}{N-1}.
\]
Because the challenge states lie in $\ket{+_N}^\perp$, only the projected vectors $\widetilde{v}_i$ matter in the overlap computation. Therefore, by \cref{lem:generic-spectral-reduction,lem:subspace-specialization} with $S=\ket{+_N}^\perp$ and $L=N-1$, it is enough to bound
\[
\mathbb{E}_\pi \norm{M_\pi}^2,
\qquad
M_\pi:=\sum_{k,x} R_\pi(k,x) B_{k,x}.
\]
Here
\[
B_{k,x}=\frac{1}{\sqrt{NK}}\sum_{i\in[M]:\,p_i>0} \frac{\braket{\widetilde{v}_i}{x}}{\sqrt{p_i}}\,\Pi_k\ketbra{i},
\]
which satisfies
\[
\sum_{k,x} B_{k,x}B_{k,x}^\dagger
\preceq \frac{N-1}{NK} \cdot \Id, \qquad 
\sum_{k,x} B_{k,x}^\dagger B_{k,x}
\preceq
\frac{N-1}{NK} \cdot \Id,
\]
and
\[B_{k_1,x}^\dagger B_{k_2, y} = 0 \text{ for }k_1\neq k_2.
\]
Because $\pi$ and $\pi^{-1}$ are identically distributed, we will freely replace $\pi^{-1}$ by $\pi$ in the calculations below.

\subsection{Rewriting $M_\pi$ as a combinatorial matrix sum}\label{sec:PH-combinatorial-sum}

\paragraph{Step 1: write in terms of deterministic matrices.}
From the definition of $R_\pi$,
\[
M_\pi = \sum_{k,x} (-1)^{k\cdot \pi(x)} B_{k,x}.
\]
For $x,y\in[N]$, define
\begin{equation}
\label{eq:perm-Axy}
A_{x,y}:=\sum_{k} (-1)^{k\cdot y} B_{k,x}.
\end{equation}
Then
\[
M_\pi = \sum_x A_{x,\pi(x)}.
\]
Thus the random permutation now appears only through the combinatorial matrix sum $\sum_x A_{x,\pi(x)}$.

\paragraph{Step 2: pass to a Hermitian dilation.}
The matrices $A_{x,y}$ need not be Hermitian, so we replace them with their Hermitian dilations
\[
\widehat{B}_{k,x}:=
\begin{pmatrix}
0 & B_{k,x}\\
B_{k,x}^\dagger & 0
\end{pmatrix},
\qquad
\widehat{A}_{x,y}:=
\begin{pmatrix}
0 & A_{x,y}\\
A_{x,y}^\dagger & 0
\end{pmatrix}.
\]
Then
\begin{equation}
\label{eq:perm-hermitian-dilation}
\mathbb{E}_\pi\norm{M_\pi}^2
=
\mathbb{E}_\pi\left\|\sum_x \widehat{A}_{x,\pi(x)}\right\|^2.
\end{equation}

\subsection{Parameter estimates for the combinatorial matrix sum}\label{sec:PH-parameters}

We now verify the hypotheses of \cref{thm:matrix-bernstein-permutation} for the family $\{\widehat{A}_{x,y}\}_{x,y}$.

\paragraph{Zero total sum.}
Using \cref{eq:perm-Axy},
\begin{align*}
\sum_{x,y} A_{x,y}
&=\sum_{x,y}\sum_k (-1)^{k\cdot y} B_{k,x}
=\sum_{k,x} B_{k,x}\left(\sum_y (-1)^{k\cdot y}\right)
= N\sum_x B_{0,x}.
\end{align*}
As 0 is not in the key space, we set $\Pi_0=0$ by convention, so $B_{0,x}=0$ for every $x$, and therefore $\sum_{x,y}A_{x,y}=0$. The same holds for the Hermitian dilations.

\paragraph{Uniform norm bound.}
Since $\norm{\widehat{A}_{x,y}}=\norm{A_{x,y}}$, it suffices to bound $\norm{A_{x,y}}$. We compute
\begin{align*}
A_{x,y}^\dagger A_{x,y}
&=
\left(\sum_{k_1} (-1)^{k_1\cdot y} B_{k_1,x}^\dagger\right)
\left(\sum_{k_2} (-1)^{k_2\cdot y} B_{k_2,x}\right)\\
&=
\sum_{k_1,k_2} (-1)^{(k_1+k_2)\cdot y} B_{k_1,x}^\dagger B_{k_2,x}
=
\sum_{k} B_{k,x}^\dagger B_{k,x} \preceq \frac 1 K \cdot \Id.
\end{align*}
Thus, we obtain that 
\begin{align}
\norm{\widehat{A}_{x,y}}=\norm{A_{x,y}}
&=\sqrt{\lambda_{\max}(A_{x,y}^\dagger A_{x,y})} \leq \frac{1}{\sqrt{K}},
\label{eq:perm-R-bound}
\end{align}
so the role of $r$ in \cref{thm:matrix-bernstein-permutation} is played by $1/\sqrt{K}$.

\paragraph{Variance bound.}
The matrix variance parameter is
\[
\sigma^2
=
\frac{1}{N}\left\|\sum_{x,y}\widehat{A}_{x,y}^2\right\|
=
\frac{1}{N}
\max\!\left\{
\left\|\sum_{x,y} A_{x,y}A_{x,y}^\dagger\right\|,
\left\|\sum_{x,y} A_{x,y}^\dagger A_{x,y}\right\|
\right\}.
\]
For the first term,
\begin{align*}
\sum_{x,y} A_{x,y}A_{x,y}^\dagger
&=
\sum_{x,y}\sum_{k_1,k_2} (-1)^{(k_1+k_2)\cdot y} B_{k_1,x} B_{k_2,x}^\dagger\\
&=
N\sum_{k,x} B_{k,x}B_{k,x}^\dagger
\preceq
\frac{N-1}{K} \cdot \Id,
\end{align*}
using \cref{eq:generic-BBd} with $L=N-1$. The second term is identical and equals $N\sum_{k,x} B_{k,x}^\dagger B_{k,x} \preceq \frac{N-1}{K} \cdot \Id$ by \cref{eq:generic-BdB}. Therefore
\begin{equation}
\label{eq:perm-variance-bound}
\sigma^2=\frac{N-1}{NK}\le \frac{1}{K}.
\end{equation}

\subsection{Upper bounding the search game winning probability}\label{sec:PH-bound-prob}

Apply \cref{thm:matrix-bernstein-permutation} to the Hermitian matrix
\[
X:=\sum_x \widehat{A}_{x,\pi(x)}.
\]
Using \cref{eq:perm-R-bound} and \cref{eq:perm-variance-bound}, we obtain the 
\[
\Pr\bigl[\lambda_{\max}(X)\ge t\bigr]
\le 2M\cdot \exp\!\left(
-\frac{t^2}{12/K + 4\sqrt{2}\,t/\sqrt{K}}
\right).
\]
Choose
\[
t:= C\,\frac{\log M\,\log K}{\sqrt{K}}
\qquad\text{for a sufficiently large universal constant }C.
\]
Then
\[
\Pr\!\left[\lambda_{\max}(X)^2 \ge C^2\frac{\log^2 M\,\log^2 K}{K}\right]
\le \frac{1}{K}.
\]
Combining this with \cref{eq:perm-hermitian-dilation},
\begin{align*}
\mathbb{E}_\pi \norm{M_\pi}^2
&=
\mathbb{E}\norm{X}^2
=
\mathbb{E}\bigl[\lambda_{\max}(X)^2\bigr]\\
&\le \Pr[\lambda_{\max}(X)^2\ge t^2]\cdot 1
+ \Pr[\lambda_{\max}(X)^2\le t^2]\cdot t^2\\
&=O\!\left(\frac{\log^2 M\,\log^2 K}{K}\right).
\end{align*}
By \cref{lem:generic-spectral-reduction}, this is also an upper bound on the average search success probability, and thus proves \cref{thm:perm-search-main}.

\subsection{One-query lower bound with classical advice}\label{sec:PH-classical-advice}

The same tail bound also derives the classical-advice lower bound. For a fixed advice string, the spectral reduction from  \cref{lem:generic-spectral-reduction} suggests that constant winning probability requires $\norm{M_\pi}^2=\Omega(1)$. By setting $t=c$ for some constant $c=\Omega(1)$, this occurs with probability bounded by
\[
\Pr\bigl[\lambda_{\max}(X)\ge c\bigr]
\le 2M\cdot \exp\!\left(
-\frac{c^2}{12/K + 4\sqrt{2}\,c/\sqrt{K}}
\right)=2M\cdot \exp(-\Omega(\sqrt{K})),
\]
which is exponentially small if $\log M\ll \sqrt{K}$. A union bound over all $2^S$ advice strings then yields the lower bound
\[
S + \log(M)=\Omega(\sqrt{K})
\]
in order to achieve constant win probability.

\section{One-query lower bound for $F_2HF_1$}
\label{sec:f2hf1h-lb}
In this section, we analyze the state family $\{F_2HF_1H\ket{k}\}_k$  through the oracle state search game. 

We start with the formulation of the search game in \cref{sec:f2hf1-search-formulation}. 
The proof of lower bound also makes use of the spectral relaxation from \cref{lem:generic-spectral-reduction}. However, unlike the permutation case, the random matrix $M_R$ is no longer a combinatorial sum. Instead, we exploit the two independent sources of randomness (from $f_1$ and $f_2$) in two steps. First, we condition on $f_1$ and use the randomness of $f_2$ to invoke a matrix Rademacher series concentration inequality in \cref{sec:f2hf1-fix-f1-Rademacher}. Then, we analyze the \emph{matrix variance terms} from \cref{sec:f2hf1-fix-f1-Rademacher}, in expectation over $f_1$, by a second matrix-concentration argument in \cref{sec:fhf-MMd,sec:fhf-MdM}. We conclude our winning probability upper bound in \cref{sec:f2hf1-final-bound}. Finally, we present a classical advice lower bound for non-uniform one-query algorithms in \cref{sec:f2hf1-classical-advice}. We also extend our lower bound to the $t$-case state family $\{F_tHF_{t-1}H\cdots F_2HF_1H\ket{k}\}_k$ for $t\ge 3$, as in \cref{sec:f2hf1-t-case}.

\subsection{Search formulation}\label{sec:f2hf1-search-formulation}

Let $f_1,f_2:\{0,1\}^n\to\{0,1\}$ be uniformly random Boolean functions, and define the phase unitaries
\[
F_j:=\sum_{x\in\{0,1\}^n} (-1)^{f_j(x)}\ketbra{x}
\qquad (j\in\{1,2\}).
\]
The search problem is: given one copy of $F_2HF_1H\ket{k}$ for uniformly random $k\in[K]$, recover $k$ using one oracle query.

Write the challenge state as
\[
F_2HF_1H\ket{k}
=\frac{1}{\sqrt{N}}\sum_{x\in[N]} R(k,x)\ket{x},
\]
where
\begin{equation}
\label{eq:fhf-R-definition}
R(k,x):=(-1)^{f_2(x)}\cdot
\left(
\frac{1}{\sqrt{N}}\sum_{y\in[N]} (-1)^{f_1(y)+y\cdot (k+x)}
\right).
\end{equation}
Here and throughout this section, $k+x$ denotes addition in $\mathbb{F}_2^n$. Define
\begin{equation}
\label{eq:fhf-alpha-definition}
\alpha_u:=\frac{1}{\sqrt{N}}\sum_{y\in[N]} (-1)^{f_1(y)+y\cdot u}
\qquad (u\in\{0,1\}^n),
\end{equation}
so that
\[
R(k,x)=(-1)^{f_2(x)}\alpha_{k+x}.
\]
We observe that the marginal distribution of the challenge state $\ket{\psi_{R,k}}$ over random $(f_1,f_2)$ is independent of $k$: replacing $f_1$ by
\[
f_1^{(k)}(y):=f_1(y)\oplus (y\cdot k)
\]
transforms the $k$-th family into the $0$-th, and $(f_1^{(k)},f_2)$ has the same distribution as $(f_1,f_2)$. Moreover,
\[
\mathbb{E}_{f_1,f_2}\bigl[F_2 H F_1 H \ketbra{k} H F_1 H F_2\bigr]=\frac{1}{N}\cdot \Id
\qquad\text{for every }k,
\]
because averaging over $f_2$ kills all off-diagonal entries while each diagonal entry is $1/N$.

Therefore, \cref{lem:generic-spectral-reduction,lem:subspace-specialization} apply with $S=\mathbb{C}^N$ and $L=N$. 
~This means that
\[
\mathbb{E}_{f_1,f_2}[\Win(\cA\mid f_1,f_2)]
\le
\mathbb{E}_{f_1,f_2}\norm{M_R}^2,
\qquad
M_R:=\sum_{k,x} R(k,x) B_{k,x},
\]
where
\[
\sum_{k,x} B_{k,x}B_{k,x}^\dagger \preceq \frac 1 K \cdot \Id, \qquad 
\sum_{k,x} B_{k,x}^\dagger B_{k,x}
\preceq
\frac{1}{K}\Id,
\]
and $B^\dagger_{k_1, x} B_{k_2, y} = 0$ for all $k_1\neq k_2$. 

\subsection{Conditioning on $f_1$: a matrix Rademacher series}\label{sec:f2hf1-fix-f1-Rademacher}

For each fixed $f_1$ and each $x\in[N]$, define
\[
Z_x:=\sum_k \alpha_{k+x} B_{k,x}.
\]
Then
\begin{equation}
\label{eq:fhf-M-as-rademacher}
M_R = \sum_x (-1)^{f_2(x)} Z_x.
\end{equation}
Conditioned on $f_1$, the signs $\{(-1)^{f_2(x)}\}_x$ are independent Rademacher variables. Therefore, \cref{thm:tropp-rademacher} implies
\begin{align}
\mathbb{E}_{f_2}\norm{M_R}^2
&\le O(\log M)\cdot \Var(M_R)
\notag\\
&\le O(\log M)\cdot
\left(
\left\|\mathbb{E}_{f_2}[M_RM_R^\dagger]\right\|+
\left\|\mathbb{E}_{f_2}[M_R^\dagger M_R]\right\|
\right).
\label{ineq:fhf-first-rademacher}
\end{align}
Averaging over $f_1$ gives
\begin{equation}
\label{ineq:fhf-master-decomposition}
\mathbb{E}_{f_1,f_2}\norm{M_R}^2
=O(\log M)\cdot
\left(
\underbrace{\mathbb{E}_{f_1}\left\|\mathbb{E}_{f_2}[M_RM_R^\dagger]\right\|}_{\text{Section~\ref{sec:fhf-MMd}}}
+
\underbrace{\mathbb{E}_{f_1}\left\|\mathbb{E}_{f_2}[M_R^\dagger M_R]\right\|}_{\text{Section~\ref{sec:fhf-MdM}}}
\right).
\end{equation}
In the next section, we bound these two terms separately.

\subsection{\texorpdfstring{Bounding $\mathbb{E}_{f_1}\|\mathbb{E}_{f_2}[M_RM_R^\dagger]\|$}{Bounding E_{f_1}\|E_{f_2}[M_R M_R^\dagger]}}
\label{sec:fhf-MMd}

In terms of equation \cref{eq:fhf-M-as-rademacher}, the $f_2$-average kills all cross terms in $x$, so
\begin{equation}
\label{eq:fhf-MMd-expand}
\mathbb{E}_{f_1}\left\|\mathbb{E}_{f_2}[M_RM_R^\dagger]\right\|
=
\mathbb{E}_{f_1}\left\|\sum_x Z_x Z_x^\dagger\right\|.
\end{equation}
It is convenient to package these matrices into a single rectangular matrix. Define
\[
N_{f_1}:=\sum_x Z_x\otimes \bra{x}.
\]
Then
\begin{equation}
\label{eq:fhf-N-reduction}
\sum_x Z_x Z_x^\dagger = N_{f_1}N_{f_1}^\dagger,
\qquad\text{so}
\qquad
\mathbb{E}_{f_1}\left\|\sum_x Z_x Z_x^\dagger\right\|
=
\mathbb{E}_{f_1}\norm{N_{f_1}}^2.
\end{equation}

Next we rewrite $N_{f_1}$ as a second matrix Rademacher series. Expanding \cref{eq:fhf-alpha-definition},
\begin{align}
N_{f_1}
&=
\sum_x Z_x\otimes \bra{x}
=
\sum_{k,x} \alpha_{k+x} B_{k,x}\otimes \bra{x}
\notag\\
&=
\frac{1}{\sqrt{N}}\sum_{k,x,y} (-1)^{f_1(y)+y\cdot(k+x)} B_{k,x}\otimes \bra{x}
\notag\\
&=
\sum_y (-1)^{f_1(y)} Q_y,
\label{eq:fhf-N-as-rademacher}
\end{align}
where
\[
Q_y:=\frac{1}{\sqrt{N}}\sum_{k,x} (-1)^{y\cdot(k+x)} B_{k,x}\otimes \bra{x}.
\]
The signs $\{(-1)^{f_1(y)}\}_y$ are again independent Rademacher variables, so a second application of \cref{thm:tropp-rademacher} yields
\begin{align}
\mathbb{E}_{f_1}\norm{N_{f_1}}^2
&\le O(\log(MN))\cdot \Var(N_{f_1})
\notag\\
&\le O(\log(MN))\cdot
\left(
\left\|\mathbb{E}_{f_1}[N_{f_1}N_{f_1}^\dagger]\right\|+
\left\|\mathbb{E}_{f_1}[N_{f_1}^\dagger N_{f_1}]\right\|
\right).
\label{ineq:fhf-second-rademacher}
\end{align}

\paragraph{First variance term.}
Averaging \cref{eq:fhf-N-as-rademacher} over $f_1$ removes the cross terms in $y$, so
\begin{align*}
\left\|\mathbb{E}_{f_1}[N_{f_1}N_{f_1}^\dagger]\right\|
&=
\left\|\sum_y Q_yQ_y^\dagger\right\|\\
&=
\frac{1}{N}
\left\|
\sum_y \sum_{k_1,x_1}\sum_{k_2,x_2}
(-1)^{y\cdot(k_1+x_1)+y\cdot(k_2+x_2)}
B_{k_1,x_1} B_{k_2,x_2}^\dagger \langle x_1|x_2\rangle
\right\|\\
&=
\frac{1}{N}
\left\|
\sum_y \sum_x \sum_{k_1,k_2}
(-1)^{y\cdot(k_1+k_2)} B_{k_1,x} B_{k_2,x}^\dagger
\right\|\\
&=
\left\|\sum_{k,x} B_{k,x}B_{k,x}^\dagger\right\|
\leq 
\frac{1}{K}.
\end{align*}

\paragraph{Second variance term.}
Similarly,
\begin{align*}
\left\|\mathbb{E}_{f_1}[N_{f_1}^\dagger N_{f_1}]\right\|
&=
\left\|\sum_y Q_y^\dagger Q_y\right\|\\
&=
\frac{1}{N}
\left\|
\sum_y \sum_{k_1,x_1}\sum_{k_2,x_2}
(-1)^{y\cdot(k_1+x_1)+y\cdot(k_2+x_2)}
B_{k_1,x_1}^\dagger B_{k_2,x_2}\otimes \ket{x_1}\bra{x_2}
\right\|\\
&=
\left\|\sum_{k,x} B_{k,x}^\dagger B_{k,x}\otimes \ket{x}\bra{x}\right\|\\
&=
\max_x \left\|\sum_k B_{k,x}^\dagger B_{k,x}\right\|
\le
\left\|\sum_{k,x} B_{k,x}^\dagger B_{k,x}\right\|
\leq 
\frac{1}{K}.
\end{align*}

Combining these estimates with \cref{eq:fhf-MMd-expand}, \cref{eq:fhf-N-reduction}, and \cref{ineq:fhf-second-rademacher}, we obtain
\begin{equation}
\label{eq:fhf-MMd-final}
\mathbb{E}_{f_1}\left\|\mathbb{E}_{f_2}[M_RM_R^\dagger]\right\|
=O\!\left(\frac{\log(MN)}{K}\right).
\end{equation}

\subsection{\texorpdfstring{Bounding $\mathbb{E}_{f_1}\|\mathbb{E}_{f_2}[M_R^\dagger M_R]\|$}{Bounding E_{f_1}\|E_{f_2}[M_R^\dagger M_R]}}
\label{sec:fhf-MdM}

Again using \cref{eq:fhf-M-as-rademacher}, averaging over $f_2$ removes the cross terms in $x$ and gives
\begin{align*}
\mathbb{E}_{f_1}\left\|\mathbb{E}_{f_2}[M_R^\dagger M_R]\right\|
&=
\mathbb{E}_{f_1}\left\|\sum_x Z_x^\dagger Z_x\right\|\\
&=
\mathbb{E}_{f_1}\left\|
\sum_x \sum_{k_1,k_2} \alpha_{k_1+x}^*\alpha_{k_2+x}
B_{k_1,x}^\dagger B_{k_2,x}
\right\|\\
&=
\mathbb{E}_{f_1}\left\|
\sum_x \sum_k |\alpha_{k+x}|^2 B_{k,x}^\dagger B_{k,x}
\right\|,
\end{align*}
where the $(k_1, k_2)$ cross terms vanish by \cref{eq:generic-orthogonal}. 

Next, we note that 
\[
\sum_x \sum_k |\alpha_{k+x}|^2 B_{k,x}^\dagger B_{k,x} \preceq \max_{u\in\{0,1\}^n}|\alpha_u|^2
 \cdot \sum_{k,x}B_{k,x}^\dagger B_{k,x} \preceq \max_{u\in\{0,1\}^n}|\alpha_u|^2 \cdot \frac 1 K \cdot \Id.
\]
Therefore,
\begin{equation}
\label{eq:fhf-MdM-reduction}
\mathbb{E}_{f_1}\left\|\mathbb{E}_{f_2}[M_R^\dagger M_R]\right\|
\le
\frac{1}{K} \cdot \mathbb{E}_{f_1}\left[\max_u |\alpha_u|^2\right].
\end{equation}

For each fixed $u\in\{0,1\}^n$, the quantity $\alpha_u$ is a sum of independent mean-zero random signs of magnitude $1/\sqrt{N}$, so Hoeffding's inequality gives
\[
\Pr_{f_1}\bigl[|\alpha_u|\ge \sqrt{4\log N}\bigr]
\le \frac{2}{N^2}.
\]
A union bound over all $u\in\{0,1\}^n$ yields
\[
\Pr_{f_1}\left[\max_u |\alpha_u|^2\ge 4\log N\right]
\le \frac{2}{N}.
\]
Since $|\alpha_u|\le \sqrt{N}$, we obtain the expectation bound
\begin{equation}
\label{eq:fhf-MdM-final}
\mathbb{E}_{f_1}\left[\max_u |\alpha_u|^2\right]
=O(\log N).
\end{equation}
Combining \cref{eq:fhf-MdM-reduction} and \cref{eq:fhf-MdM-final},
\[
\mathbb{E}_{f_1}\left\|\mathbb{E}_{f_2}[M_R^\dagger M_R]\right\|
=O\!\left(\frac{\log N}{K}\right).
\]

\subsection{Final bound on the search success probability}\label{sec:f2hf1-final-bound}

Plugging \cref{eq:fhf-MMd-final} and the bound from Section~\ref{sec:fhf-MdM} into \cref{ineq:fhf-master-decomposition} gives
\begin{align*}
\mathbb{E}_{f_1,f_2}[\Win(\cA\mid f_1,f_2)]
&\le \mathbb{E}_{f_1,f_2}\norm{M_R}^2\\
&= O(\log M)\cdot
\left(
O\!\left(\frac{\log(MN)}{K}\right)
+ O\!\left(\frac{\log N}{K}\right)
\right)\\
&= O\!\left(\frac{\log M \cdot \log (M N)}{K}\right).
\end{align*}
This proves the claimed one-query search bound for the family $F_2HF_1H\ket{k}$ in \cref{thm:f2hf1h-search-main}.

\subsection{One-query lower bound with classical advice}\label{sec:f2hf1-classical-advice}
The tail bound in \cref{thm:tropp-rademacher} also implies a classical-advice lower bound. 
For a fixed advice string, the spectral reduction from \cref{lem:generic-spectral-reduction} suggests that constant winning probability  would require $\norm{M_R}^2=\Omega(1)$. 
For fixed $f_1$, by setting $t=c$ for some constant $c=\Omega(1)$, \cref{thm:tropp-rademacher} implies
\[
    \Pr_{f_2}[\|M_R\|\ge c]\le 2M\cdot \exp\!\left(
-\frac{c^2}{2\cdot \Var(M_R)}
\right).
\]
Therefore, for parameter $c'$ to be defined later,
\begin{align*}
    \Pr_{f_1,f_2}\!\left[\|M_R\|\ge c\right]
    &\le
    \Pr_{f_1}\!\left[\Var(M_R)\ge c'\right]
    +
    \Pr_{f_1,f_2}\!\left[\Var(M_R)<c' \;\land\; \|M_R\|\ge c\right]\\
    &\le \Pr_{f_1}\!\left[\Var(M_R)\ge c'\right]
    +
    2M\cdot \exp\!\left(
    -\frac{c^2}{2\cdot c'}
    \right).
\end{align*}
From the definition of $\Var(M_R)$ and \cref{eq:fhf-N-reduction}, \cref{eq:fhf-MdM-reduction},
\[
    \Var(M_R)\le \left\|\mathbb{E}_{f_2}[M_RM_R^\dagger]\right\|+
\left\|\mathbb{E}_{f_2}[M_R^\dagger M_R]\right\|
\le \|N_{f_1}\|^2 +\frac1K\max_u|\alpha_u|^2.
\]
Therefore, by applying tail bound in \cref{thm:tropp-rademacher} for $N_{f_1}$, the bound of $\Var(N_{f_1})\le 2/K$ in \cref{sec:fhf-MMd}, together with a tail bound for \cref{eq:fhf-MdM-final},
\begin{align*}
    \Pr_{f_1}\!\left[\Var(M_R)\ge c'\right]
    &\le 
    \Pr_{f_1}\!\left[\|N_{f_1}\|\ge \sqrt{\frac{c'}{2}}\right]
    +
    \Pr\!\left[\max_u|\alpha_u|^2\ge \frac{c'K}{2}\right]\\
    &\le 2MN\cdot \exp\!\left(
    -\frac{c'}{4\cdot \Var(N_{f_1})}
    \right)
    + 
    N\cdot\exp\!\left(
    -\frac{c'K}{2}
    \right)\\
    &\le 2MN\cdot \exp\!\left(
    -\frac{c'K}{2}
    \right).
\end{align*}
Therefore, by setting $c'=1/\sqrt{K}$, we can bound the probability for $c=\Omega(1)$ by 
\begin{align*}
    \Pr_{f_1,f_2}\!\left[\|M_R\|\ge c\right]
    &\le
    \Pr_{f_1}\!\left[\Var(M_R)\ge c'\right]
    +
    2M\cdot \exp\!\left(
    -\frac{c^2}{2\cdot c'}
    \right)\\
    &\le 2MN\cdot \exp\!\left(
    -\frac{c'K}{2}
    \right)
    + 
    2M\cdot \exp\!\left(
    -\frac{c^2}{2\cdot c'}
    \right)\\
    &= 2MN\cdot \exp(-\sqrt{K}/2) + 2M\cdot \exp(-c^2\sqrt{K}/2),
\end{align*}
which is exponentially small if $\log MN\ll \sqrt{K}$. A union bound over all $2^S$ advice strings then yields the lower bound 
\[
S+\log (MN)= \Omega(\sqrt{K})
\]
in order to achieve constant win probability.

\subsection{\texorpdfstring{One-query lower bound for $F_tHF_{t-1}H\cdots F_2HF_1$}{One-query lower bound for FtH...F2HF1}}
\label{sec:f2hf1-t-case}
In this subsection, we extend our one-query lower bound to the oracle state search game with state family $\{F_tHF_{t-1}H\cdots F_2HF_1H\ket{k}\}_k$. 

\begin{corollary}[\cref{cor:f2hf1-t-case-search-main} restated]\label{thm:f2hf1-t-case-search-in7}
Let $f_1,\dots,f_t:\{0,1\}^n\to\{0,1\}$ be uniformly random Boolean functions, and let $F_j=\sum_x (-1)^{f_j(x)}\ketbra{x}$ for $j\in\{1,\dots,t\}$. For the search game on the family $\{F_tHF_{t-1}H\cdots  F_2HF_1H\ket{k}\}_{k\in[K]}$, every one-query adversary with workspace dimension $M$ satisfies
\[
\mathbb{E}_{f_1,\ldots,f_t}\bigl[\Win(\cA\mid f_1,\ldots,f_t)\bigr]
=O\!\left(\frac{\log M \cdot \log (MN)}{K}\right).
\]
\end{corollary}

\begin{proof}
We prove the lower bound for the search game on $\{F_tHF_{t-1}H\cdots F_2HF_1H\ket{k}\}_k$ by reducing it to the $t=2$ case. In fact, our reduction shows that the one-query oracle state search game on $\{F_tHF_{t-1}H\cdots F_2HF_1H\ket{k}\}_k$ is at least as $(M,\epsilon)$-hard as the corresponding $t=2$ case. 

For simplicity, for any $F_1,\dots, F_t$, we define a unitary 
\[
    U_{[3:t]}:=F_tHF_{t-1}H\cdots F_3H
\]
and states
\[
    \ket{\psi_{2,k}}:=F_2HF_1H\ket{k}, 
    \qquad
    \ket{\psi_{t,k}}:=F_tHF_{t-1}H\cdots F_1H\ket{k}=U_{[3:t]}\ket{\psi_{2,k}}.
\]

From the definition of the search game in \cref{sec:search-game}, we can write the adversary's winning probability on $\ket{\psi_{t,k}}$, in expectation over $f_1, \hdots, f_t$, as
\[
    \mathbb{E}_{f_1,\dots ,f_t} \left[\Win(\cA\mid f_1,\dots, f_t)\right]
    =\mathbb{E}_{f_1,\dots ,f_t}\left[
    \max_f 
    \mathbb{E}_{k} \bra{\psi_{t,k}}V^\dagger O_f^\dagger \Pi_k O_f V\ket{\psi_{t,k}}
    \right].
\]
The maximum winning probability for this $t$-case can thus be upper bounded by the $t=2$ bound:
\begin{align}
    &\sup_{\cA}\left\{\mathbb{E}_{f_1,\dots ,f_t} \left[\Win(\cA\mid f_1,\dots f_t)\right]\right\} \nonumber\\
    &\qquad=\sup_{V,\{\Pi_k\}}\left\{\mathbb{E}_{f_3,\dots ,f_t}\mathbb{E}_{f_1,f_2}
    \left[\max_f \mathbb{E}_k\bra{\psi_{t,k}}V^\dagger O_f^\dagger \Pi_kO_fV\ket{\psi_{t,k}}
    \right]\right\}\nonumber\\
    &\qquad\le
    \mathbb{E}_{f_3,\dots ,f_t} \sup_{V,\{\Pi_k\}}\left\{\mathbb{E}_{f_1,f_2}
    \left[\max_f \mathbb{E}_k\bra{\psi_{t,k}}V^\dagger O_f^\dagger \Pi_kO_fV\ket{\psi_{t,k}}\right]
    \right\}\nonumber\\
    &\qquad=
    \mathbb{E}_{f_3,\dots ,f_t} \sup_{V,\{\Pi_k\}}\left\{\mathbb{E}_{f_1,f_2}
    \left[\max_f \mathbb{E}_k\bra{\psi_{2,k}}U_{[3:t]}^\dagger V^\dagger O_f^\dagger \Pi_kO_fV U_{[3:t]}\ket{\psi_{2,k}}
    \right]\right\}\nonumber\\
    &\qquad= \mathbb{E}_{f_3,\dots ,f_t} \sup_{V,\{\Pi_k\}}\left\{\mathbb{E}_{f_1,f_2}
    \left[\max_f \mathbb{E}_k\bra{\psi_{2,k}}V^\dagger O_f^\dagger \Pi_kO_fV\ket{\psi_{2,k}}
    \right]\right\}\label{eq:f1hf2-change-adversary}\\
    &\qquad=
    \sup_{\cA}\left\{
        \mathbb{E}_{f_1,f_2}[\Win(\cA\mid f_1,f_2)]
    \right\} \nonumber.
\end{align}
In particular, \cref{eq:f1hf2-change-adversary} holds because for every fixed $f_3, \hdots, f_t$ and adversary $\cA$, the previous expression describes the effect of a modified adversary $\cA'$ that applies the isometry $V U_{[3:t]}$ instead of $V$. 

By \cref{sec:f2hf1-final-bound} (or \cref{thm:f2hf1h-search-main}), 
\[
    \mathbb{E}_{f_1,\dots ,f_t} \left[\Win(\cA\mid f_1,\dots f_t)\right]
    =
    O\!\left(
        \frac{\log M\cdot \log (MN)}{K}
    \right),
\]
and this proves the claimed one-query search bound for the family $F_tHF_{t-1}H\cdots F_1 H\ket{k}$.
\end{proof}

For search game over $\{F_tHF_{t-1}H\cdots F_2HF_1H\ket{k}\}_k$, this proof also shows a reduction from $t$ to $t-1$: if the search game for some fixed $t_0$ is $(M,\epsilon)$-hard, then it is also $(M,\epsilon)$-hard for any $t\ge t_0$. This reduction also holds for adversaries that make any fixed number of queries (such as $t_0-1$). 

A similar argument also holds for the $(M,\epsilon,\delta)$-hardness. Therefore, since the one-query lower bound with classical advice for $t=2$ (see \cref{sec:f2hf1-classical-advice}) is obtained by union bounding all classical advice over its $(M,\epsilon,\delta)$-hardness, the same one-query lower bound with classical advice extends to all $t\ge 2$.
\section{One-query synthesis for phase unitaries with constant correctness}
\label{sec:diagonal-positive}

In this section, we give a one-query algorithm for synthesizing (diagonal) phase  unitaries with constant correctness (\cref{thm:diagonal-positive-main}). By combining with our unitary synthesis composition theorem, the one-query algorithm also implies \cref{cor:diagonal-synthesis}.

\subsection{Phase unitary setup}

Any phase unitary on $n$ qubits can be written as
\[
D(F)=\sum_{x\in\{0,1\}^n} \omega_q^{F(x)}\ketbra{x},
\qquad
\omega_q:=e^{2\pi i/q},
\]
for a sufficiently fine phase discretization $F:\{0,1\}^n\to [q]$. The goal is to synthesize $D(F)$ using a single oracle query.

\begin{remark}[Oracle interface used in this section]
The constructive algorithm below is most naturally written in the standard function-oracle model
\[
\ket{x}\ket{y}\longmapsto \ket{x}\ket{y\oplus F(x)}.
\]
For $q=4$, the second register consists of two qubits. 

This interface can be reduced to the boolean phase-oracle model within one query: define $f:\{0,1\}^{n+\lceil\log q\rceil}\to\{0,1\}$ such that $f(x,y)=y\cdot F(x)$. Then $(\Id_n\otimes H^{\otimes \lceil\log q\rceil})\cdot O_f \cdot (\Id_n\otimes H^{\otimes \lceil\log q\rceil})$ will implement the above interface.
\end{remark}

\subsection{The special case $q=4$}
\label{sec:diagonal-q4}

When $q=4$, the four target phases are $\{1,i,-1,-i\}$. Write $F(x)=b_0(x)||b_1(x)\in\{0,1\}^2$, and define the ancilla states
\[
\ket{-}=\stminus,
\qquad
\ket{+i}=\stplusi,
\qquad
\ket{-i}=\stminusi.
\]
The key identities are
\[
X\ket{-}=-\ket{-},
\qquad
X\ket{+i}=i\ket{-i},
\qquad
\braket{+i}{-i}=0.
\]
Observe that if $F(x) \in \{0,1\}^{2}$ is identified as an integer, we have that $i^{F(x)} = (-1)^{b_0(x)} \cdot i^{b_1(x)}$. 

\begin{proposition}
\label{prop:diagonal-q4}
There is a one-query algorithm that synthesizes $D(F)$ for every $F:\{0,1\}^n\to\{0,1,2,3\}$ with correctness at least $1/2$.
\end{proposition}

\begin{proof}
Start from an arbitrary joint input state, with input register $X$ and auxiliary register $\mathrm{Aux}$,
\[
\ket{\psi}_{X,\mathrm{Aux}}=\sum_x \alpha_x \ket{x}_X\ket{\psi_x}_\mathrm{Aux}.
\]
Append the ancilla state $\ket{-}\ket{+i}$. Then query the oracle so that on computational basis, the query acts by
\[
\ket{x}\ket{y}\longmapsto \ket{x}\ket{y\oplus F(x)}.
\]
Equivalently, if $F(x)=b_0(x)||b_1(x)$ then the ancilla undergoes $X^{b_0(x)}\otimes X^{b_1(x)}$.

Let
\[
S_0:=\{x:F(x)\in\{0,2\}\},
\qquad
S_1:=\{x:F(x)\in\{1,3\}\}.
\]
Using the identities above, the post-query state is
\begin{align*}
&\sum_{x\in S_0} \alpha_x\, i^{F(x)}\ket{x}\ket{\psi_x}\ket{-}\ket{+i}
+
\sum_{x\in S_1} \alpha_x\, i^{F(x)}\ket{x}\ket{\psi_x}\ket{-}\ket{-i}\\
&\qquad=
\sum_{x\in S_0} \alpha_x\, D(F)\ket{x}\ket{\psi_x}\ket{-}\ket{+i}
+
\sum_{x\in S_1} \alpha_x\, D(F)\ket{x}\ket{\psi_x}\ket{-}\ket{-i}.
\end{align*}
Tracing out the ancilla destroys the coherence between the $S_0$ and $S_1$ parts but preserves each part exactly.

Define
\[
p_0:=\sum_{x\in S_0}|\alpha_x|^2,
\qquad
p_1:=1-p_0,
\]
and normalized states
\[
\ket{\psi^{(0)}}:=\frac{1}{\sqrt{p_0}}\sum_{x\in S_0}\alpha_x\ket{x}\ket{\psi_x},
\qquad
\ket{\psi^{(1)}}:=\frac{1}{\sqrt{p_1}}\sum_{x\in S_1}\alpha_x\ket{x}\ket{\psi_x}.
\]
The reduced output state is then given by
\[
\rho = p_0\, D(F)\ketbra{\psi^{(0)}} D(F)^\dagger
+ p_1\, D(F)\ketbra{\psi^{(1)}} D(F)^\dagger.
\]
Therefore the fidelity with the ideal output $D(F)\ket{\psi}$ is
\begin{align*}
\Fid\bigl(\rho, D(F)\ket{\psi}\bigr)
&=
\bra{\psi} D(F)^\dagger \rho D(F) \ket{\psi}\\
&= p_0\, |\langle \psi\mid \psi^{(0)}\rangle|^2 + p_1\, |\langle \psi\mid \psi^{(1)}\rangle|^2\\
&= p_0^2+p_1^2\\
&\ge \frac{1}{2},
\end{align*}
since $p_0+p_1=1$ and the minimum of $p_0^2+p_1^2$ occurs at $p_0=p_1=1/2$.
\end{proof}

\subsection{General $q$ via rounding to the nearest quadrant}

For general $q$, write
\[
\omega_q^{F(x)} = a(x)+i b(x)
\qquad\text{with }a(x),b(x)\in\mathbb{R}.
\]
Define a rounded phase function $G:\{0,1\}^n\to\{0,1,2,3\}$ by choosing the nearest fourth root of unity:
\begin{itemize}[leftmargin=1.5em]
    \item if $a(x)\ge 1/\sqrt{2}$, set $G(x)=0$;
    \item if $a(x)\le -1/\sqrt{2}$, set $G(x)=2$;
    \item if $b(x)\ge 1/\sqrt{2}$, set $G(x)=1$;
    \item otherwise set $G(x)=3$.
\end{itemize}
Equivalently, $i^{G(x)}$ is the fourth root of unity whose angle differs from $\omega_q^{F(x)}$ by at most $\pi/4$.

\begin{proposition}
\label{prop:diagonal-general-q}
Applying the $q=4$ algorithm to the rounded phase function $G$ yields a one-query synthesis algorithm for $D(F)$ with correctness at least $1/4$.
\end{proposition}

\begin{proof}
Run the $q=4$ construction from \cref{prop:diagonal-q4} using $G$ in place of $F$. As before, write
\[
S_0:=\{x:G(x)\in\{0,2\}\},
\qquad
S_1:=\{x:G(x)\in\{1,3\}\},
\]
and decompose the input state as
\[
\ket{\psi}=\sum_x \alpha_x \ket{x} \ket{\psi_x}
= \sqrt{p_0}\ket{\psi^{(0)}}+\sqrt{p_1}\ket{\psi^{(1)}},
\]
where
\[
p_0:=\sum_{x\in S_0}|\alpha_x|^2,
\qquad
p_1:=\sum_{x\in S_1}|\alpha_x|^2=1-p_0,
\]
and
\[
\ket{\psi^{(0)}}:=\frac{1}{\sqrt{p_0}}\sum_{x\in S_0}\alpha_x\ket{x}\ket{\psi_x},
\qquad
\ket{\psi^{(1)}}:=\frac{1}{\sqrt{p_1}}\sum_{x\in S_1}\alpha_x\ket{x}\ket{\psi_x}.
\]
The output state after tracing out the ancilla is
\[
\rho = p_0\, D(G)\ketbra{\psi^{(0)}} D(G)^\dagger
+ p_1\, D(G)\ketbra{\psi^{(1)}} D(G)^\dagger.
\]
To compare $D(G)$ with the target $D(F)$, define for each string $x$
\[
\beta_x:=\omega_q^{-F(x)} i^{G(x)}.
\]
Then $|\beta_x|=1$, and by construction of $G(x)$ the angle of $\beta_x$ is at most $\pi/4$ in absolute value. Equivalently,
\[
\Re(\beta_x)\ge \frac{1}{\sqrt{2}}
\qquad\text{for every }x\in\{0,1\}^n.
\]
Also,
\[
D(F)^\dagger D(G)=\sum_x \beta_x \ketbra{x}.
\]
Therefore
\begin{align*}
\Fid\bigl(\rho,D(F)\ket{\psi}\bigr)
&=
 p_0\left|\bra{\psi}D(F)^\dagger D(G)\ket{\psi^{(0)}}\right|^2
 + p_1\left|\bra{\psi}D(F)^\dagger D(G)\ket{\psi^{(1)}}\right|^2.
\end{align*}
We now bound the two overlap terms separately. For the $S_0$ term,
\begin{align*}
\bra{\psi}D(F)^\dagger D(G)\ket{\psi^{(0)}}
&=
\frac{1}{\sqrt{p_0}}
\sum_{x\in S_0}
|\alpha_x|^2 \beta_x.
\end{align*}
Indeed, all cross-terms vanish because $D(F)^\dagger D(G)$ is diagonal in the computational basis. Taking real parts and using $|z|\ge \Re(z)$, we get
\begin{align*}
\left|\bra{\psi}D(F)^\dagger D(G)\ket{\psi^{(0)}}\right|
&\ge
\Re\left(\bra{\psi}D(F)^\dagger D(G)\ket{\psi^{(0)}}\right)\\
&=
\frac{1}{\sqrt{p_0}}
\sum_{x\in S_0}|\alpha_x|^2\Re(\beta_x)\\
&\ge
\frac{1}{\sqrt{2p_0}}
\sum_{x\in S_0}|\alpha_x|^2
= \sqrt{\frac{p_0}{2}}.
\end{align*}
Squaring and multiplying by $p_0$ yields
\[
p_0\left|\bra{\psi}D(F)^\dagger D(G)\ket{\psi^{(0)}}\right|^2
\ge \frac{p_0^2}{2}.
\]
By the same argument,
\[
p_1\left|\bra{\psi}D(F)^\dagger D(G)\ket{\psi^{(1)}}\right|^2
\ge \frac{p_1^2}{2}.
\]
Substituting these two bounds gives
\[
\Fid\bigl(\rho,D(F)\ket{\psi}\bigr)
\ge \frac{p_0^2+p_1^2}{2}
\ge \frac{1}{4}. \qedhere
\]
\end{proof}

{} {}

\section{Quantum advice lower bound for $FH\ket{k}$}
\label{sec:fh-advice}

In this section, we prove a quantum advice lower bound for the state search game with state family $\{FH\ket{k}\}_{k\in[K]}$,
where $F$ is a binary phase unitary. That is, instead of making a query to a classical oracle, the adversary is only given quantum advice that may depend on the underlying state family (equivalently, on $F$) before receiving the input state $\ket{\psi_{F,k}}$.

Note that a binary phase unitary $F$ can be exactly synthesized with one query. Therefore, a quantum advice lower bound for zero-query synthesis gives a separation between one-query unitary synthesis and quantum programs (zero-query synthesis algorithms with quantum advice). By a similar reduction as in \cref{prop:search-implies-synthesis-hardness}, the result of this section implies such a lower bound/separation. 

We remark that a more straightforward but quantitatively weaker separation holds by considering unitaries of the form $\ket{x}\ket{y} \mapsto \ket{x} \ket{y\oplus f(x)}$ for a random (possibly long output) function $f$. The separation is weaker because these unitaries are only as hard as computing a function on $n$ input bits, they will not have the same quantitative hardness as binary phase unitaries in the same dimension: either $y$ is short and there is a non-trivial approximation by guessing $f(x)$ on input $\ket{x}$, or $y$ is long and $2^n$ advice length is sublinear in the Hilbert space dimension. 

Unlike the permutation and $F_2HF_1$ one-query lower bounds, %
~we do not build on the matrix concentration-based approach of \cite{STOC:LomMaWri24} for this result. Instead, we make use of the alternating measurement hardness approach to advice lower bounds of \cite{EC:Liu23}. 

\paragraph{Organization.} We start with the formulation of the search game in \cref{sec:FH-search-formulation}. Then, we prove the quantum advice lower bound in two steps. First, we reduce the one-instance search hardness to the hardness of an alternating measurement game in \cref{sec:fh-reduce-to-alternating-measurement}. Next, in \cref{sec:fh-upper-bounding-alternating-measurement}, we upper bound the maximum winning probability of this alternating measurement game. We combine these results and conclude the quantum advice lower bound in \cref{sec:fh-conclusion}. In \cref{sec:fh-tightness}, we show an algorithm that matches the lower bound (up to $\log N$ factor).

\subsection{Search formulation}\label{sec:FH-search-formulation}

Let $f:\{0,1\}^n\to\{0,1\}$ be uniformly random and let
\[
F:=\sum_{x\in\{0,1\}^n} (-1)^{f(x)}\ketbra{x}.
\]
The challenge state is $FH\ket{k}$ for uniformly random $k\in[K]$. A non-uniform (zero-query) algorithm is allowed to use an $S$-qubit advice state $\ket{\phi_f}$ depending only on $f$, and is asked to output $k$.

Equivalently, we will work in the following normal form. The adversary has an $S$-qubit advice together with some ancilla initialized as $\ket{0^m}$ on the adversary's workspace register $\mathsf{Z}$. The challenge state will be generated and sent to the adversary on input register $\mathsf{X}$. The adversary will then apply a fixed projective measurement $\{\Pi_k\}_{k\in [K]}$ to the input register $\mathsf{X}$ and the workspace $\mathsf{Z}$. Writing
\[
\ket{\psi_{f,k}}:=FH\ket{k},
\]
its winning probability is
\[
\varepsilon
:=
\mathbb{E}_{f,k}\left\|\Pi_k\bigl(\ket{\psi_{f,k}}_\mathsf{X}\ket{\phi_f,0^m}_{\mathsf{Z}}\bigr)\right\|^2.
\] 
The main result of this section is that an adversary's maximum winning probability is upper bounded by

\[
\varepsilon = O\!\left(\frac S K\right).
\]

\subsection{Reduction to alternating measurement game}\label{sec:fh-reduce-to-alternating-measurement}

Our first step is to reduce the one-instance winning probability to the winning probability of a $t$-round alternating measurement game. The alternating measurement game is  first introduced in \cite{EC:Liu23} in order to obtain better security in the presence of quantum advice. Within our context of the search game, we define our alternating measurement game as the following.

\begin{definition}[Alternating measurement game] A random boolean function $f$ is sampled at the beginning. For a (non-uniform) quantum algorithm $\mathcal{A}$ and any integer $t$, the alternating measurement game\footnote{Although we refer to it as a ``game,'' we remark that it is only a thought experiment, not a game that can physically be played between the challenger and adversary.} we consider here is defined as follows:
\begin{itemize}
    \item The challenger initializes its challenge register as $\ket{\psi_{\mathrm{init}}}:=\frac{1}{\sqrt{K}}\sum_{k}\ket{k}_{\mathsf{C}}\otimes H\ket{k}_{\mathsf{X}}$ on challenge register $\mathsf{C}$ and input register $\mathsf{X}$. 
    \item The adversary initializes its state (or an advice) on adversary's workspace register $\mathsf{Z}$. Their algorithm is defined by a $K$-outcome measurement $\{\Pi_k\}_{k\in [K]}$ on $\mathsf{X},\mathsf{Z}$.
    \item The challenger generates the first challenge by applying $F$ on register $\mathsf{X}$,
    \item They repeat the following procedure $t$ times, for $i=1,2,\ldots,t$:
    \begin{itemize}
        \item If $i$ is odd, apply the measurement defined by projection $\Pi_{\mathrm{Win}}:=\sum_k\ketbra{k}_{\mathsf{C}}\otimes (\Pi_k)_{\mathsf{XZ}}$ to $\mathsf{CXZ}$.
        \item If $i$ is even, apply the measurement defined by projection 
        $\Pi_{\mathrm{init},F}:=F \Pi_{\mathrm{init}}F$ to $\mathsf{CX}$, where $\Pi_{\mathrm{init}}=\ketbra{\psi_{\mathrm{init}}}$. %
    \end{itemize}
    \item The adversary wins the game if all measurement outcomes are $1$. 
\end{itemize}
\end{definition}

\begin{lemma}[Reducing to alternating measurement hardness]
\label{lem:fh-powering-by-alternating-measurement}
If a non-uniform algorithm with $S$ qubits of advice wins the one-instance search game with probability $\varepsilon$, then for every integer $t\ge 1$, there exists a $t$-round alternating measurement game, using the same $S$ qubits of advice, that wins with probability at least $\epsilon^t$.
\end{lemma}

\begin{proof}[Proof sketch] By the definition of these two projectors $\{\Pi_{\mathrm{Win}},\Pi_{\mathrm{init},F}\}$, we can rewrite our one-instance winning probability 
\[
\varepsilon
:=
\mathbb{E}_{f,k}\left\|\Pi_k\bigl(\ket{\psi_{f,k}}_\mathsf{X}\ket{\phi_f,0^m}_{\mathsf{Z}}\bigr)\right\|^2
=\mathbb{E}_f\|\Pi_{\mathrm{Win}}F\ket{\psi_{\mathrm{init}}}_{\mathsf{CX}}\ket{\phi_f,0^m}_\mathsf{Z}\|^2,
\] 
where the starting state $F\ket{\psi_{\mathrm{init}}}$ is in the image of $\Pi_{\mathrm{init},F}$. The lemma now follows by a standard rewinding argument \cite{FOCS:CMSZ21,EC:Liu23}. By Jordan's lemma, the two projectors decompose the space into orthogonal invariant subspaces of dimension at most two. 
On block with singular value $p$, if the initial state is on the corresponding singular vector in the image of $\Pi_{\mathrm{init},F}$, then its probability of surviving $t$ rounds is $p^t$. Thus, for a general initial state with overall success probability $\varepsilon = \mathbb{E}[p]$, the probability of surviving $t$ rounds is $\mathbb{E}[p^t]$, which is at least $(\mathbb{E}[p])^t = \varepsilon^t$ by Jensen's inequality.
\end{proof}

\subsection{Upper-bounding the alternating measurement game}\label{sec:fh-upper-bounding-alternating-measurement}
The second step is to upper bound the winning probability of a non-uniform $t$-round alternating measurement game with $S$-qubit quantum advice. The proof proceeds in two sub-steps: first reduce the non-uniform hardness to a uniform one by replacing the advice with the maximally mixed state; then prove uniform hardness of the $t$-round alternating measurement game by bounding the conditional winning probability at each round. To bound each conditional winning probability, we will rely on the randomness of $F$ and apply Zhandry's compressed oracle technique \cite{C:Zha19}.

Here we prove an upper bound for the winning probability of a uniform $t$-round alternating measurement game. This will also upper bound the non-uniform case: for any adversary with $S$-qubit quantum advice with winning probability $\epsilon_{S}$, a uniform algorithm can always sample an $S$-qubit maximally mixed state and run the non-uniform algorithm on the maximally mixed state, with winning probability at least $2^{-S}\cdot\epsilon_S$. Therefore, an upper bound for the winning probability of the uniform case will upper bound $2^{-S}\cdot \epsilon_S$, and thus give an upper bound for $\epsilon_S$.

\begin{proposition}[Winning probability of (uniform) $t$-round alternating measurement game]\label{prop:t-round-alternating-measurment}
For every uniform adversary in the alternating measurement game with $t$ measurement rounds, its maximum winning probability is upper bounded by $(8t/K)^t$.
\end{proposition}

We denote the measurement outcome in the $i$-th round as $b_i$, and let $b_i=1$ if the state successfully projects onto $\Pi_{\mathrm{Win}}$ (if $i$ is odd), or $\Pi_{\mathrm{init},F}$ (if $i$ is even). We also define the conditional probability for successfully projecting on the $i$-th round as $\epsilon_i$, 
\[
\epsilon_i = \Pr[b_i = 1 \mid b_{<i}=1]=\frac{\Pr[b_{\le i}=1]}{\Pr[b_{\le i-1}=1]}.
\]

For alternating measurement game with $t$ rounds, we define $b_0=1$ always, and the winning probability can be written as 
\[
\Pr[b_t=1]=\prod_{i=1}^{t}\Pr[b_i=1\mid b_{<i}=1]=\prod_{i=1}^t \epsilon_i.
\]
The conditional probability is monotonically non-decreasing, as argued in \cite[Corollary~6.10]{EC:Liu23}. 

\begin{proposition}[Non-decreasing of $\epsilon_i$]\label{prop:alternating-measurement-non-decreasing}
$\{\epsilon_i\}_{i\in\{1,\ldots,t\}}$ is monotonically non-decreasing, i.e.,
for every $i\in\{2,\ldots,t\}$, $\epsilon_{i-1}\le \epsilon_{i}$.
\end{proposition}
\begin{proof}[Proof sketch]
$\epsilon_i = \Pr[b_i = 1 \mid b_{<i}=1]=\Pr[b_{\le i}=1]/\Pr[b_{\le i-1}=1]$. Similarly as in the proof for \cref{lem:fh-powering-by-alternating-measurement}, by Jordan's lemma, $\Pr[b_{\le i}=1]$ can be expressed as $\mathbb{E}[p^i]$. By the Cauchy-Schwarz inequality, $\mathbb{E}[p^{i-1}]\cdot \mathbb{E}[p^{i+1}]\ge \mathbb{E}[p^{i}]^2$, and thus $\epsilon_{i+1}\ge \epsilon_i$.
\end{proof}

With the non-decreasing property, it is sufficient to bound the conditional probability at only odd rounds. %
\begin{lemma}\label{lem:alternating-measurement-conditional-probability}
For every uniform adversary in the alternating measurement game, $\epsilon_t\le 8t/K$ for every $t$. Specifically, for every odd $t$, $\epsilon_t \le 4t/K$. By the non-decreasing property of $\epsilon_t$, for even $t$, we have $\epsilon_t \le \epsilon_{t+1}\le 8t/K$ with one more round of the alternating measurement game.
\end{lemma}

To bound this conditional winning probability, we will need to apply Zhandry's compressed oracle framework \cite{C:Zha19}. Since $\Pi_{\mathrm{Win}}$ is independent of $F$, while $\Pi_{\mathrm{init},F}=F\Pi_{\mathrm{init}}F$, we can view the conditional probability as making several queries to $F$ (or phase queries to the underlying boolean function $f$), while performing some intermediate measurements in between. Since we are analyzing probability over a random $f$, we can purify the register for $f$ and view it in the Fourier basis as ``database''. The algorithm starts with a pure uniform superposition of all possible $f$, which corresponds to an initialized empty database; from the framework in  \cite{C:Zha19}, any query to $f$ performing $\ket{x}\ket{f}\mapsto (-1)^{f(x)}\ket{f}$ will correspond to $\ket{x}\ket{D}\mapsto \ket{x}\ket{D\oplus \{x\}}$ in the database view, where $D\oplus \{x\}=D\backslash\{x\}$ if $x\in D$, and $D\oplus \{x\}=D\cup\{x\}$ if $x\notin D$. 

Within the compressed oracle framework, our alternating measurement game can be viewed with one more register $\mathsf{D}$ for the database. It is initialized as the empty set, and each ``query to $F$'' is replaced by a compressed oracle update on the database. The database register $\mathsf{D}$ is not touched by the projectors $\Pi_{\mathrm{Win}}, \Pi_{\mathrm{init}}$, although it is affected by (the purification of) $F \Pi_{\mathrm{init}} F$.%

\begin{proof}[Proof of \cref{lem:alternating-measurement-conditional-probability}]
We prove the lemma for odd $t$ (as even $t$ follows from \cref{prop:alternating-measurement-non-decreasing}).

Define $\ket{\Phi_{t-1}}$ to be the normalized result state just after the $(t-1)$-th round with outcome $\{b_{i}=1\}_{i\in [t-1]}$. For odd $t$, by definition, $\epsilon_t=\|\Pi_{\mathrm{Win}}\ket{\Phi_{t-1}}\|^2$.

\paragraph{Result state after $(t-1)$ rounds.} We start by describing the state $\ket{\Phi_{t-1}}$. Since $\Pi_{\mathrm{init},F}=F\Pi_{\mathrm{init}}F$, any result state after $(t-2)$ rounds in the alternating measurement game can be viewed as obtained by making $(t-3)$ queries to $F$ and performing many intermediate measurements that are independent of $F$ (for odd $t$). Within the compressed oracle framework, this gives an upper bound on the size of the database on register $\mathsf{D}$. 

For the $(t-1)$-th round (for odd $t$), we can view the projector $\Pi_{\mathrm{init},F}=F\Pi_{\mathrm{init}}F$ as the following: the algorithm first makes 1 query to $F$, then it successfully measures on $\ket{\psi_{\mathrm{init}}}$ on register $\mathsf{C}\mathsf{X}$, and then it makes another query to $F$. Therefore, the result state after the first $(t-2)$ rounds along with the next query to $F$ has the form
\[
\propto \sum_{k}\ket{k}_{\mathsf{C}}\otimes \sum_{x,z,D}\alpha_{k,x,z,D}\ket{x}_\mathsf{X}\ket{z}_{\mathsf{Z}}\ket{D}_{\mathsf{D}}
\]
with database size $|D|\le t-2$. Next, within the $(t-1)$-th round, since the algorithm successfully measures on $\ket{\psi_{\mathrm{init}}}$ on register $\mathsf{C}\mathsf{X}$, the result state on $\mathsf{C}\mathsf{X}$ is a pure state $\ket{\psi_{\mathrm{init}}}$ and thus is unentangled with $\mathsf{Z}$ and $\mathsf{D}$, which has the form
\[
\propto \left(\sum_{k}\ket{k}_{\mathsf{C}}\otimes H\ket{k}_{\mathsf{X}}\right)
\otimes \left(\sum_{z,D}\beta_{z,D}\ket{z}_{\mathsf{Z}}\ket{D}_{\mathsf{D}}\right).
\]
Therefore, by making one query to $F$, we will end up with $\ket{\Phi_{t-1}}$, the result state just after the $(t-1)$-th round, 
\[
\ket{\Phi_{t-1}}=\frac{1}{\sqrt{KN}}
\sum_{k,x}(-1)^{k\cdot x}\ket{k,x}_{\mathsf{CX}}\otimes \sum_{z,D}\beta_{z,D}\ket{z}_{\mathsf{Z}}\ket{D\oplus\{x\}}_{\mathsf{D}}.
\]

Now our goal is to give an upper bound for $\epsilon_t=\|\Pi_{\mathrm{Win}}\ket{\Phi_{t-1}}\|^2$. In the below analysis we use $\ket{\Phi}:=\ket{\Phi_{t-1}}$, as we only analyze the result state after $(t-1)$-th round.

We define 2 parts for $\ket{\Phi}$ as the following $\ket{\Phi_1},\ket{\Phi_2}$, $\ket{\Phi}=\ket{\Phi_1}+\ket{\Phi_2}$:
\begin{align*}
    \ket{\Phi_1}
    &:=\frac{1}{\sqrt{KN}}\sum_{k,z,D,x\in D}(-1)^{k\cdot x }\cdot \beta_{z,D}\cdot \ket{k,x,z}\ket{D\backslash\{x\}}
    \\
    \ket{\Phi_2}
    &:=\frac{1}{\sqrt{KN}}\sum_{k,z,D,x\notin D}(-1)^{k\cdot x }\cdot \beta_{z,D}\cdot \ket{k,x,z}\ket{D\cup\{x\}}.
\end{align*}
Therefore,
\[
\|\Pi_{\mathrm{Win}}\ket{\Phi}\|^2
=\|\Pi_{\mathrm{Win}}(\ket{\Phi_1}+\ket{\Phi_2})\|^2
\le 2\cdot (\|\Pi_{\mathrm{Win}}\ket{\Phi_1}\|^2+\|\Pi_{\mathrm{Win}}\ket{\Phi_2}\|^2),
\]
and now our goal is to bound $\|\Pi_{\mathrm{Win}}\ket{\Phi_1}\|^2$ and $\|\Pi_{\mathrm{Win}}\ket{\Phi_2}\|^2$ separately.

\paragraph{For $\ket{\Phi_1}$.} Over all $x\in [N]$, for any database $D$ with size $|D|\le t-2$, only a small fraction of $x$ will lie in $D$. This intuition gives us the following upper bound:
\begin{align*}
    \|\Pi_{\mathrm{Win}}\ket{\Phi_1}\|^2
    &=\frac{1}{KN}\sum_k\left\|
    \sum_{z,D,x\in D}(-1)^{k\cdot x }\cdot \beta_{z,D}\cdot \Pi_k\ket{x,z,D\backslash\{x\}}
    \right\|^2\\
    &\le\frac{1}{KN}\sum_k\left\|
    \sum_{z,D,x\in D}(-1)^{k\cdot x }\cdot \beta_{z,D}\cdot \ket{x,z,D}
    \right\|^2\\
    &=\frac{1}{KN}\sum_k \sum_{z,D}|\beta_{z,D}|^2\cdot \left(\sum_{x\in D}[x\in D]\right)\\
    &\le \frac{t-2}{N}.
\end{align*}

\paragraph{For $\ket{\Phi_2}$.} Since $\Pi_k$ never acts on the database register $\mathsf{D}$, we will have orthogonality for different database, and thus we can expand the term as
\begin{align*}
    \|\Pi_{\mathrm{Win}}\ket{\Phi_2}\|^2
    &=\frac{1}{K}\sum_k\left\|
    \sum_{z,D,x\notin D}(-1)^{k\cdot x }\cdot \frac{\beta_{z,D}}{\sqrt{N}}\cdot (\Pi_k\ket{x,z})\otimes \ket{D\cup\{x\}}
    \right\|^2\\
    &=\frac{1}{K}\sum_k\left\|
    \sum_{z,D',x\in D'}(-1)^{k\cdot x }\cdot \frac{\beta_{z,D'\backslash\{x\}}}{\sqrt{N}}\cdot (\Pi_k\ket{x,z})\otimes \ket{D'}
    \right\|^2\\
    &=\frac{1}{K}\sum_k\sum_{D'}\left\|\Pi_k\cdot 
    \sum_{z,x\in D'}(-1)^{k\cdot x }\cdot \frac{\beta_{z,D'\backslash\{x\}}}{\sqrt{N}}\cdot \ket{x,z}
    \right\|^2
    \\
    &=\frac{1}{K}\sum_k \sum_{D'} p_{D'}\cdot\|\Pi_k\ket{\phi_{k,D'}}\|^2
\end{align*}
with the following definitions
\[
p_{D'}:=\sum_{z,x\in D'}\frac{|\beta_{z,D'\backslash\{x\}}|^2}{N},
\qquad
\ket{\phi_{k,D'}}:=\frac{1}{\sqrt{p_{D'}}}\sum_{z,x\in D'}(-1)^{k\cdot x }\cdot \frac{\beta_{z,D'\backslash\{x\}}}{\sqrt{N}}\cdot \ket{x,z}.
\]
Since $\|\ket{\Phi_2}\|\le 1$, $\sum_{D'}p_{D'}=\|\ket{\Phi_2}\|^2\le 1$.

Note that from the definition of $\ket{\phi_{k,D'}}$, it is normalized, and if we define state $\ket{v_{x,D'}}$ as
\[
    \ket{v_{x,D'}}=\frac{1}{\sqrt{p_{D'}}}\sum_{z}\frac{\beta_{z,D'\backslash\{x\}}}{\sqrt{N}}\ket{x,z},
\]
then for any $k\in \{0,1\}^n$, $\ket{\phi_{k,D'}}$ actually lies in a small subspace spanned by no more than $|D'|$ states,
\[
\ket{\phi_{k,D'}}=\sum_{x\in D'}(-1)^{k\cdot x}\ket{v_{x,D'}}\quad
\in \quad
\mathrm{Span}\{\ket{v_{x,D'}}\}_{x\in D'}.
\]
We define $\Pi_{D'}$ as the projection on this subspace, and $\Trace(\Pi_{D'})\le |D'|\le t-2+1=t-1$.

This limitation of subspace spanned by $\{\phi_{k,D'}\}_k$ for every $D'$ will help us upper bound $\|\Pi_{\mathrm{Win}}\ket{\Phi_2}\|^2$:
\begin{align*}
    \|\Pi_{\mathrm{Win}}\ket{\Phi_2}\|^2
    &=\frac{1}{K}\sum_k \sum_{D'} p_{D'}\cdot\|\Pi_k\ket{\phi_{k,D'}}\|^2
    =\frac{1}{K}\sum_k \sum_{D'} p_{D'}\cdot\|\Pi_k\Pi_D'\ket{\phi_{k,D'}}\|^2\\
    &=\frac{1}{K}\sum_k \sum_{D'} p_{D'}\cdot\Trace(\Pi_{D'}\Pi_k\Pi_{D'}\ketbra{\phi_{k,D'}})\\
    &\le \frac{1}{K}\sum_k \sum_{D'} p_{D'}\cdot\Trace(\Pi_{D'}\Pi_k\Pi_{D'})\\
    &=\frac{1}{K}\sum_{D'} p_{D'}\cdot\Trace(\Pi_{D'})\\
    &\le \frac{t-1}{K}.
\end{align*}

\paragraph{Concluding the proof.}
Since here the key space $K\le N$, we have 
\[
\|\Pi_{\mathrm{Win}}\ket{\Phi}\|^2
\le 2\cdot (\|\Pi_{\mathrm{Win}}\ket{\Phi_1}\|^2+\|\Pi_{\mathrm{Win}}\ket{\Phi_2}\|^2)
\le 2\cdot \left(\frac{t-1}{K}+\frac{t-2}{N}\right)<\frac{4t}{K}. \qedhere 
\]
\end{proof}

\begin{proof}[Proof of \cref{prop:t-round-alternating-measurment}]
With the upper bound of $\epsilon_s$ from \cref{lem:alternating-measurement-conditional-probability} and non-decreasing property from \cref{prop:alternating-measurement-non-decreasing}, 
\[
\Pr[b_t=1]=\prod_{i=1}^t\epsilon_i\le \epsilon_t^t \le \left(\frac{8t}{K}\right)^t. \qedhere
\]
\end{proof}

\subsection{Conclusion}\label{sec:fh-conclusion}
Combining \cref{lem:fh-powering-by-alternating-measurement} and \cref{prop:t-round-alternating-measurment}, we obtain the following theorem.

\begin{theorem}
\label{thm:fh-advice-final}
If a non-uniform algorithm with $S$ qubits of advice succeeds in recovering $k$ from one copy of $FH\ket{k}$ with probability $\epsilon$, then for every integer $t\ge 1$,
\[
\varepsilon^{t}\le 2^S\left(\frac{8t}{K}\right)^{t}
\qquad\text{for every integer }t\ge 1.
\]
In particular, setting $t=S$ gives 
\[
\epsilon\le \frac{16S}{K}.
\]
\end{theorem}
\begin{proof}
    By \cref{lem:fh-powering-by-alternating-measurement}, the $t$-round alternating measurement can be won with probability at least $\epsilon^t$. By \cref{prop:t-round-alternating-measurment}, every such game has success probability at most $2^S(8t/K)^t$. Combining the two inequalities proves the claim.
\end{proof}

\subsection{Matching Algorithm}
\label{sec:fh-tightness}

In this section, we give an algorithm with an $S$-qubit advice for the state search game on state family $\{FH\ket{k}\}_{k\in [K]}$. This algorithm can reach winning probability $\tilde{\Omega}(S/N)$ for $S=o(N/\log N)$, which matches our lower bound above (up to $\log N$ factor, for the case when $K=\Omega(N)$).

\paragraph{Algorithm.}
Suppose that the algorithm receives $q$ copies of
$\ket{\psi_{f,0}}:=FH\ket{0}$ as its $q\log N$-qubit advice. Denote the advice registers by $X_1,\ldots,X_q$ and the challenge
register by $Y$. 
\begin{enumerate}
    \item The algorithm coherently computes into an ancilla the
smallest index
\[
    I:=\min\{i\in[q]:X_i=Y\},
\]
setting $I=q+1$ if no such index exists, and then measures $I$.

\item If
$I=q+1$, the algorithm aborts. If $I=i\le q$, it applies a bitwise
CNOT from $X_i$ to $Y$, discards $Y$ and all advice registers other
than $X_1,\ldots,X_i$, and measures $X_i$ in the Hadamard basis to
obtain its output.
\end{enumerate}

\paragraph{Analysis.}
For every $i\in[q]$, $\Pr[I=i]
    =
    \frac1N(1-\frac1N)^{i-1}$.
Conditioned on $I=i$, the equality of the computational-basis values in
$X_i$ and $Y$ cancels their two copies of the phase $(-1)^{f(x)}$.
After the CNOT, the remaining state can be written as
\[
    \ket{\Gamma_{i,k}}
    =
    \frac1{\sqrt N}\sum_x
      (-1)^{k\cdot x}
      \ket{\theta_x}^{\otimes(i-1)}\otimes \ket{x},
\]
where $\ket{\theta_x}
    :=
    \frac1{\sqrt{N-1}}
    \sum_{z\ne x}(-1)^{f(z)}\ket{z}$.

Note that for distinct $x,y$, $\left<\theta_y|\theta_x\right>=\frac{N-2}{N-1}$.
Therefore, after tracing out the first $i-1$ registers, the state of
$X_i$ is
\[
    \rho_{i,k}
    =
    \alpha_i\cdot  H\ketbra{k}H
    +(1-\alpha_i)\cdot \frac{\Id}{N},
    \qquad
    \alpha_i:=
    \left(\frac{N-2}{N-1}\right)^{i-1}.
\]
The Hadamard-basis measurement thus outputs $k$ with probability at
least $\alpha_i$. 

Consequently, for $q=o(N)$, the overall winning probability is at
least
\begin{align*}
    \sum_{i=1}^q \Pr[I=i]\cdot \alpha_i
    =
    \frac1N\sum_{i=1}^q
      \left(1-\frac2N\right)^{i-1} 
    =\frac{1-(1-2/N)^q}{2}
    =\Omega\left(\frac{q}{N}\right).
\end{align*}

Therefore, in terms of the total advice length $S=q\log N$, the algorithm
uses $S$ advice qubits and achieves $\Omega\left(q/N\right)=\tilde{\Omega}(S/N)$ winning probability in the search game.

\bibliographystyle{alpha}
\bibliography{custom,abbrev3,crypto}

\appendix
\section{Simple search game upper bounds}
\label{app:iid-search}

This section describes two additional applications of the spectral reduction from \cref{lem:generic-spectral-reduction}. The first results in a simpler alternative proof of some of the main results of \cite{STOC:LomMaWri24} about random binary phase states. The second proves a quantitatively similar statement bounding the maximum winning probability of the search game for Haar-random unitaries, which corresponds to a one-query lower bound for synthesizing Haar-random unitaries. 

\subsection{Random binary phase states}

Let $R=(R(k,x))_{k\in[K],\,x\in[N]}$ be a family of independent Rademacher random variables, and define
\[
\ket{\psi_{R,k}}:=\frac{1}{\sqrt{N}}\sum_{x\in[N]} R(k,x)\ket{x}
\qquad\text{for each }k\in[K].
\]
Thus the challenge states are independent random binary phase states.

\begin{theorem}[Search bound for binary phase states]
\label{thm:iid-binary-phase-search}
For the oracle state search game associated with the family $\{\ket{\psi_{R,k}}\}_{k\in[K]}$ above, every one-query adversary with workspace dimension $M$ satisfies
\[
\mathbb{E}_R\bigl[\Win(\cA\mid R)\bigr]
\le \frac{2\bigl(1+\log(2M)\bigr)}{K}.
\]
\end{theorem}

\begin{proof}
Fix a one-query adversary $\cA=(V,\{\Pi_k\}_{k\in[K]})$. For each $k\in[K]$, the random state $\ket{\psi_{R,k}}$ has marginal
\[
\mathbb{E}_R\bigl[\ketbra{\psi_{R,k}}\bigr]=\frac{1}{N}\cdot \Id.
\]
\cref{lem:generic-spectral-reduction,lem:subspace-specialization} then tell us that
\[
\mathbb{E}_R\bigl[\Win(\cA\mid R)\bigr] \leq \mathbb E_R\bigl[\norm{M_R}^2 \bigr],
\]
where
\[
M_R=\sum_{k\in[K],\,x\in[N]} R(k,x)\,B_{k,x}
\]
and the matrices $\{B_{k,x}\}_{k,x}$ satisfy
\[
\sum_{k,x} B_{k,x}B_{k,x}^\dagger=\frac{1}{K} \cdot \Id,
\qquad
\sum_{k,x} B_{k,x}^\dagger B_{k,x}\preceq \frac{1}{K} \cdot \Id.
\]
Since the coefficients $\{R(k,x)\}_{k,x}$ are independent Rademacher random variables, we may apply \cref{thm:tropp-rademacher} to the matrix Rademacher series $M_R$. Its matrix variance parameter is
\[
v(M_R)=\max\!\left\{
\left\|\sum_{k,x} B_{k,x}B_{k,x}^\dagger\right\|,
\left\|\sum_{k,x} B_{k,x}^\dagger B_{k,x}\right\|
\right\}=\frac{1}{K}.
\]
Because $M_R$ is an $M\times M$ matrix, \cref{thm:tropp-rademacher} gives
\[
\mathbb{E}_R\norm{M_R}^2
\le 2\bigl(1+\log(2M)\bigr)\cdot \frac{1}{K},
\]
completing the proof. 
\end{proof}
\cref{thm:iid-binary-phase-search} gives an alternative proof of the hardness of one-query unitary synthesis \cite{STOC:LomMaWri24}. 

\subsection{Haar-random unitaries}

Let $U=(U_{k,x})_{k,x\in[N]}$ be Haar-random in $U(N)$, and define
\[
\ket{\psi_{U,k}}:= U\ket{k} = \sum_{x\in[N]} U_{k,x}\ket{x}
\qquad\text{for each }k\in[N].
\]
Thus $\{\ket{\psi_{U,k}}\}_{k\in[N]}$ is a Haar-random orthonormal basis of $\mathbb C^N$. We wish to upper bound the probability of winning the state search game on input $\ket{\psi_{U, k}}$ for uniform $k \in [K]\subset [N]$. 

We will make use of the following inequality for matrix Gaussian series. 

\begin{lemma}[Matrix Gaussian series, {\cite[Theorem~4.1.1 and Equation~(4.1.7)]{tropp2015book}}]
\label{lem:tropp-gaussian-complex}
Let $\{\gamma_j\}_j$ be independent standard complex Gaussian random variables, and let $\{B_j\}_j$ be fixed complex matrices of dimension $d_1\times d_2$. Define
\[
Z:=\sum_j \gamma_j B_j,
\qquad
v(Z):=\max\!\left\{
\left\|\sum_j B_j B_j^\dagger\right\|,
\left\|\sum_j B_j^\dagger B_j\right\|
\right\}.
\]
Then
\[
\mathbb E\norm{Z}^2\le 2\,v(Z)\bigl(1+\log(d_1+d_2)\bigr).
\]
\end{lemma}

\begin{proof}
Write $\gamma_j=(g_j+ih_j)/\sqrt 2$, where $\{g_j\}_j$ and $\{h_j\}_j$ are independent families of real standard normal random variables. Then
\[
Z=
\sum_j g_j\,\frac{B_j}{\sqrt 2}
+
\sum_j h_j\,\frac{iB_j}{\sqrt 2},
\]
so $Z$ is a real Gaussian matrix series with coefficient family $\{B_j/\sqrt 2,\,iB_j/\sqrt 2\}_j$. The corresponding variance parameter is exactly $v(Z)$, because
\[
\sum_j \frac{B_j B_j^\dagger}{2}+
\sum_j \frac{(iB_j)(iB_j)^\dagger}{2}
=
\sum_j B_j B_j^\dagger
\]
and similarly on the right. The claimed bound therefore follows from \cite[Theorem~4.1.1 and Equation~(4.1.7)]{tropp2015book}.
\end{proof}

Next, by a reduction to the case of independent Gaussians --- analogous to comparison-based arguments of Tropp \cite{tropp-comparison} for Haar-random \emph{real} orthogonal matrices --- we analyze random matrices with coefficients coming from Haar-random unitaries (rather than fully independent coefficients). 

\begin{proposition}[Haar-unitary matrix series]
\label{prop:haar-unitary-matrix-series}
Let $1\le K\le N$, let $U\in U(N)$ be Haar-random, and let $\{B_{k,x}\}_{k\in[K],\,x\in[N]}$ be fixed complex matrices of dimension $d_1\times d_2$. Define
\[
Z_{U,K}:=\sum_{k\in[K],\,x\in[N]} \sqrt N\,U_{k,x} B_{k,x},
\qquad
v_K:=\max\!\left\{
\left\|\sum_{k\in[K],\,x\in[N]} B_{k,x}B_{k,x}^\dagger\right\|,
\left\|\sum_{k\in[K],\,x\in[N]} B_{k,x}^\dagger B_{k,x}\right\|
\right\}.
\]
Then
\[
\mathbb E_U\norm{Z_{U,K}}^2
\le 4\,v_K\bigl(1+\log(d_1+d_2)\bigr).
\]
\end{proposition}

\begin{proof}
Extend the family $(B_{k,x})_{k\in[K],\,x\in[N]}$ to indices $k\in[N]$ by setting $B_{k,x}:=0$ for $k\in[N]\setminus[K]$. Then
\[
Z_{U,K}=\sum_{k,x\in[N]} \sqrt N\,U_{k,x} B_{k,x},
\qquad
v_K=\max\!\left\{
\left\|\sum_{k,x\in[N]} B_{k,x}B_{k,x}^\dagger\right\|,
\left\|\sum_{k,x\in[N]} B_{k,x}^\dagger B_{k,x}\right\|
\right\}.
\]
Thus, it suffices to prove the bound in the special case $K=N$.

Let $G=(G_{k,x})_{k,x\in[N]}$ be a random matrix with independent complex Gaussian entries $G_{k,x}\sim \mathcal N_{\mathbb C}(0,1/N)$. Since $G$ is invertible with probability $1$, we write its polar decomposition
\[
G=U P,
\qquad
U\in U(N),
\qquad
P=(G^\dagger G)^{1/2}\succeq 0.
\]
Since $G$ is (left-) unitary invariant, it holds that $U$ is Haar-random and independent of $P$. Moreover, since $G = G\cdot W$ is invariant under \emph{right} multiplication of an arbitrary fixed unitary $W$, we have that $P = (G^\dagger G)^{1/2} = W^\dagger P W$ is conjugation-invariant. Thus, $\mathbb E[P]$ commutes with every unitary, and so
\[
\mathbb E[P]=a_N \cdot \Id,
\]
for
\[
a_N=\frac{1}{N} \cdot \mathbb E\Tr(P).
\]
Since $U$ and $P$ are independent, this lets us calculate the conditional expectation 
\[
\mathbb E[G\mid U]=U \cdot \mathbb E[P]=a_N \cdot U.
\]
Next, define the linear map
\[
T(G):=\sum_{k,x\in[N]} G_{k,x} B_{k,x}.
\]
The function
\[
f(G):=\norm{\sqrt N\,T(G)}^2
\]
is convex and satisfies $f(\lambda \cdot A)=\lambda^2 f(A)$ for every scalar $\lambda \ge 0$. Jensen's inequality therefore yields
\[
a_N^2 \cdot f(U) = f(a_N U)=f\bigl(\mathbb E[G\mid U]\bigr)
\le \mathbb E\bigl[f(G)\mid U\bigr].
\]
Taking expectations, this implies that
\[
a_N^2\,\mathbb E_U\norm{Z_{U,K}}^2
\le
\mathbb E\norm{\sum_{k,x\in[N]} \sqrt N\,G_{k,x} B_{k,x}}^2.
\]
Because $\sqrt N\,G_{k,x}$ are independent standard complex Gaussian variables, \cref{lem:tropp-gaussian-complex} gives
\[
\mathbb E\norm{\sum_{k,x\in[N]} \sqrt N\,G_{k,x} B_{k,x}}^2
\le 2\,v_K\bigl(1+\log(d_1+d_2)\bigr).
\]
Thus, all that remains is to lower bound $a_N$. Fortunately, it is known (see, e.g., \cite{unitary-vs-gaussian}) that $a_N \geq 1/\sqrt{2}$ for all $N$, which allows us to conclude that

\[
\mathbb E_U\norm{Z_{U,K}}^2
\le 4\,v_K\bigl(1+\log(d_1+d_2)\bigr),
\]
as claimed.
\end{proof}

\begin{theorem}[Search bound for Haar-random unitaries]
\label{thm:haar-basis-search}
For the oracle state search game associated with the family $\{\ket{\psi_{U,k}} = U\ket{k}\}_{k\in[K]}$ above, every one-query adversary with workspace dimension $M$ satisfies
\[
\mathbb E_U\bigl[\Win(\cA\mid U)\bigr]
\le \frac{4\bigl(1+\log(2M)\bigr)}{K}.
\]
\end{theorem}

\begin{proof}
Fix a one-query adversary $\cA=(V,\{\Pi_k\}_{k\in[N]})$. For each $k\in[N]$, the random state $\ket{\psi_{U,k}}$ is Haar-random in $\mathbb C^N$, and therefore
\[
\mathbb E_U\bigl[\ketbra{\psi_{U,k}}\bigr]=\frac{1}{N} \cdot \Id.
\]
By \cref{lem:generic-spectral-reduction,lem:subspace-specialization},
\[
\mathbb E_U\bigl[\Win(\cA\mid U)\bigr]
\le \mathbb E_U\norm{M_U}^2,
\qquad
M_U=\sum_{k \in [K], x\in[N]} \sqrt N\,U_{k,x} B_{k,x},
\]
where the matrices $\{B_{k,x}\}_{k,x}$ satisfy
\[
\sum_{k,x} B_{k,x}B_{k,x}^\dagger \preceq \frac{1}{K}\cdot \Id,
\qquad
\sum_{k,x} B_{k,x}^\dagger B_{k,x}\preceq \frac{1}{K} \cdot \Id.
\]
Hence the variance parameter in \cref{prop:haar-unitary-matrix-series} is at most $1/K$. Applying that proposition with $d_1=d_2=M$ yields
\[
\mathbb E_U\norm{M_U}^2
\le \frac{4\bigl(1+\log(2M)\bigr)}{K},
\]
which completes the proof.
\end{proof}
\section{Unitary search game is hardest on Haar random unitaries}
\label{app:haar-hardest}
This section describes a reduction from ($t$-query) oracle state search game hardness for Haar random unitaries to hardness for \emph{any} distribution over unitaries, under two notions of hardness of the oracle state search game. 

\begin{theorem}\label{thm:haar-is-the-hardest}
Fix the number of oracle queries $t\in\mathbb{N}$. For unitary $U_R$ depending on random variable $R$, consider the oracle state search game on the state family
\[
\{\ket{\psi_{R,k}} := U_R\ket{k}\}_{k\in [K]},
\]
where the adversary is given one copy of $\ket{\psi_{R,k}}$ for a random $k\in [K]$ and is asked to output $k$ after making $t$ oracle queries. Then the Haar-random unitary family is the hardest among all distribution $\{U_R\}$.

More specifically, if any distribution on $U_R$ is $(M,\epsilon)$-hard (respectively, $(M,\epsilon,\delta)$-hard) in the $t$-query search game, then so is the Haar distribution.
\end{theorem}

The proof idea is similar to the one in \cref{thm:f2hf1-t-case-search-in7}, which shows the one-query $(M,\epsilon)$-hardness of the $\{F_tH\cdots F_1H\ket{k}\}$ search game based on the corresponding $t=2$ case. Here we first use the same idea to prove the $(M,\epsilon)$-hardness of the Haar distribution in expectation in \cref{app:haar-expectation}, and then extend to the $(M,\epsilon,\delta)$-hardness in \cref{app:haar-tail}. 

\subsection{$(M,\epsilon)$-hardness}\label{app:haar-expectation}

\begin{proof} For any $t$-query oracle circuit $\mathcal{A}^{(\cdot)}$ for an oracle state search game, it can be specified by a set of unitaries $\{U_i\}_{i\in [t]}$ between queries and final measurement projectors $\{\Pi_{k}\}_{k\in[K]}$, such that with oracle access to $f$, $\mathcal{A}^f$ will output $k$ with probability
\[
\|\Pi_k \cdot O_f\cdot U_t\cdot \cdots\cdot O_f \cdot U_2\cdot O_f\cdot U_1\ket{\psi}\ket{0^{m-n}}\|^2.
\]
For notational convenience, we absorb the fixed ancilla initialization into the first operation, and define $V_1$ to be the isometry $U_1\ket{0^{m-n}}$ mapping $n$ qubits to $m$ qubits. Thus, for a Haar random unitary $U\in U(N)$, define $\ket{\psi_{U,k}}:=U\ket{k}$, and we can write the adversary's winning probability as
\[
\mathbb{E}_{U\in U(N)}\Win(\mathcal{A}\mid U)
=\mathbb{E}_{U\in U(N)}\max_f\left[\mathbb{E}_k\bra{\psi_{U,k}}V_1^\dagger O_f^\dagger\cdots U_t O_f^\dagger  \Pi_kO_fU_t\cdots O_fV_1\ket{\psi_{U,k}}\right].
\]
The maximum winning probability for Haar random unitary $U\in U(N)$ can then be upper bounded by the one for $U_R$ from any distribution of $R$:
\begin{align}
    &\sup_{\mathcal{A}}\left\{\mathbb{E}_{U\in U(N)}\Win(\mathcal{A}\mid U)\right\} \nonumber\\
    &\qquad=\sup_{\{U_t\},\{\Pi_k\}}
    \left\{\mathbb{E}_{U\in U(N)}\left[\max_f\mathbb{E}_k\bra{\psi_{U,k}}V_1^\dagger O_f^\dagger\cdots U_t O_f^\dagger  \Pi_kO_fU_t\cdots O_fV_1\ket{\psi_{U,k}}\right] \right\} \nonumber \\
    &\qquad=
    \sup_{\{U_t\},\{\Pi_k\}}
    \left\{\mathbb{E}_{U\in U(N),U_R}\left[\max_f\mathbb{E}_k\bra{\psi_{UU_R,k}}V_1^\dagger O_f^\dagger\cdots U_t O_f^\dagger  \Pi_kO_fU_t\cdots O_fV_1\ket{\psi_{UU_R,k}}\right] \right\} \nonumber 
    \\
    &\qquad\le
    \mathbb{E}_{U\in U(N)} \sup_{\{U_t\},\{\Pi_k\}}\left\{\mathbb{E}_{R}\left[\max_f \mathbb{E}_k\bra{\psi_{R,k}}U^\dagger V_1^\dagger O_f^\dagger\cdots U_t O_f^\dagger  \Pi_kO_fU_t\cdots O_fV_1U\ket{\psi_{R,k}}\right] \right\} \nonumber \\
    &\qquad=
    \mathbb{E}_{U\in U(N)} \sup_{\{U_t\},\{\Pi_k\}}\left\{\mathbb{E}_{R}\left[\max_f \mathbb{E}_k\bra{\psi_{R,k}}V_1^\dagger O_f^\dagger\cdots U_t O_f^\dagger  \Pi_kO_fU_t\cdots O_fV_1\ket{\psi_{R,k}}\right] \right\} \label{eq:change-adversary}\\
    &\qquad=
    \sup_{\mathcal{A}}\left\{
        \mathbb{E}_{R} \Win(\mathcal{A}\mid R)
    \right\}, \nonumber 
\end{align}
where the supremum is taken over all $\mathcal A$ acting on $m$ total qubits post-isometry. Notably, \cref{eq:change-adversary} holds because for any fixed unitary $U$ and adversary strategy $\mathcal A = (\{U_t\}, \{\Pi_k\})$ the success probability described in the previous expression is the $\{\ket{\psi_{R,k}}\}$ success probability of the adversary $\mathcal A' = ( \{U'_t\}, \{\Pi_k\})$, where $V'_1 = V_1 \cdot U$ and the rest of the strategy is unchanged. 
\end{proof}

\subsection{$(M,\epsilon,\delta)$-hardness}\label{app:haar-tail}
\begin{proof} The proof is almost identical to that of \cref{app:haar-expectation}. 
\begin{align*}
    &\sup_{\mathcal{A}}\left\{\Pr_{U\in U(N)}[\Win(\mathcal{A}\mid U)>\epsilon]\right\}\\
    &\qquad=\sup_{\{U_t\},\{\Pi_k\}}
    \left\{\Pr_{U\in U(N)}
    \left[
    \max_f \mathbb{E}_k\bra{\psi_{U,k}}V_1^\dagger O_f^\dagger\cdots U_t O_f^\dagger  \Pi_kO_fU_t\cdots O_fV_1\ket{\psi_{U,k}} >\epsilon
    \right]
    \right\}\\
    &\qquad=\sup_{\{U_t\},\{\Pi_k\}}
    \left\{\mathbb{E}_{U\in U(N)}\Pr_{R}
    \left[
    \max_f \mathbb{E}_k\bra{\psi_{UU_R,k}}V_1^\dagger O_f^\dagger\cdots U_t O_f^\dagger  \Pi_kO_fU_t\cdots O_fV_1\ket{\psi_{UU_R,k}} >\epsilon
    \right]
    \right\}\\
    &\qquad\le \mathbb{E}_{U\in U(N)}
    \sup_{\{U_t\},\{\Pi_k\}}
    \left\{\Pr_{R}
    \left[
    \max_f\mathbb{E}_k\bra{\psi_{R,k}}U^\dagger V_1^\dagger O_f^\dagger\cdots U_t O_f^\dagger  \Pi_kO_fU_t\cdots O_fV_1U\ket{\psi_{R,k}} >\epsilon
    \right]
    \right\}\\
    &\qquad= \mathbb{E}_{U\in U(N)}
    \sup_{\{U_t\},\{\Pi_k\}}
    \left\{\Pr_{R}
    \left[
    \max_f\mathbb{E}_k\bra{\psi_{R,k}} V_1^\dagger O_f^\dagger\cdots U_t O_f^\dagger  \Pi_kO_fU_t\cdots O_fV_1\ket{\psi_{R,k}} >\epsilon
    \right]
    \right\}\\
    &\qquad= \sup_{\mathcal{A}}\left\{
        \Pr_{R}[\mathrm{Win}(\mathcal{A}\mid R)]>\epsilon
    \right\}. \qedhere
\end{align*}    
\end{proof}

\section{Hardness of the oracle Choi state game}\label{sec:choi-hardness}

This appendix presents proofs of one-query hardness of the oracle Choi state game. Two of the results (\cref{thm:permutation-epr,cor:fhf-epr}) rederive theorems that we proved using the oracle state search game in the body of the paper, while \cref{thm:fuf-epr} extends \cref{thm:f2hf1h-search-main} to analogous classes of unitaries where the Hadamard unitary $H^{\otimes n}$ has been replaced by a fairly general one-qubit unitary $U_0$. 

\subsection{Theorem statements}
\begin{theorem}[Permutation family Choi bound, see \cref{thm:perm-search-main}]\label{thm:permutation-epr}

Let $P$ be the in-place permutation unitary associated with a uniformly random permutation $\pi$ on $\{0,1\}^n$, $P = P_\pi:\ket{x}\mapsto\ket{\pi(x)}$. For the oracle Choi state game for unitary family $\{P_\pi \}_{\pi}$, every one-query adversary with workspace dimension $M$ satisfies
\[
\mathbb{E}_{\pi}\bigl[\Win(\cA\mid \pi)\bigr]
=O\!\left(\frac{\log^2 M\cdot\log^2 N}{N}\right).
\]
\end{theorem}

\begin{theorem}[$F_2U_0F_1$ Choi bound]\label{thm:fuf-epr}

Let $f_1,f_2:\{0,1\}^n\to\{0,1\}$ be uniformly random Boolean functions, and let $F_j=\sum_x (-1)^{f_j(x)}\ketbra{x}$ for $j\in\{1,2\}$. $U_0:\mathbb{C}^N\to \mathbb{C}^N$ is a fixed unitary that is independent of $F_1,F_2$. For the oracle Choi state game for unitary family $\{U_{f_1,f_2}:=F_2U_0F_1\}_{f_1,f_2}$,  every one-query adversary with workspace dimension $M$ satisfies
\[
\mathbb{E}_{f_1,f_2}\bigl[\Win(\cA\mid f_1,f_2)\bigr]
=O\!\left(\frac{b^2\cdot \log M \cdot \log (MN)}{N}\right),
\]
for $b:=\sqrt{N}\cdot\max_{k,x\in[N]}|\left<x|U_0|k\right>|$.
\end{theorem}

\begin{corollary}[$F_2HF_1$ Choi bound, see \cref{thm:f2hf1h-search-main}]\label{cor:fhf-epr}

Let $f_1,f_2:\{0,1\}^n\to\{0,1\}$ be uniformly random Boolean functions, and let $F_j=\sum_x (-1)^{f_j(x)}\ketbra{x}$ for $j\in\{1,2\}$. For the oracle Choi state game for unitary family $\{F_2HF_1\}_{f_1,f_2}$, every one-query adversary with workspace dimension $M$ satisfies
\[
\mathbb{E}_{f_1,f_2}\bigl[\Win(\cA\mid f_1,f_2)\bigr]
=O\!\left(\frac{\log M \cdot \log (MN)}{N}\right).
\]
\end{corollary}

\subsection{General setup}

In this section, we generically reduce the problem of upper bounding $\mathbb{E}_R[\Win(\cA\mid R)]$ to the calculation of the expected (squared) spectral norm of a random matrix. This is similar to the search reduction as for the oracle state search game. While the search relaxation in \Cref{lem:generic-spectral-reduction} applies $V\ket{\psi_{R,k}} = D_{V,R,k}\ket{\wt_{V,\mathbf R,k}}$ from \cref{lem:weight-vector-decomposition}, viewing $\ket{\wt_{V,\mathbf R,k}}$ as an average over $\{\ket{\psi_{R,k}}\}_R$, here in the Choi state game, the algorithm will always receive a maximally mixed state on its register $A$, so we will set $\ket{\wt_{V,\mathbf R,k}}=\ket{\wt_V}$ independent of $k, \mathbf R$. 

\begin{lemma}[Generic spectral relaxation]\label{lem:spectral-relax-epr}
Let $\cA=(V,U)$ be a one-query adversary for the oracle Choi state game, and let 
\[
D_{V,R,k}:=
\sum_{i\in[M]:\,p_i>0}
\frac{\braket{v_i}{\psi_{R,k}}}{\sqrt{p_i}}\ketbra{i},
\]
be the diagonal matrix
similar to the one from \cref{lem:weight-vector-decomposition},  while $p_i=\frac1N\cdot\braket{v_i}{v_i}$, defined from the distribution on a Haar random state. Define
\[
M_R:=\frac{1}{\sqrt{N}}\sum_{k\in\{0,1\}^n} \Pi_\EPR \left(\ket{k}\otimes  U D_{V,R,k}\right).
\]
Then, for every $R$,
\[
\Win(\cA\mid R)\le \norm{M_R}^2.
\]
Consequently,
\[
\underset{R}{\mathbb E}\bigl[\Win(\cA\mid R)\bigr]\le \underset{R}{\mathbb E}\norm{M_R}^2.
\]
\end{lemma}

\begin{proof}
Fix $R$. For $p_i=\frac1N\cdot\braket{v_i}{v_i}$ modified from \cref{lem:weight-vector-decomposition}, we see that $\ket{\wt_{V,\mathbf R,k}} = \ket{\wt_{V}}$ is independent of $k$. Hence
\begin{align*}
\Win(\cA\mid R)
&=\max_f \left\|
    \Pi_{\mathrm{EPR}}\frac{1}{\sqrt{N}}
    \left(\sum_k\ket{k}\otimes U\cdot O_f\cdot V\ket{\psi_{R_k}}
    \right)
    \right\|^2
\\
&=
\max_f \left\|
    \Pi_{\mathrm{EPR}}\frac{1}{\sqrt{N}}
    \left(\sum_k\ket{k}\otimes U\cdot D_{V,R,k}\cdot O_f\ket{\wt_V}
    \right)
    \right\|^2
\\
&\le 
\left\|
\Pi_{\mathrm{EPR}}\frac{1}{\sqrt{N}}
\left(\sum_k\ket{k}\otimes U\cdot D_{V,R,k}\right)
\right\|^2
=\norm{M_R}^2.
\end{align*}
Averaging over $R$ gives the final inequality.
\end{proof}

\subsection{Permutation lower bound in the oracle Choi state game}

\paragraph{Choi formulation.} Here we consider the oracle Choi state game with unitary family $\{P\}_{\pi\in S_N}$ and the EPR state $\ket{\Psi_{\EPR}}\propto\sum_{k\in \{0,1\}^n}\ket{kk}$. We apply the general spectral reduction in \cref{lem:spectral-relax-epr}. 
Now our goal is to upper bound $\mathbb{E}_\pi \|M_\pi\|^2$, where
\[
M_\pi=
\Pi_{\mathrm{EPR}}
(I\otimes U)
\left[
\frac{1}{\sqrt{N}}\sum_k \ket{k}\otimes 
\left(
    \sum_i
    \frac{\left<v_i|\pi(k)\right>}{\sqrt{p_i}}\ketbra{i}
\right)
\right]
.
\]

\paragraph{$M_R$ as a combinatorial matrix sum.}
Define $\widehat{A}_{k,x}$ and $A_{k,x}$ as
\begin{align*}
    A_{k,x}&=
    \Pi_{\mathrm{EPR}}
    (I\otimes U)
\left[
\frac{1}{\sqrt{N}} \ket{k}\otimes 
\left(
    \sum_i
    \frac{\left<v_i|x\right>}{\sqrt{p_i}}\ketbra{i}
\right)
\right],
\\
    \widehat{A}_{k,x}&
    =\begin{pmatrix}
        0 & A_{k,x}\\
        A^\dagger_{k,x} & 0
    \end{pmatrix},
\end{align*}
and define $B_{k,x}$ as the shifted Hermitian version
\[
B_{k,x}=\widehat{A}_{k,x}-\frac{1}{N^2}\left(\sum_{k,x}\widehat{A}_{k,x}\right).
\]
Now $\|M_\pi\|^2=\left\|\sum_{k}A_{k,\pi(k)}\right\|^2
=\left\|\sum_{k}\widehat{A}_{k,\pi(k)}\right\|^2$. Now our goal is to bound $\left\|\sum_{k}\widehat{A}_{k,\pi(k)}\right\|^2$, by bounding its shifted version $\left\|\sum_{k}B_{k,\pi(k)}\right\|^2$.

\paragraph{Parameter estimates for $B_{k,x}$.} To apply  \cref{thm:matrix-bernstein-permutation} on $\sum_{k}B_{k,\pi(k)}$, we have the following parameter estimates. 
\begin{enumerate}
    \item \textbf{Sum to 0.} 
    
    As $\{B_{k,x}\}$ is the shifted version, $\sum_{k,x}B_{k,x}=\sum_{k,x}(\widehat{A}_{k,x}-\widehat{A}_{k,x})=0$.

    \item \textbf{Bounded norm.}

    The operator norm of $B_{k,x}$ can be bounded by some basic properties of $A_{k,x}$.
\begin{align*}
    \|A_{k,x}\|^2
    &=\|A_{k,x}A_{k,x}^\dagger\|
    =
    \frac{1}{N}
    \left\|
    \Pi_{\mathrm{EPR}}
    (I\otimes U)
    \left[
    \ketbra{k}
    \otimes
    \left(
        \sum_i
        \frac{
        \left<v_i|x\right>
        \left<x|v_i\right>
        }{p_i}
        \ketbra{i}
    \right)
    \right]
    (I\otimes U^\dagger)
    \Pi_{\mathrm{EPR}}
    \right\|
    \\
    &\le 
    \left\|
    \Pi_{\mathrm{EPR}}
    \cdot
    (\ketbra{k}\otimes I)
    \cdot
    \Pi_{\mathrm{EPR}}
    \right\|
    =\frac{1}{N},
    \\
    \left\|
    \sum_{k,x}A_{k,x}
    \right\|^2
    &=
    N\cdot \left\|
    \Pi_{\mathrm{EPR}}
    (I\otimes U)
    \left[
    \ket{+}_{k}\otimes 
    \sum_i
    \frac{\left<v_i|+\right>_x}{\sqrt{p_i}}
    \ketbra{i}
    \right]
    \right\|^2
    \\
    &\le 
    N^2\cdot
    \left\|
    \Pi_{\mathrm{EPR}}
    \cdot
    (\ketbra{+}_k\otimes I)
    \cdot
    \Pi_{\mathrm{EPR}}
    \right\|
    = N.
\end{align*}
With these, we can bound the operator norm of $B_{k,x}$ as follows:
\begin{align}
    \|B_{k,x}\|
    &\le 
    \|\widehat{A}_{k,x}\|
    +
    \frac{1}{N^2}\left\|\sum_{k,x}\widehat{A}_{k,x}\right\|
    =
    \|A_{k,x}\|
    +
    \frac{1}{N^2}\left\|\sum_{k,x}A_{k,x}\right\|
    \le
    \frac{1}{\sqrt{N}}+\frac{1}{N^2}\sqrt{N}
    \le \frac{2}{\sqrt{N}}.
    \label{ineq:epr-permutation-bounded-norm}
\end{align}

    \item \textbf{Bounded variance.}

    The variance of $B_{k,x}$ can be bounded by the following properties of $A_{k,x}$.
\begin{align*}
    \left\|
    \sum_{k,x}A_{k,x} A_{k,x}^\dagger
    \right\|
    &= \frac{N}{N}
    \left\|
    \Pi_{\mathrm{EPR}}
    \cdot
    (I\otimes I)
    \cdot
    \Pi_{\mathrm{EPR}}
    \right\|
    =1,
    \\
    \left\|
    \sum_{k,x}A_{k,x}^\dagger A_{k,x}
    \right\|
    &= \frac{1}{N}
    \left\|
    \sum_{k,x}\sum_{i,j}
    \ketbra{j}
    \frac{\left<x|v_j\right>}{\sqrt{p_j}}
    \bra{k}
    (I\otimes U^\dagger) 
    \Pi_{\mathrm{EPR}}
    (I\otimes U)
    \ket{k}
    \frac{\left<v_i|x\right>}{\sqrt{p_i}}
    \ketbra{i}
    \right\|
    =\frac{1}{N}.
\end{align*}
With these, we can bound the variance of $B_{k,x}$ as follows:
\begin{align}
    \sigma^2
    &=\frac{1}{N}
    \left\|
    \sum_{k,x}B_{k,x}^2
    \right\|
    =
    \frac{1}{N}\left\|
    \sum_{k,x}\widehat{A}_{k,x}^2
    -\frac{1}{N^2}\left(\sum_{k,x}\widehat{A}_{k,x}\right)^2
    \right\|
    \notag\\
    &\le 
    \frac{1}{N}\left\|
    \sum_{k,x}\widehat{A}_{k,x}^2
    \right\|
    +
    \frac{1}{N^3}\left\|
    \sum_{k,x}\widehat{A}_{k,x}
    \right\|^2
    \notag
    \\
    &=
    \frac1N\cdot \max\left\{
    \left\|
    \sum_{k,x}A_{k,x} A_{k,x}^\dagger
    \right\|,
    \left\|
    \sum_{k,x}A_{k,x}^\dagger A_{k,x}
    \right\|
    \right\}
    +
    \frac{1}{N^3}\cdot \left\|
    \sum_{k,x}A_{k,x}
    \right\|^2
    \notag
    \\
    &\le
    \frac{1}{N}+\frac{1}{N^3}\cdot N\le \frac{2}{N}.
    \label{ineq:epr-permutation-variance}
\end{align}
\end{enumerate}

\paragraph{Upper bounding the Choi state game winning probability.}
Apply \cref{thm:matrix-bernstein-permutation} to the Hermitian matrix
\[
X:=\sum_k B_{k,\pi(k)}.
\]
Using \cref{ineq:epr-permutation-bounded-norm} and \cref{ineq:epr-permutation-variance}, we obtain the
\begin{align*}
    \Pr\left[
        \lambda_{\max}(X)\ge t
    \right]
    \le 2M\cdot \exp\left(-
        \frac{t^2}{24/N+8\sqrt{2}t/\sqrt{N}}
    \right).
\end{align*}
Choose
\[
t:= C\,\frac{\log M\,\log N}{\sqrt{N}}
\qquad\text{for a sufficiently large universal constant }C.
\]
Then
\[
\Pr\!\left[\lambda_{\max}(X) \ge C\frac{\log M\,\log N}{\sqrt{N}}\right]
\le \frac{1}{N}.
\]

Combining this with $\|M_\pi\|^2=\left\|\sum_{k}A_{k,\pi(k)}\right\|^2= \left\|\sum_{k}\widehat{A}_{k,\pi(k)}\right\|^2$,
\begin{align}
    \Pr\left[
        \left\|M_\pi\right\|\ge t'
    \right]
    &=
    \Pr\left[
        \left\|\sum_{k}\widehat{A}_{k,\pi(k)}\right\|\ge t'
    \right]
    \le 
    \Pr\left[
        \left\|\sum_{k}\widehat{A}_{k,\pi(k)}\right\|\ge t'
        +
        \frac{N}{N^2}\left(\left\|\sum_{k,x}\widehat{A}_{k,x}\right\|-\sqrt{N}\right)
    \right]
    \notag
    \\
    &\le 
    \Pr\left[
        \left\|\sum_{k}\widehat{A}_{k,\pi(k)}-N\cdot \frac{1}{N^2}\sum_{k,x}\widehat{A}_{k,x}\right\|\ge t' -\frac{1}{\sqrt{N}}
    \right]
    \notag
    \\
    &=
    \Pr\left[
        \lambda_{\max}(X)\ge t' -\frac{1}{\sqrt{N}}
    \right]
    \le \frac{1}{N}\label{ineq:epr-permutation-relation-between-shifted}
\end{align}
by letting $t'=C\cdot \log M\log N/\sqrt{N} +1/\sqrt{N}$. 

From this tail bound, we can conclude that,
\begin{align*}
    \mathbb{E}_\pi \|M_\pi\|^2 
    &\le  
    \Pr[\|M_\pi\|\ge t']\cdot 1 + \Pr[\|M_\pi\|<t']\cdot (t')^2
    \\
    &\le 
    \frac{1}{N}+1\cdot C'\cdot 
    \left(
    \frac{\log^2 M \cdot \log^2 N}{N}+\frac{1}{N}
    \right)
    \\
    &=O\left(\frac{\log^2 M\cdot \log^2 N}{N}\right)
    .
\end{align*}
By \cref{lem:spectral-relax-epr}, this is also an upper bound on the average Choi state game winning probability, and thus proves \cref{thm:permutation-epr}.

\paragraph{Extending to game with classical advice.}
The same tail bound also derives the classical-advice lower bound. For a fixed advice string, the spectral reduction from  \cref{lem:spectral-relax-epr} suggests that constant winning probability requires $\norm{M_\pi}^2=\Omega(1)$. By setting $t'-\frac{1}{\sqrt{N}}=c$ for some $c=\Omega(1)$ as in \cref{ineq:epr-permutation-relation-between-shifted}, this occurs with probability bounded by
\[
\Pr\bigl[\lambda_{\max}(X)\ge c\bigr]
\le 2M\cdot \exp\!\left(
-\frac{c^2}{24/N + 8\sqrt{2}\,c/\sqrt{N}}
\right)=2M\cdot \exp(-\Omega(\sqrt{N})),
\]
which is exponentially small if $\log M\ll \sqrt{N}$. A union bound over all $2^S$ advice strings then yields the lower bound
\[
S + \log(M)=\Omega(\sqrt{N})
\]
in order to achieve constant win probability.

\subsection{$F_2U_0F_1$ lower bound in the oracle Choi state game}

\paragraph{Choi formulation.} Here we consider the oracle Choi state game with unitary family $\{F_2U_0F_1\}_{f_1,f_2:\{0,1\}^n\to\{0,1\}}$. We apply the general spectral reduction in \cref{lem:spectral-relax-epr}. Now our goal is to upper bound $\mathbb{E}_{f_1,f_2}\|M_{f_1,f_2}\|^2$, where
\begin{align*}
    \mathbb{E}_{f_1,f_2} \|M_{f_1,f_2}\|^2
    &=
    \mathbb{E}_{f_1,f_2}\left\|
    \Pi_{\mathrm{EPR}}(I\otimes U)\cdot \frac{1}{\sqrt{N}}\sum_k \ket{k}\otimes D_{f_1,f_2,k}
    \right\|^2,
    \\
    \text{for}\quad
    D_{f_1,f_2,k}
    &:= \frac{1}{\sqrt{N}}
    \sum_{x,i} \frac{\left<v_i|x\right>}{\sqrt{p_i}}\cdot R(k,x)\cdot \ketbra{i},\\
    \text{and}\quad
    R(k,x)&:= \sqrt{N}\cdot(-1)^{f_1(k)+f_2(x)}\cdot \left<x|U_0|k\right>.
\end{align*}

\paragraph{Conditioning on $f_2$, a matrix Rademacher series from $f_1$.}
For every fixed $f_2$, we can write $M_R:=M_{f_1,f_2}$ as a matrix Rademacher series in terms of $(-1)^{f_1(k)}$. Therefore, \cref{thm:tropp-rademacher} implies
\begin{align*}
    \mathbb{E}_{f_1,f_2} \|M_{f_1,f_2}\|^2 
    &\le O(\log M)\cdot \mathbb{E}_{f_2}\mathrm{Var}_{f_1}(M_{f_1,f_2})
    \\
    &\le O(\log M)
    \cdot \left(
    \mathbb{E}_{f_2} \left\|\mathbb{E}_{f_1}M_{f_1,f_2}M_{f_1,f_2}^\dagger\right\|
    +
    \mathbb{E}_{f_2} \left\|\mathbb{E}_{f_1}M_{f_1,f_2}^\dagger M_{f_1,f_2}\right\|
    \right).
\end{align*}

Note that for any $x_1,x_2,f_2$, if $k_1\neq k_2$,
\begin{align*}
    \mathbb{E}_{f_1}[R(k_1,x_1)R(k_2,x_2)]=0,
\end{align*}
and thus
\begin{align*}
    \|\mathbb{E}_{f_1} M_{f_1,f_2}^\dagger M_{f_1,f_2}\|
    &=
    \frac{1}{N^2}
    \left\|\mathbb{E}_{f_1} \sum_{k,k'}D^\dagger _{f_1,f_2,k}U^\dagger \ket{k}\bra{k'}UD_{f_1,f_2,k'}\right\|
    \notag\\
    &\le\frac{1}{N^2}\left\|\mathbb{E}_{f_1}\sum_k D_{f_1,f_2,k}^\dagger D_{f_1,f_2,k}\right\|
    \le \frac{1}{N}\max_{k\in \{0,1\}^n}\left\|
    \mathbb{E}_{f_1} D^\dagger_{f_1,f_2,k} D_{f_1,f_2,k}
    \right\|
    \\
    \|\mathbb{E}_{f_1}M_{f_1,f_2} M_{f_1,f_2}^\dagger\| 
    &= \frac{1}{N}\left\|\mathbb{E}_{f_1}
    \Pi_{\mathrm{EPR}}
    (I\otimes U)
    \left(
    \sum_{k,k'}\ket{k}\bra{k'}
    \otimes D_{f_1,f_2,k} D_{f_1,f_2,k'}^\dagger 
    \right)
    (I\otimes U^\dagger)
    \Pi_{\mathrm{EPR}}
    \right\|\notag\\
    &\le \frac{1}{N}
    \left\|\mathbb{E}_{f_1} \sum_k \ketbra{k}\otimes D_{f_1,f_2,k} D_{f_1,f_2,k}^\dagger\right\|
    =\frac{1}{N}\max_{k\in \{0,1\}^n} 
    \left\|\mathbb{E}_{f_1} D_{f_1,f_2,k} D_{f_1,f_2,k}^\dagger\right\|.
\end{align*}
For diagonal matrix $D_{f_1,f_2,k}$, $D_{f_1,f_2,k}D^\dagger_{f_1,f_2,k}$ is independent of $f_1$:
\begin{align}
    D_{f_1,f_2,k} D^\dagger_{f_1,f_2,k} &=D_{f_1,f_2,k}^\dagger D_{f_1,f_2,k}
    =\sum_i \sum_{x,x'}\frac{\left<v_i|x\right>\left<x'|v_i\right>}{p_i}\ketbra{i}\cdot (-1)^{f_2(x)+f_2(x')}\cdot \left<x|U_0|k\right>\cdot \bra{k}U_0^\dagger \ket{x'}
    \notag
    \\
    \left\|D_{f_1,f_2,k} D^\dagger_{f_1,f_2,k}\right\|&=\max_{i\in [M]}
    \left|
    \sum_x (-1)^{f_2(x)}\frac{\left<v_i|x\right>\cdot \left<x|U_0|k\right>}{\sqrt{p_i}}
    \right|^2.
    \label{eq:epr-DDd}
\end{align}

\paragraph{Another Rademacher series from $f_2$.} From \cref{eq:epr-DDd}, inside $\max$ it can be viewed as a Rademacher series from $f_2$, and therefore \cref{thm:tropp-rademacher} implies

\begin{align*}
    \Pr_{f_2}\left[
    \left|
    \sum_x (-1)^{f_2(x)}\cdot \left<x|U_0|k\right>\cdot\frac{\left<v_i|x\right>}{\sqrt{p_i}}
    \right|>t 
    \right]
    \le 2\cdot\exp\left(-\frac{t^2}{\sum_x \frac{|\left<v_i|x\right>|^2|\left<x|U_0|k\right>|^2}{p_i}}\right)\le 2\cdot \exp\left(-\frac{t^2}{N\cdot \max_{k,x}\left|\left<x|U_0|k\right>\right|^2}\right).
\end{align*}
Now the right hand side is independent of $k$. By union bound over all $i\in [M]$ and $k\in \{0,1\}^n$, 
\begin{align}
    \Pr_{f_2}\left[\left\|\mathbb{E}_{f_1} M_{f_1,f_2}^\dagger M_{f_1,f_2}\right\| >\frac{t^2}{N}\right]\le 2MN\cdot \exp\left(-\frac{t^2}{N\cdot\max_{k,x}\left|\left<x|U_0|k\right>\right|^2}\right),\label{ineq:epr-tail-mdm}\\
    \Pr_{f_2}\left[\left\|\mathbb{E}_{f_1} M_{f_1,f_2} M_{f_1,f_2}^\dagger\right\| >\frac{t^2}{N}\right]\le 2MN\cdot \exp\left(-\frac{t^2}{N\cdot \max_{k,x}\left|\left<x|U_0|k\right>\right|^2}\right).\label{ineq:epr-tail-mmd}
\end{align}
Denote $b:=\sqrt{N}\cdot \max_{k,x}\left|\left<x|U_0|k\right>\right|$. This will give
\begin{align*}
    \mathbb{E}_{f_1,f_2}\left\|M_{f_1,f_2}\right\|^2=O\left(\frac{b^2\cdot \log M\cdot \log MN}{N}\right)
\end{align*}
and proves \cref{thm:fuf-epr}. 

Specifically, for $U_0=H$, $\max_{k,x}\left|\left<x|U_0|k\right>\right|=1/\sqrt{N}$ and thus $b=1$, and this proves \cref{cor:fhf-epr}.

\paragraph{Extending to game with classical advice.} The tail bound in \cref{thm:tropp-rademacher} also implies a classical-advice lower bound. 
For a fixed advice string, the spectral reduction from \cref{lem:spectral-relax-epr} suggests that constant winning probability would require $\norm{M_R}^2=\Omega(1)$. For fixed $f_2$, by setting $t=c$ for some constant $c=\Omega(1)$, \cref{thm:tropp-rademacher} implies 
\[
    \Pr_{f_1}[\|M_R\|\ge c]\le 2M\cdot \exp\!\left(
    -\frac{c^2}{2\cdot \mathrm{Var}_{f_1}(M_R)}
    \right).
\]
Therefore, for parameter $c'$ to be defined later,
\begin{align*}
    \Pr_{f_1,f_2}\!\left[\|M_R\|\ge c\right]
    &\le
    \Pr_{f_2}\!\left[\Var_{f_1}(M_R)\ge c'\right]
    +
    \Pr_{f_1,f_2}\!\left[\Var_{f_1}(M_R)<c' \;\land\; \|M_R\|\ge c\right]\\
    &\le \Pr_{f_2}\!\left[\Var_{f_1}(M_R)\ge c'\right]
    +
    2M\cdot \exp\!\left(
    -\frac{c^2}{2\cdot c'}
    \right).
\end{align*}
From the definition of $\Var(M_R)$ and \cref{ineq:epr-tail-mdm}, \cref{ineq:epr-tail-mmd}, 
\begin{align*}
    \Pr_{f_2}[\Var_{f_1}(M_R)\ge c']
    &\le          
    \Pr_{f_2}\left[ \left\|\mathbb{E}_{f_1}M_{f_1,f_2}M_{f_1,f_2}^\dagger\right\|
    +
    \left\|\mathbb{E}_{f_1}M_{f_1,f_2}^\dagger M_{f_1,f_2}\right\|\ge c'
    \right]
    \\
    &\le  \Pr_{f_2}\left[ \left\|\mathbb{E}_{f_1}M_{f_1,f_2}M_{f_1,f_2}^\dagger\right\|
    \ge \frac{c'}{2}
    \right]
    +
    \Pr_{f_2}\left[ \left\|\mathbb{E}_{f_1}M_{f_1,f_2}^\dagger M_{f_1,f_2}\right\|
    \ge \frac{c'}{2}
    \right]
    \\
    &\le 4MN\cdot \exp\left(-\frac{c' N}{2N\cdot\max_{k,x}\left|\left<x|U_0|k\right>\right|^2}\right).
\end{align*}
Therefore, by setting $c'=1/\sqrt{N}$, we can bound the probability for $c=\Omega(1)$ by the following:
\begin{align*}
    \Pr_{f_1,f_2}\!\left[\|M_R\|\ge c\right]
    &\le \Pr_{f_2}\!\left[\Var_{f_1}(M_R)\ge c'\right]
    +
    M\cdot \exp\!\left(
    -\frac{c^2}{2\cdot c'}
    \right)
    \\
    &\le 
    4MN\cdot \exp\left(-\frac{\sqrt{N}}{2N\cdot\max_{k,x}\left|\left<x|U_0|k\right>\right|^2}\right)
    +
    2M\cdot \exp\!\left(
    -\frac{c^2\sqrt{N}}{2}
    \right)
\end{align*}
For $b:=\sqrt{N}\cdot \max_{k,x}\left|\left<x|U_0|k\right>\right|\ge 1$, a union bound over all $2^S$ advice strings will yield a lower bound 
\[
S+\log MN=\Omega(\sqrt{N}/b^2)
\]
in order to achieve constant win probability.

\end{document}